%% file: paper.tex
\PassOptionsToPackage{table}{xcolor}

\documentclass[conference]{IEEEtran}

\usepackage{booktabs}
\usepackage[noline,ruled,noend,linesnumbered]{algorithm2e}
\usepackage{amsmath,amssymb,amsthm}
\usepackage[english]{babel} %
\usepackage{caption,subcaption} %
\usepackage{cite} %
\usepackage{enumitem} %
\usepackage{gensymb}
\usepackage{graphicx}
\PassOptionsToPackage{hyphens}{url}
\usepackage{hyperref} %
\usepackage[htt]{hyphenat} %
\usepackage{microtype} %
\usepackage{multirow} %
\usepackage{placeins} %
\usepackage[normalem]{ulem}
\usepackage[redefsymbols=false]{siunitx} %
\usepackage{pgfplots} %
\usepackage[american,siunitx]{circuitikz} %
\usepackage[T1]{fontenc} %
\usepackage{authblk}

\usepackage{collcell}
\usepackage{hhline}
\def\colorModel{hsb} %
\newcommand\ColCell[1]{
  \pgfmathparse{#1>=0?1:0}  %
    \ifnum\pgfmathresult=0\relax\color{white}\fi
  \pgfmathsetmacro\compA{0}      %
  \pgfmathsetmacro\compB{0} 	  %
  \pgfmathsetmacro\compC{#1/800+.875}      %
  \edef\x{\noexpand\centering\noexpand\cellcolor[\colorModel]{\compA,\compB,\compC}}\x #1
  } 
\newcolumntype{E}{>{\collectcell\ColCell}m{0.4cm}<{\endcollectcell}}  %
\newcommand*\rot{\rotatebox{90}}

\usepackage{tikzsymbols, tikzpeople}
\pgfplotsset{width=8cm,compat=1.8} %
\usetikzlibrary{patterns,shapes,3d,decorations.markings,arrows}	%
\usepgfplotslibrary{groupplots}
\pgfdeclarelayer{bg}
\pgfdeclarelayer{fg}
\pgfsetlayers{bg,main,fg}

\newcommand{\E}{\mathbb{E}}
\newcommand{\I}{\mathcal{I}}
\newcommand{\LL}{\mathcal{L}_\text{adv}}
\newcommand{\bmmu}{\boldsymbol{\mu}}
\newcommand{\bmdelta}{\boldsymbol{\Delta}}

\DeclareMathOperator*{\argmin}{arg\,min~~}
\DeclareMathOperator*{\subjto}{subject\,to~~}

\newtheorem{theorem}{Theorem}[section]
\newtheorem{lemma}[theorem]{Lemma}

\begin{document}

\input{./secs/titleabs}

\input{./secs/introduction}

\input{./secs/threat_model}
\input{./secs/background}
\input{./secs/naive_attack}
\input{./secs/adv_attack}

\input{./secs/detailed_adv.tex}

\input{./secs/experiments}

\input{./secs/discussion}

\input{./secs/related_work}

\input{./secs/conclusion}

\appendices
\input{./secs/appendix}

\newpage
\bibliographystyle{unsrt}
\bibliography{adv}

\end{document}

%% file: secs/titleabs.tex
\title{GhostImage: Remote Perception Attacks against  Camera-based Image
Classification Systems}


%
%
%
%

\author[1]{Yanmao Man}
\author[1]{Ming Li}
\author[2]{Ryan Gerdes}
\affil[1]{University of Arizona}
\affil[2]{Virginia Tech}

\maketitle

\begin{abstract}
In vision-based object classification systems imaging sensors perceive the
environment and machine learning is then used to detect and classify
objects for decision-making purposes; e.g., to maneuver an automated
vehicle around an obstacle or to raise an alarm to indicate the presence of
an intruder in surveillance settings. In this work we demonstrate how the
perception domain can be remotely and unobtrusively exploited to enable an
attacker to create spurious objects or alter an existing object.  An
automated system relying on a detection/classification framework subject to
our attack could be made to undertake actions with catastrophic results due
to attacker-induced misperception.

We focus on camera-based systems and show that it is possible to remotely
project adversarial patterns into camera systems by exploiting two common
effects in optical imaging systems, viz., lens flare/ghost effects and
auto-exposure control. To improve the robustness of the attack to channel
effects, we generate optimal patterns by integrating adversarial machine
learning techniques with a trained end-to-end channel model. We
experimentally demonstrate our attacks using a low-cost projector, on three
different image datasets, in indoor and outdoor environments, and with
three different cameras.  Experimental results show that, depending on the
projector-camera distance, attack success rates can reach as high as 100\%
and under targeted conditions.
\end{abstract}


%% file: secs/introduction.tex
\section{Introduction}

Object detection and classification have been widely adopted in autonomous
systems, such as automated vehicles \cite{uber, waymo, autopilot} and unmanned
aerial vehicles \cite{skydio, primeair}, as well as surveillance systems, e.g.,
smart home monitoring systems \cite{nest-systems, ring-systems}.  These systems
first perceive the surrounding environment via sensors (e.g., cameras, LiDARs,
and motion sensors) that convert analog signals into digital data, then try to
understand the environment using object detectors and classifiers (e.g.,
recognizing traffic signs or unauthorized persons), and finally make a decision
on how to influence/interact with the environment (e.g., a vehicle may
decelerate or a surveillance system raises an alarm). %

While the cyber (digital) attack surface of such systems have been widely studied
\cite{checkoway2011comprehensive,miller2014survey,petit2014potential,costin2014large},
vulnerabilities in the perception domain are less well-known, despite perception being
the first and critical step in the decision-making pipeline. That is, if sensors can be
compromised then false data can be injected and the decision making process will indubitably be harmed as the system is not acting on an accurate view of its environment. Recent work has demonstrated false data injection against sensors in a remote manner via either electromagnetic (radio frequency)
interference~\cite{selvaraj2018electromagnetic}, laser pulses (against
microphones~\cite{sugawaralight}, or LiDARs~\cite{shin2017illusion,
petit2015remote, cao2019adversarial}), and acoustic waves \cite{son2015rocking,
yan2020surfingattack}. These perception domain sensor attacks alter the
data at the source, hence bypassing traditional digital defenses (such as
crypto-based authentication or access control), and are subsequently much harder to defend
against \cite{yan2020minimalist, giechaskiel2020taxonomy}. These attacks can also be remote in that the attacker needn't physically contact/access/modify devices or objects.

\begin{figure}[tb]
	\centering
	\begin{subfigure}[b]{.45\columnwidth}
		\centering
		\includegraphics[width=.8\columnwidth]{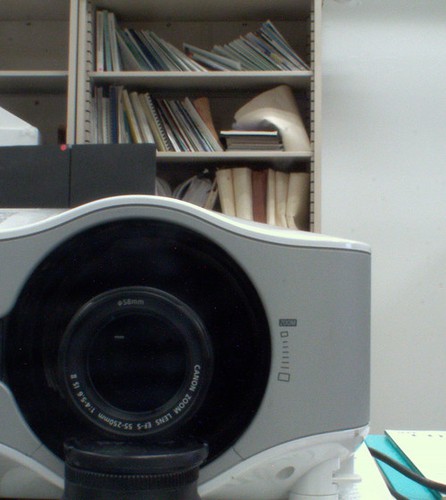}
		\caption{Projector off}
	\end{subfigure}%
	\begin{subfigure}[b]{.45\columnwidth}
		\centering
		\includegraphics[width=.8\columnwidth]{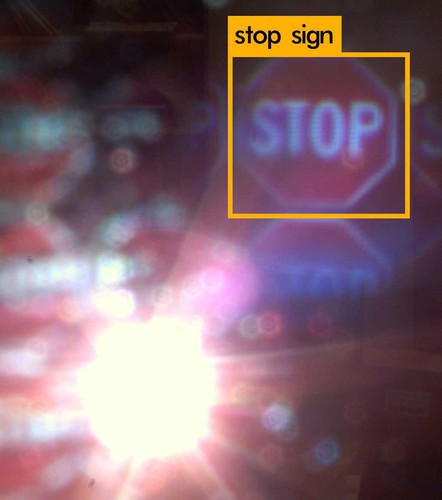}
		\caption{Projector on}
	\end{subfigure}

	\caption{A STOP sign image was injected into a camera by a projector, which
	was detected by YOLOv3~\cite{redmon2016you}.}
	\label{fig:yolov3_stop}
\end{figure}

Among the aforementioned sensors, at least for automated systems in the transportation and surveillance domains, cameras are more common/crucial. Existing
remote attacks against cameras are limited to, essentially, denial-of-service attacks
\cite{truong2005preventing, petit2015remote, yan2016can}, which are easily
detectable (e.g., by tampering detection~\cite{ribnick2006real}) and for which effective mitigation strategies exist (e.g., by sensor fusion~\cite{li2013sensor}). In this work, we
consider attacks that cause camera-based image classification system to either misperceive actual objects or perceive non-existent objects
by remotely injecting light-based interference into a camera, without
blinding it. Formally, we consider \emph{creation attacks}
whereby a spurious object (e.g., a non-existent traffic sign, or obstacle) is 
seen to exist in the environment by a camera, and \emph{alteration attacks}, in which an
existing object in the camera view is changed into another attacker-determined object (e.g.,
changing a STOP sign to a YIELD sign or changing an intruder into a
bicycle).

As it is not possible, due to optical principles, to directly project an image into a camera, we propose to exploit two common effects
in optical imaging systems, viz., \emph{lens flare effects} and \emph{exposure
control} to induce camera-based misperception. The former effect is due to the imperfection of lenses, which causes light
beams to be refracted and reflected multiple times resulting in polygon-shape
artifacts (a.k.a., \emph{ghosts}) to appear in images~\cite{fotographee,
hullin2011flare}.  Since ghosts and their light sources typically appear at
different locations, an attacker can overlap specially crafted ghosts with the target
object's without having the light source blocking it.  Auto exposure control
is a feature common to cameras that determines the amount of light incident on the imager and is used, for example, to make
images look more natural.  An attacker can leverage exposure control to make the
background of an image darker and the ghosts brighter, so as to make the ghosts more prominent (i.e., noticeable to the detector/classifier) and thus increase attack success
rates. Fig.~\ref{fig:yolov3_stop} presents an example of a creation attack,
where we used a projector to inject an image of a STOP sign in a ghost, which
is detected and classified as a STOP sign by YOLOv3~\cite{redmon2016you}, a
state-of-the-art object detector. 

Theoretically arbitrary patterns can be injected via ghosts. However, it is challenging to practically and precisely control the ghosts, in terms of their
resolutions and positions in images, making arbitrary injection impracticable in some scenarios. Hence, we propose an empirical
projector-camera channel model that predicts the resolution and color of injected
ghost patterns, as well as the location of ghosts, for a given 
projector-camera arrangement. Experimental
results show that at short distances attack success rates are as high
as 100\%, but at longer distances the rates decrease sharply; this is because
at long distances ghost resolutions are low, resulting in patterns that cannot be recognized by the classifier.

To improve the efficacy of our attack, which we dub GhostImage, especially at lower
resolutions, we assume that the attacker possesses knowledge  about the
image classification/detection algorithm. Based on this knowledge the attacker is able to formulate and solve an
optimization problem to find optimal attack patterns, of varying resolutions, to project that will be recognized by the image classifier as the intended target class  \cite{szegedy2014intriguing, carlini2017towards}; i.e., the pattern projected will yield a classification result of the attacker's choice. As the channel may distort the injected image (in terms of color, brightness,
and noise), we extend our projector-camera model to include auto
exposure control and color calibration and integrate the channel model into
our optimization formulation. This results in a pattern generation approach that is resistant to
channel effects and thus able to defeat a classifier under realistic conditions.

We use self-driving and surveillance systems as two illustrative examples
to demonstrate the potential impact of GhostImage attacks. Proof-of-concept experiments were conducted with
different cameras, image datasets, and environmental conditions. Results show that our attacks are able to achieve attack
success rates as high as 100\%, depending on the projector-camera distance. Our contributions are summarized as follows.
\begin{itemize}[leftmargin=10pt]
	\item We are the first to study \emph{remote} perception attacks against
		camera-based classification systems, whereby the attacker
		induces misclassification of objects by injecting light, conveying
		adversarially generated patterns, into the camera.

	\item Our attack leverages optical effects/techniques, namely, lens flare and auto-exposure control, that are widespread and common, making the attack likely to be effective against most cameras.  Furthermore, we incorporate these effects in an
		end-to-end manner into an adversarial
		machine learning-based optimization framework to find the optimal
		patterns an attacker should inject to cause misperception.

	\item We demonstrate the efficacy of the attacks through experiments with varying image datasets, cameras,
		distances, and indoor to outdoor environments. Results show that
		GhostImage attacks are able to achieve attack success rates as high as
		100\%, depending on the projector-camera distance. %
\end{itemize}

%% file: secs/threat_model.tex
\section{System and Threat Model}
\label{sec:model}

System and attack models are described, including two
attack objectives and the attacker's capabilities.

\subsection{System Model}
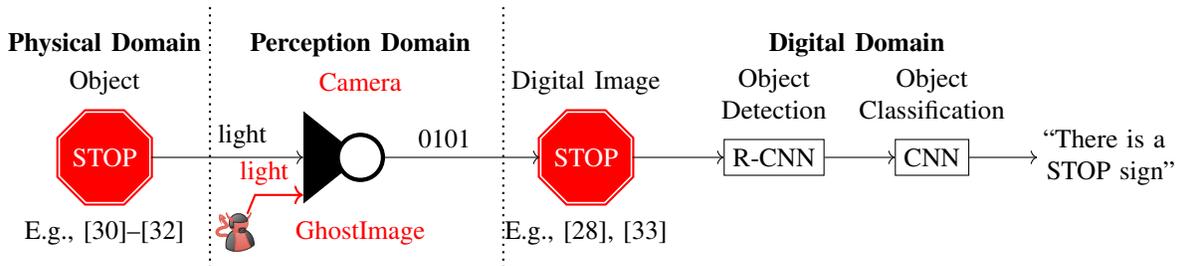
\begin{figure*}[t]
	\centering
	\input{./figs/sys_model.tikz}

	\caption{Camera-based object classification systems. GhostImage attacks
	target the perception domain, i.e., the camera.}

	\label{fig:sys_model}
\end{figure*}

We assume an end-to-end camera-based object classification system
(Fig.~\ref{fig:sys_model}) in which a camera captures an image of a scene with
objects of interest. The image is then fed to an object detector to crop out
the areas of objects, and finally these areas are given to a neural network to
classify the objects. Autonomous systems increasingly rely on such classification systems to make  decisions and actions. If the classification result is incorrect (e.g., modified by an adversary), wrong actions could be taken. For example, in a surveillance system, if an intruder is not
detected, the house may be broken-in without raising an alarm.

\subsection{Threat Model}
We consider two different attack objectives. In  \textbf{creation attacks} the goal is to inject a spurious (i.e., non-existent)
object into the scene and have it be recognized (classified) as though it were physically present. For \textbf{alteration attacks} an attacker injects adversarial patterns over an object of interest in the scene that causes the object to be misclassified.

There are two types of attackers with differing capabilities:
\textbf{Camera-aware attackers} who possess knowledge of the victim's camera
(i.e., they do not know the configuration of the lens system, nor
post-processing algorithms, but they can possess the same type of camera used
in the target system), from which they can train a channel model using the
camera as a black-box. With such capabilities, they are able to achieve
creation attacks and alteration attacks. \textbf{System-aware attackers} not
only possess the capabilities of the camera-aware attackers, but also know
about the image classifier including its architecture and parameters, i.e.,
black-box attacks on the camera but white-box attacks on the classifier. With
such capabilities, it is able to achieve creation attacks and alteration
attacks as well, but with higher attack success rates.

Both types of attackers are remote (unlike \cite{li2019adversarial}), i.e.,
they do not have access to the hardware or the firmware of the victim camera,
nor to the images that the camera captures. We assume that both attackers
are able to track and aim victim cameras \cite{cao2019adversarial,
truong2005preventing, chinalasers}.

%% file: figs/sys_model.tikz
\begin{tikzpicture}

\begin{scope}[shift={(0, 0)}]
\node (s1) [regular polygon, regular polygon sides=8, draw=red, double,
thick, fill=red, text=white, align=center, inner sep=0mm, node contents={STOP}];
\end{scope}

\begin{scope}[shift={(3.4,0)}, rotate=90, scale=.3]
	    \draw[fill=black] (0,0) -- (2,2.5) -- (-2,2.5) -- cycle;
	    \draw [fill=white,ultra thick](0,0) circle (1);
\end{scope}

\begin{scope}[shift={(6.4, 0)}]
\node (s2) [regular polygon, regular polygon sides=8, draw=red, double,
thick, fill=red, text=white, align=center, inner sep=0mm, 
node contents={STOP}];
\end{scope}

\node (od) [draw] at (8.9, 0) {R-CNN};
\node (oc) [draw] at (11, 0) {CNN};
\node (cls) [text width=50] at (13.4, 0) {``There is a STOP sign''};

\node at (0, 1) {Object};
\node [red] at (3.4, 1) {Camera};
\node at (6.4, 1) {Digital Image};
\node [text width=40, align=center, anchor=center] at (8.9, .85) {Object Detection};
\node [text width=55, align=center, anchor=center] at (11,  .85) {Object Classification};

\node at (0, -1) {E.g., \cite{eykholt2018robust,sharif2016accessorize,zhao2019seeing}};
\node [red] at (3.4, -1) {GhostImage};
\node at (6.4, -1) {E.g., \cite{szegedy2014intriguing, goodfellow2014explaining}};

\draw[->, black] (s1) -- (2.64,0) node [pos=.6, anchor=east, above] {light};
\draw[->, black] (3.68, 0) -- (s2) node [pos=.4, anchor=west, above] {0101};
\draw[->, black] (s2) -- (od);
\draw[->, black] (od) -- (oc);
\draw[->, black] (oc) -- (cls);

\draw[dotted, thick] (1.4, 2) -- (1.4, -1.5);
\draw[dotted, thick] (5.3, 2) -- (5.3, -1.5);
\node at (0, 1.5) {\textbf{Physical Domain}};
\node at (3.4, 1.5) {\textbf{Perception Domain}};
\node at (10, 1.5) {\textbf{Digital Domain}};

\node[devil] (devil) at (1.8, -1) {};
\draw[->, red, thick] (devil) -- (2, -.5) -- (2.64, -.5) node [pos=.9, anchor=south east, red] {light};

\end{tikzpicture}

%% file: secs/background.tex
\section{Background}
\label{sec:background}

In this section, we will introduce optical imaging principles, including
flare/ghost effects and exposure control, which we will exploit to \emph{realize}
GhostImage attacks. Then, we will discuss the preliminaries about neural
networks and adversarial examples that we will use to \emph{enhance} GhostImage
attacks.

\subsection{Optical Imaging Principles}
\label{sec:flare}

Due to the optical principles of camera-based imaging systems, it is not
feasible to directly point a projector at a camera, hoping that the projected
patterns can appear at the same location with the image of the targeted object,
because the projector has to obscure the object in order to make
the two images  overlap. We prove this infeasibility in Appendix~\ref{sec:dir}.
Instead, we exploit lens flare effects and auto exposure control to inject
adversarial patterns.

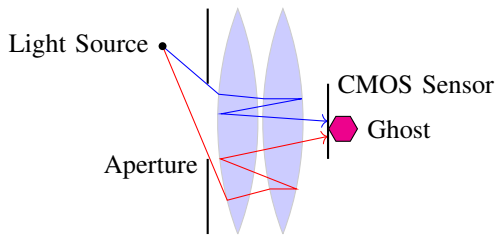
\begin{figure}[b]
	\centering
	\input{./figs/flare.tikz}
	\caption{Ghost effect principle}
	\label{fig:ghost}
\end{figure}

\emph{Lens flare effects}~\cite{hullin2011flare} refer to a phenomenon where
one or more undesirable artifacts appear on an image because bright light  get
scattered or flared in a non-ideal lens system (Fig.~\ref{fig:ghost}).
Ideally, all light beams should pass directly through the lens and reach the
CMOS sensor. However, due to the quality of the lens elements, a small portion
of light gets reflected several times within the lens system and then reaches
the sensor, forming multiple polygons (called ``ghosts'') on the image.  The
shape of polygons depends on the shape of the aperture. For example, if the
aperture has six sides, there will be hexagon-shaped ghosts in the image.
Normally ghosts are very weak and one cannot see them, but when a strong light
source (such as the sun, a light bulb, a laser, or a projector) is present
(unnecessarily captured by the CMOS sensor, though \cite{wiki-flare}), the
ghost effects become visible.  Fig.~\ref{fig:ghost} shows only one reflection
path, but there are many other paths and that is why there are usually multiple
ghosts in an image.

Existing literature~\cite{hullin2011flare, lee2013practical,
steinert2011general} about ghosts focused on the simulation of ghosts given the
detailed lens configurations, in which the algorithms simulate every possible
reflection path. Such white-box models are computationally expensive, and also
requires white-box knowledge of internal lens configurations, thus are not
suitable for our purposes. In Sections~\ref{sec:naive-attacks} and
\ref{sec:adv-attacks}, we study flare effects in a black-box manner (more
general than Vitoria et al.~\cite{vitoria2019automatic}), where we train a
lightweight end-to-end model that is able to predict the locations of ghosts,
estimate the resolutions within ghost areas, and also calibrate colors.

\emph{Exposure control} mechanisms \cite{lee2005introduction,
exposurecontrol} are often equipped in cameras to adjust brightness by changing
the size of the aperture or the exposure time. In this work, we will model and
exploit auto exposure control to manipulate the brightness balance between the
targeted object and the injected attack patterns in ghosts.

\subsection{Neural Nets and Adversarial Examples}
\label{sec:nn-adv}

We abstract a neural network as a function $Y = f_\theta(x)$ and we omit the
details of it due to the page limit. The input $x \in \mathbb{R}^{w\times
h\times 3}$ (width, height and RGB channels) is an image, $Y \in \mathbb{R}^m$
is the output vector, and $\theta$ is the parameters of the network (which is
fixed thus we omit it for convenience).  A softmax layer is usually added to
the end of a neural network to make sure that $\sum_{i=1}^m Y_i = 1$ and $Y_i
\in [0, 1]$.  The classification result is $C(x) = \mathrm{argmax}_i Y_i$.
Also, the inputs to the softmax layer are called \emph{logits} and denoted as
$Z(x)$.

An adversarial example \cite{szegedy2014intriguing} is denoted as $y$, where $y
= x + \Delta$. Here, $\Delta$ is additive noise that has the same
dimensionality with $x$. Given a benign image $x$ and a target label $t$, an
adversary wants to find a $\Delta$ such that $C(x+\Delta) = t$, i.e.,
\emph{targeted attacks}. Note that, in this paper, the magnitude of $\Delta$ is not constrained below a small threshold, since the perceived images are
usually not directly observed by human users. But we still try to minimize it
because it represents the attack power and cost.

%% file: figs/flare.tikz
\begin{tikzpicture}
	\draw[fill=blue, opacity=.2] (0,0) arc (-30:30:2 and 3) arc (150:210:2 and 3) ;
	\draw[fill=blue, opacity=.2] (.6,0) arc (-30:30:2 and 3) arc (150:210:2 and 3) ;

	\draw[thick] (-.4, 0) -- (-.4, 1) node[pos=.9, anchor=east] {Aperture};
	\draw[thick] (-.4, 2) -- (-.4, 3);

	\draw[->, blue] (-1, 2.5) -- (-.253, 1.86) -- (0.35, 1.8) -- (.85, 1.8) -- (-.22, 1.6) -- (1.2, 1.5);
	\draw[->, red] (-1, 2.5) -- (-.14, 0.45) -- (.43, 0.6) -- (.78, .6) -- (-.23, 1) -- (1.2, 1.3);

	\draw[thick] (1.2, 1) -- (1.2, 2) node[pos=1, anchor=west] {CMOS Sensor};

	\node [draw, regular polygon, regular polygon sides=6, fill=magenta, label=right:{Ghost}]  at (1.4, 1.4) {};

	\node[circle,fill=black,inner sep=0pt,minimum size=3pt,label=left:{Light Source}] (a) at (-1, 2.5) {};
\end{tikzpicture}

%% file: secs/naive_attack.tex
\section{Camera-aware GhostImage Attacks}
\label{sec:naive-attacks}

In this section, we will discuss how a camera-aware attacker is able to inject
arbitrary patterns in the perceived image of the victim camera using
projectors. We will discuss the possibilities of using other attack vectors in
Section~\ref{sec:alter}.

\subsection{Technical Challenges}

Since we assume that the attacker do not have access to the images that the targeted
camera captures, he/she will have to be able to predict how ghosts might
appear in the image. First, the locations of ghosts should be predicted given
the relevant positions of the projector and the camera, so that the attacker
can align the ghost with the image of the object of interest to achieve
alteration attacks. Second, since a projector can inject shapes in ghost
areas, the attacker needs to find out the maximum resolution of shapes that it
can inject. Lastly, it is also challenging to realize the attacks
derived from the position and resolution models above with a limited budget.

\subsection{Ghost Pixel Coordinates}
\label{sec:ghost-pos}

\begin{figure}[t]
	\centering
	\begin{subfigure}[b]{.45\columnwidth}
		\centering
		\input{./figs/capture.tikz}
		\caption{Capture}
		\label{fig:capture}
	\end{subfigure}%
	\begin{subfigure}[b]{.45\columnwidth}
		\centering
		\input{./figs/cam-vs-proj.tikz}
		\caption{Projection}
		\label{fig:projection}
	\end{subfigure}

	\caption{Capture and projection are reverses of each other.}
	\label{fig:dual}
\end{figure}
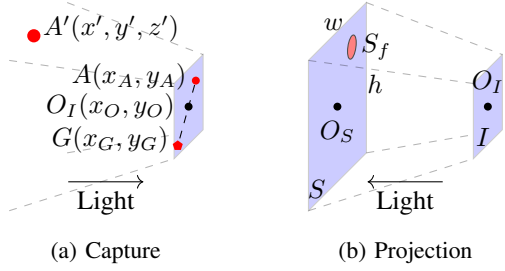

Given the pixel coordinates of the target object $G$ (Fig.~\ref{fig:capture}),
we need to derive the real-world coordinates $A'$ of the projector so that we
know where to place the projector in order to let one of the ghosts
overlap with the image of the object. To do this, we derive the relationship
between $G$ and $A'$ in two steps: We first calculate the pixel coordinates of
the light source $A$ given $A'$, and then we calculate $G$ based on $A$.

Based on homogeneous coordinates \cite{szeliski2010computer}, assuming the
camera is at the origin of the coordinate system, we have
\begin{equation}
	(u,v,w)^\top = M \cdot (x', y', z', 1)^\top,
	\label{eq:homo_coord}
\end{equation}
where $M$ is the camera's geometric model \cite{szeliski2010computer}, a
$3\times 4$ matrix. $M$ can be trained from another (similar) camera, and
then be applied to the victim camera. The
coordinates of $A$ is then $A = (x_A, y_A)^\top = (u/w, v/w)^\top$, by the
homogeneous transformation. Note that, $A$ does not have to appear in the view of the camera,
which makes the attack more stealthy (See \cite{wiki-flare} and
Fig.~\ref{fig:oov} in the appendix).

In order to find the relationship of the pixel coordinates between light
sources $A$ and their ghosts $G$, we did a simple experiment where we moved
around a flashlight in front of the camera~\cite{aptina}, and recorded the
pixel coordinates of the flashlight and the ghosts. Similar to
\cite{vitoria2019automatic}, we observe that, for
each $G$, we have $\overline{AO_I}/\overline{O_I G}=r_G$ (being constant), wherever
$A$ is (Fig.~\ref{fig:ghost-2d}), and $r \in
(-\infty, \infty)$. This means the feasible region for the placement of the projector is
large; to attack an autonomous vehicle, for example, it can be located on an overbridge,
on a traffic island, or even
in the preceding vehicle, or on a drone, etc. Finally, given $A=(x_A, y_A)$, $O_I=(x_O, y_O)$ and $r$,
\begin{equation}
	G = \begin{pmatrix}
		x_O-(x_A-x_O)/r\\
		y_O-(y_A-y_O)/r
	\end{pmatrix}.
	\label{eq:ghost-coord}
\end{equation}
With $G$'s coordinates, the attacker is able to predict the pixel location of ghosts and try
adjusting the position and orientation (which implies the angle) of the light source in the
real world so as to align one or more ghosts with the image of the object, whose pixel
coordinates can be calculated using (\ref{eq:homo_coord}) similarly.

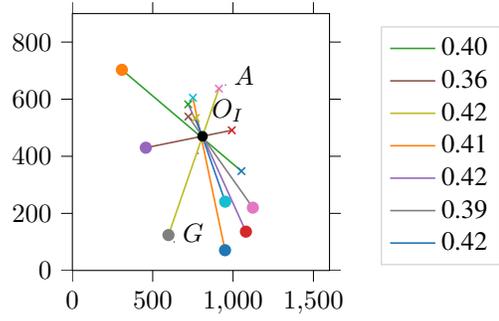
\begin{figure}[t]
	\centering
	\input{./figs/ghost-2d.tikz}

	\caption{Ghost position v.s. light source position. Crosses are light
	source at different locations and the circles are the according biggest ghosts (as examples).}

	\label{fig:ghost-2d}
\end{figure}

\subsection{Ghost Resolution}

In our daily life, ghosts normally appear as pieces of single-color
polygon-shaped artifacts; this is because the light sources that cause these
regular ghosts are single-point sources of light that have just one single
color, such as light bulbs, flashlights, etc. In this work, however, we find
out that one is able to bring patterns into these ghost areas by simply using
a low-cost projector, a special source of light that shines
variant patterns in variant colors. For example, in
Fig.~\ref{fig:yolov3_stop}, an image of a STOP sign that is projected by a
projector, appears in one of the ghost areas in the image; this is because the
pixel resolution of the projector is   high enough that multiple light beams in different colors (got reflected among lenses and then) go into the same ghost.
In this subsection, we   study the resolution of the patterns in ghost
areas~\footnote{We are interested in the resolution of the projector pixels,
not camera pixels; a projector pixel is usually captured by multiple camera
pixels.}.

Let us first define the \emph{throwing ratio} of a projector.  In
Fig.~\ref{fig:projection}, let plane $S$ be the projected screen (e.g. on a
wall), whose height and width are denoted as $h$ and $w$, respectively. The
distance $d=\overline{O_S O_I}$ is called the throwing distance. The throwing
ratio of this projection is $r_\text{throw} = d/w$.
The (physical) size of the projected screen at the victim camera's location is
denoted $S_O$, a part of which is captured by the CMOS sensor of the camera in
the ghost area, and we denote the (physical) size of that area as $S_f$. Let us
also define the resolution of the entire projected screen as $P_O$ in terms of
pixels (e.g., $1024\times 768$), and the resolution of the ghost as $P_f$.
Clearly, there is a linear relationship among them: $P_f/P_O = S_f/S_O$, where
$S_O = wh$. Finally, we can calculate the resolution of the ghost given $d$ and
$r_\text{throw}$:
\begin{equation}
	P_f = \frac{P_OS_f}{\frac{h}{w}\left(\frac{d}{r_\text{throw}}\right)^2}.
	\label{eq:ghost_reso}
\end{equation}
Here, $S_f$ is a constant because the size of the lens is fixed; e.g., the camera~\cite{aptina} has $S_f = 0.0156~\text{cm}^2$. %

\subsection{Attack Realization and Experiment Setup}
\label{sec:setup}

\begin{figure*}
	\begin{subfigure}[b]{.39\textwidth}
		\centering
		\input{./figs/inlab_setup_ill.tikz}
	\end{subfigure}
	\begin{subfigure}[b]{.39\textwidth}
		\centering
		\input{./figs/inlab_setup.tikz}
	\end{subfigure}
	\begin{subfigure}[b]{.2\textwidth}
		\begin{subfigure}{\columnwidth}
			\centering
			\includegraphics[width=.6\columnwidth]{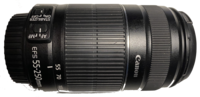}
		\end{subfigure}
		\begin{subfigure}[b]{\columnwidth}
			\centering
			\includegraphics[width=.8\columnwidth]{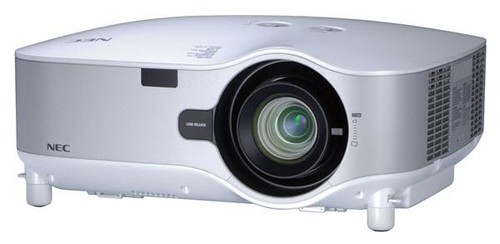}
		\end{subfigure}
	\end{subfigure}

	\caption{(Left) Attack setup diagram. (Middle) In-lab experiment setup.
	(Right) Attack equipments: We replaced the original lens of the NEC NP3150
	Projector~\cite{nec-np1150-manual} with a Canon EFS 55-250 mm zoom
	lens~\cite{canon}.}

	\label{fig:inlab_setup}

\end{figure*}

According to Eq.~\ref{eq:ghost_reso}, if the attacker wants to carry out
long-distance and high-resolution GhostImage attacks, it needs a projector with
a large throwing ratio $r_\text{throw}$. However, the factory longest-throw
lenses (NEC NP05ZL Zoom Lens~\cite{np05zl}) of our projector can achieve a
throwing ratio of maximum $7.2$ (which means $9\times 9$ at one meter), and
expensive (about \SI{1600}[\$]{}).
Instead, we use a cheap (\SI{80}[\$]{}) zoom lens (Fig.~\ref{fig:inlab_setup}, right)
\cite{canon} that was originally designed for Canon cameras. In our
experiments, such a configuration is interestingly feasible\footnote{Because
projectors and cameras are dual devices (Fig.~\ref{fig:dual}), their lenses are
interchangeable.}
(Fig.~\ref{fig:inlab_setup}), achieving the maximum throwing ratio of 20 when
the focal length is \SI{250}{\milli\meter}, which means that at a distance of
one meter, $32\times 32$-resolution attacks can be achieved. See
Sec.~\ref{sec:disc-prac} for more discussion on lens and projector selection.

Fig.~\ref{fig:inlab_setup} (left) shows a general diagram of GhostImage
attacks, where the light source (i.e., a projector) is pointing at the camera
from the side, so that the camera can still capture the object (e.g., a STOP
sign) for alteration attacks. The light source injects light interference
(marked in blue) into the camera, which gets reflected among the lenses of the
camera, resulting in ghosts that overlap with the object in the image.
Accordingly, a photo of our in-lab experiment setup is given in
Fig.~\ref{fig:inlab_setup}. The Canon lens was loaded in the NEC projector,
though it cannot be seen in the photo. We will evaluate our attack on three
different cameras (Sec.~\ref{sec:diff_cam}).

\begin{figure}[b]
    \centering
	\begin{subfigure}{0.18\columnwidth}
	  \includegraphics[width=\columnwidth]{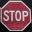}
	  \caption{$32\times 32$}
	\end{subfigure}
	\begin{subfigure}{0.18\columnwidth}
	  \includegraphics[width=\columnwidth]{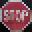}
	  \caption{$16\times 16$}
	\end{subfigure}
	\begin{subfigure}{0.18\columnwidth}
	  \includegraphics[width=\columnwidth]{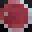}
	  \caption{$8\times 8$}
	\end{subfigure}
	\begin{subfigure}{0.18\columnwidth}
	  \includegraphics[width=\columnwidth]{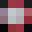}
	  \caption{$4\times 4$}
	\end{subfigure}
	\begin{subfigure}{0.18\columnwidth}
	  \includegraphics[width=\columnwidth]{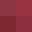}
	  \caption{$2\times 2$}
	\end{subfigure}
	\caption{Downsampling examples. We actually upsampled these images
	for the sake of presentation, otherwise they were too tiny to show.}
	\label{fig:ds_eg}
\end{figure}

To mount a creation attack, the attacker computes the maximum resolution $P_f$
for the ghost with a distance $d$ based on (\ref{eq:ghost_reso}), and then
\emph{downsamples} the target image to the resolution $P_f$
in order to fit in the ghost area. The attacker chooses downsampling as a
heuristic approach because he/she is not aware of the classification algorithm.
We present in Fig.~\ref{fig:ds_eg} some
examples of downsampling a STOP sign image.

To mount alteration attacks, in addition to (\ref{eq:ghost_reso}) for
downsampling, the attacker also needs to consider the pixel coordinates
(Eq.~\ref{eq:ghost-coord}) of the ghost because the attacker needs to align the
ghost with the image of the object of interest so that the resulting, combined
image deceives the classifier.

\subsection{GhostImage Attack Evaluation}
\label{sec:nai-eval}

We substantiate camera-aware attacks on an image classification system that we
envision would be used for automated vehicles. Specifically, images, taken by
an Aptina MT9M034 camera~\cite{aptina}, are fed to a traffic sign image
classifier trained on the LISA dataset~\cite{mogelmose2012vision}. In
Sec.~\ref{sec:ext_eval}, we will evaluate classification systems for other
applications, with different cameras and different datasets.

\subsubsection{Dataset and neural network architecture}
In order to train an unbiased classifier, we selected eight traffic signs from
the LISA dataset~\cite{mogelmose2012vision} (Table~\ref{tab:lisa} in
Appendix~\ref{sec:nntables}). The network architecture
(Table~\ref{tab:lisa_arch} in Appendix~\ref{sec:nntables}) is identical to
\cite{eykholt2018robust}.
We used $80\%$ of samples from the balanced dataset to train the network and
the rest $20\%$ to test the network; it achieved an accuracy of $96\%$.

\SetInd{0.25em}{0.7em}
\begin{algorithm}[b]
	\caption{GhostImage Attack evaluation procedure ($m$: number of classes)}
	\label{alg:naive}

	Initialize $q_d = 0$, for $1\le d\le 5$\\
	\For{$d = 1 \cdots 5$} {
		Place camera $d$ meters away from projector\\
		\For{$i=1 \cdots m$} {
			Place printed traffic sign of class $i$ at background\\
			\For{$j = 1 \cdots m$} {
				\textbf{if} $i=j$ \textbf{then} continue\\
				$Y \leftarrow$ $k$ randomly picked images of class $j$\\
				\For{$y \in Y$}{
					Downsample $y$ according to
					Eq.~\ref{eq:ghost_reso}\\
					Project and crop out $y'$\\
					\If{$y'$ is classified as $j$} {
						$q_d \leftarrow q_d + 1$\\
					}
				}
			}
		}
	}
	Success rates: $q_d \leftarrow q_d/(km^2)$, for $1\le d\le 5$\\

\end{algorithm}

\subsubsection{Evaluation methodology} The evaluation procedure for alteration
attacks is detailed in Algorithm~\ref{alg:naive} in which we iterated five
distances, $m$ source classes, $m$ target classes. For each target class, we
sampled $k$ images randomly from the dataset. For every combination, we first
downsampled the target image based on (\ref{eq:ghost_reso}), and projected the
image at the camera using the NEC projector. We then took the captured image,
cropped out the ghost area, and used the classifier to classify it. If the
classification result is the target class, we count it as a successful attack.
The procedure for creation attacks is slightly different: Rather than printed
traffic signs, we placed a blackboard as the background as it helped us locate the
ghosts. Given a throwing radio of $20$ (thanks to the Canon lens) we evaluated
five different distances from one meter to five meters. Based on
(\ref{eq:ghost_reso}), they resulted in $32\times 32$, $16\times 16$, $8\times
8$, $4\times 4$, and $2 \times 2$ resolutions, respectively.

\begin{table}[b]
	\centering
	\caption{Camera-aware attack success rates}
	\label{tab:naive-rates}
	\rowcolors{5}{}{gray!10}
	\begin{tabular}{ccccc}
		\toprule
		Distances & \multicolumn{2}{c}{Creation Attacks} & \multicolumn{2}{c}{Alteration Attacks} \\
		\cmidrule(lr){2-3} \cmidrule(lr){4-5}
		(meter) & Digital & Perception & Digital & Perception\\
		\midrule
		1 & 98\% & 41\% & 95\% &  33\% \\
		2 & 98\% & 36\% & 88\% &  33\% \\
		3 & 80\% & 34\% & 67\% &  34\% \\
		4 & 36\% & 15\% & 28\% &  10\% \\
		5 & 14\% & 10\% & 13\% &  0\%  \\ 
		\bottomrule
	\end{tabular}
\end{table}

\begin{figure}[tb]
    \centering
	\begin{subfigure}[b]{0.45\columnwidth}
   	  \centering
	  \includegraphics[width=\columnwidth]{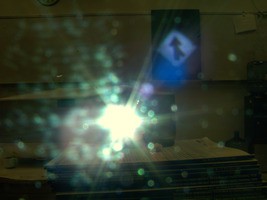}
	\end{subfigure}
	\begin{subfigure}[b]{0.45\columnwidth}
   	  \centering
	  \includegraphics[width=\columnwidth]{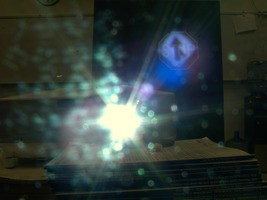}
	\end{subfigure}

	\caption{Camera-aware attack examples at one meter in perception
	domain. Left: Creating a Merge sign. Right: Altering a STOP sign (in
	the background) into a Merge sign.}

	\label{fig:naive_eg}
\end{figure}

\subsubsection{Results}

The results about attack success rates of camera-aware attacks at varying
distances are shown in Table~\ref{tab:naive-rates} (Fig.~\ref{fig:naive_eg}
illustrates two successful camera-aware attacks). For the digital domain, we
simply added attack images $\Delta$ on benign images $x$ as
$y=(x+\Delta)/\|x+\Delta\|_\infty$. Based on these experiments, we observe:
First, as the distance increases, the success rate decreases. This is because
lower-resolution images are less well recognized by the classifier.  Second,
digital domain results are  better than perception domain one, because images
are distorted by the projector-camera channel effects. Third, creation attacks
result  in higher success rates than alteration attacks do because in
alteration attacks there are benign images in the background, encouraging the
classifier to make correct classifications. We will address these issues in the
next section, so as to increase the overall attack success rate. 

%% file: figs/capture.tikz
\begin{tikzpicture}[scale=.5]
	\makeatletter
	\tikzoption{canvas is xy plane at z}[]{%
	  \def\tikz@plane@origin{\pgfpointxyz{0}{0}{#1}}%
	  \def\tikz@plane@x{\pgfpointxyz{1}{0}{#1}}%
	  \def\tikz@plane@y{\pgfpointxyz{0}{1}{#1}}%
	  \tikz@canvas@is@plane
	}
	\makeatother
	
	\tikzset{xyp/.style={canvas is xy plane at z=#1}}
	\tikzset{xzp/.style={canvas is xz plane at y=#1}}
	\tikzset{yzp/.style={canvas is yz plane at x=#1}}

	\fill[opacity=0.2,blue,draw=black,yzp=0] (0,0) (-1,-1) rectangle (1, 1);

	\node [circle, fill=black, inner sep=0pt,minimum size=3pt,label=left:{$O_I(x_O, y_O)$}] (op) at (0, 0, 0) {};
	\node [circle, fill=red, inner sep=0pt,minimum size=5pt,label={[right]:{$A'(x', y', z')$}}] (a) at (-4.5, 1.5, -1) {};
	\node [circle, fill=red, inner sep=0pt,minimum size=3pt,label={[left]:{$A(x_A, y_A)$}}] (ap) at (0, 1/2, -1/2) {};
	\node [regular polygon, regular polygon sides=5, fill=red, inner sep=0pt,minimum size=4pt,label={[left]:{$G(x_G, y_G)$}}] (g) at (0, -3/4, 3/4) {};

	\draw [dashed] (ap) -- (g);

	\draw [dashed, opacity=.3] (-4, -2, 2)  -- (0, -1, 1);
	\draw [dashed, opacity=.3] (-4, 2, -2) -- (0, 1, -1);
	\draw [dashed, opacity=.3] (-4, 2, 2) -- (0, 1, 1);
	\draw [dashed, opacity=.3] (-4, -2, -2) -- (-.8, -.8/4-1, -.8/4-1);

	\draw[->] (-3.2, -2) -- (-1.2, -2) node[midway, anchor=center, below] {Light};

\end{tikzpicture}

%% file: figs/cam-vs-proj.tikz
\begin{tikzpicture}[scale=.5]
	\makeatletter
	\tikzoption{canvas is xy plane at z}[]{%
	  \def\tikz@plane@origin{\pgfpointxyz{0}{0}{#1}}%
	  \def\tikz@plane@x{\pgfpointxyz{1}{0}{#1}}%
	  \def\tikz@plane@y{\pgfpointxyz{0}{1}{#1}}%
	  \tikz@canvas@is@plane
	}
	\makeatother
	
	\tikzset{xyp/.style={canvas is xy plane at z=#1}}
	\tikzset{xzp/.style={canvas is xz plane at y=#1}}
	\tikzset{yzp/.style={canvas is yz plane at x=#1}}

    \fill[opacity=0.2,blue,draw=black,yzp=0] (0,0) (-1,-1) rectangle (1, 1);
    \fill[opacity=0.2,blue,draw=black,yzp=-4] (0,0) (-2,-2) rectangle (2, 2);
	\draw[yzp=-4, draw=black!70, fill=red!50] (1.2, -1) circle [radius=.3] node[anchor=west] {$S_f$};
	\node at (-4.55, -2.2) {$S$};
	\node at (-.13, -.8) {$I$};

	\node [circle, fill=black, inner sep=0pt,minimum size=3pt,label=above:{$O_I$}] (op) at (0, 0, 0) {};
	\node [circle, fill=black, inner sep=0pt,minimum size=3pt,label=below:{$O_S$}] (o) at (-4, 0, 0) {};


	\draw [dashed, opacity=.3] (-4, -2, 2)  -- (0, -1, 1);
	\draw [dashed, opacity=.3] (-4, 2, -2) -- (0, 1, -1);
	\draw [dashed, opacity=.3] (-4, 2, 2) -- (0, 1, 1);
	\draw [dashed, opacity=.3] (-4, -2, -2) -- (-.8, -.8/4-1, -.8/4-1);

	\draw[<-] (-3.2, -2) -- (-1.2, -2) node[midway, anchor=center, below] {Light};

	\node at (-4.1, 2.2) {$w$};
	\node at (-3, .6) {$h$};
\end{tikzpicture}

%% file: figs/ghost-2d.tikz
\begin{tikzpicture}

\definecolor{color0}{rgb}{0.12156862745098,0.466666666666667,0.705882352941177}
\definecolor{color1}{rgb}{1,0.498039215686275,0.0549019607843137}
\definecolor{color2}{rgb}{0.172549019607843,0.627450980392157,0.172549019607843}
\definecolor{color3}{rgb}{0.83921568627451,0.152941176470588,0.156862745098039}
\definecolor{color4}{rgb}{0.580392156862745,0.403921568627451,0.741176470588235}
\definecolor{color5}{rgb}{0.549019607843137,0.337254901960784,0.294117647058824}
\definecolor{color6}{rgb}{0.890196078431372,0.466666666666667,0.76078431372549}
\definecolor{color7}{rgb}{0.737254901960784,0.741176470588235,0.133333333333333}
\definecolor{color8}{rgb}{0.0901960784313725,0.745098039215686,0.811764705882353}

\begin{axis}[
height=5cm, width=5cm,
legend cell align={left},
legend style={at={(1.2, .95)}, anchor=north west, draw=white!80.0!black},
tick align=outside,
tick pos=both,
x grid style={white!69.01960784313725!black},
xmin=0, xmax=1600,
xtick style={color=black},
y grid style={white!69.01960784313725!black},
ymin=0, ymax=900,
ytick style={color=black}
]
\addplot [semithick, color0, mark=x, mark size=2, mark options={solid}, only marks, forget plot]
table {%
1052 348
};
\addplot [semithick, color1, mark=*, mark size=2, mark options={solid}, only marks, forget plot]
table {%
308 703
};
\addplot [semithick, color2]
table {%
1052 348
308 703
};
\addlegendentry{0.40}
\addplot [semithick, color3, mark=x, mark size=2, mark options={solid}, only marks, forget plot]
table {%
992 491
};
\addplot [semithick, color4, mark=*, mark size=2, mark options={solid}, only marks, forget plot]
table {%
457 430
};
\addplot [semithick, color5]
table {%
992 491
457 430
};
\addlegendentry{0.36}
\addplot [semithick, color6, mark=x, mark size=2, mark options={solid}, only marks, forget plot]
table {%
912 637
};
\addplot [semithick, white!49.80392156862745!black, mark=*, mark size=2, mark options={solid}, only marks, forget plot]
table {%
598 124
};
\addplot [semithick, color7]
table {%
912 637
598 124
};
\addlegendentry{0.42}
\addplot [semithick, color8, mark=x, mark size=2, mark options={solid}, only marks, forget plot]
table {%
750 605
};
\addplot [semithick, color0, mark=*, mark size=2, mark options={solid}, only marks, forget plot]
table {%
950 71
};
\addplot [semithick, color1]
table {%
750 605
950 71
};
\addlegendentry{0.41}
\addplot [semithick, color2, mark=x, mark size=2, mark options={solid}, only marks, forget plot]
table {%
721 582
};
\addplot [semithick, color3, mark=*, mark size=2, mark options={solid}, only marks, forget plot]
table {%
1080 136
};
\addplot [semithick, color4]
table {%
721 582
1080 136
};
\addlegendentry{0.42}
\addplot [semithick, color5, mark=x, mark size=2, mark options={solid}, only marks, forget plot]
table {%
722 539
};
\addplot [semithick, color6, mark=*, mark size=2, mark options={solid}, only marks, forget plot]
table {%
1123 220
};
\addplot [semithick, white!49.80392156862745!black]
table {%
722 539
1123 220
};
\addlegendentry{0.39}
\addplot [semithick, color7, mark=x, mark size=2, mark options={solid}, only marks, forget plot]
table {%
768 534
};
\addplot [semithick, color8, mark=*, mark size=2, mark options={solid}, only marks, forget plot]
table {%
951 241
};
\addplot [semithick, color0]
table {%
768 534
951 241
};
\addlegendentry{0.42}
\draw[] (axis cs:950,650) -- (axis cs:950,650);
\node at (axis cs:950,650)[
  anchor=base west,
  text=black,
  rotate=0.0
]{$A$};
\draw[] (axis cs:630,100) -- (axis cs:630,100);
\node at (axis cs:630,100)[
  anchor=base west,
  text=black,
  rotate=0.0
]{$G$};
\draw[] (axis cs:780,410) -- (axis cs:780,410);
\node [circle, fill=black, inner sep=0pt,minimum size=4pt,label={[above right]:{$O_I$}}] at (axis cs:810,470) {};
\end{axis}

\end{tikzpicture}

%% file: figs/inlab_setup_ill.tikz
\begin{tikzpicture}[decoration={markings,mark=at position 0.5 with {\arrow{>}}}, scale=.7]

\begin{scope}[shift={(0, 0)}]
\node (s) [regular polygon, regular polygon sides=8, draw=red, double,
thick, fill=red, text=white, align=center, inner sep=0mm] {\footnotesize \textbf{STOP}};
\end{scope}

\begin{scope}[shift={(5.5,0)}, rotate=90, scale=.5]
	    \draw[fill=black] (0,0) -- (2,2.5) -- (-2,2.5) -- cycle;
	    \draw[fill=white,ultra thick](0,0) circle (1);
\end{scope}

\begin{scope}[shift={(1.5,-2)}, rotate=-50, scale=.3]
	    \draw[fill=red!80, draw=none] (0,0) -- (2,2.5) -- (-2,2.5) -- cycle;
	    \draw [draw=red!80, fill=blue!80,ultra thick](0,0) circle (1);
\end{scope}

\node (is) [regular polygon, regular polygon sides=8, draw=red,
		fill=red, inner sep=0pt, minimum size=10pt] at (5.55, .05) {};
\node (g) [regular polygon, regular polygon sides=6, draw=blue!40, opacity=.8,
		fill=blue!40] at (5.45, -0.05) {};

\draw[->, dashed] (6, -0.5) -- (g) node[pos=0, anchor=west] {Ghost};
\draw[->, dashed] (6, 0.5) -- (is) node[pos=0, anchor=west, text width=30] {Object Image};

\node at (0, -1.1) {Object};
\node at (5.3, -1.1) {Camera};
\node[anchor=east] at (1.2, -2) {Projector};

\draw[postaction={decorate}, red] (s) -- (4.23, 0);
\draw[postaction={decorate}, blue] (2.1, -1.55) -- (4.23, 0);

\end{tikzpicture}%


%% file: figs/inlab_setup.tikz
\begin{tikzpicture}[y=-1cm]
\node[anchor=north west,inner sep=0pt] at (3.81,5.08)
{\includegraphics[width=7cm]{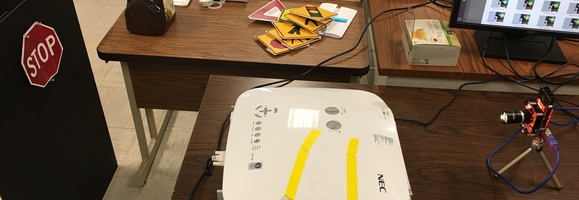}};
\path (10.7,7.2) node[text=white,anchor=base east] {Camera~\cite{aptina}};
\path (8.7, 7.3) node[text=black,anchor=base east] {Projector~\cite{nec-np1150-manual}};
\path (5.5,6.5) node[text=white,anchor=base east,text width=40] {Traffic Sign (printed)};
\draw [draw=black, fill=black] (5,6) rectangle (10.8, 5.1);
\draw[->, draw=white, thick] (9.3, 5.9) -- (8.6, 6.5) node[pos=0, anchor=south, white] {Canon
Lens~\cite{canon}};

\end{tikzpicture}%


%% file: secs/adv_attack.tex
\section{System-aware GhostImage Attacks}
\label{sec:adv-attacks}

There are some limitations of the camera-aware attack  introduced in the
previous section. First, increasing distances results in lower success rates
because the classifier cannot recognize the resulting low-resolution images.
Second, there are large gaps between digital domain results and perception
domain results, as channel effects (which cause the inconsistency between the intended pixels and the perceived pixels) are not taken into account. In this section, we resolve these limitations and improve GhostImage attacks' success rates by proposing a framework which consists of a channel model that predicts the pixels perceived by the camera, given the pixels as input to the projector, as well as an optimization formulation based on which the attacker can solve
for optimal attack patterns that cause misclassification by the target
classifier with high confidence.

\subsection{Technical Challenges}
\label{sec:adv-cha}

First, the injected pixel values are often difficult to control as they exhibit
randomness due to variability of the channel between the projector and the
camera, thus the adversary is not able to manipulate each pixel
deterministically. Second, to achieve optimal results, the attacker needs to
precisely predict the projected and perceived pixels, thus channel effects must
be modeled in an end-to-end manner, i.e., considering not only the physical channel
(air propagation), but also the internal processes of the projector and the camera.
Lastly, the resolution of attack patterns is limited by distances and projector
lens (Eq.~\ref{eq:ghost_reso}), thus the ghost patterns must be carefully
designed to fit the resolution with few degrees of freedom.

\subsection{System-aware Attack Overview}

The system-aware attacker aims to find optimal patterns that can cause
misclassification by the target classifier with high confidence by taking advantage of the
non-robustness of the classifier \cite{szegedy2014intriguing}. We adopt an adversarial example-based
optimization formulation into GhostImage attacks, in which the attacker tries to solve
\begin{equation}
	\Delta^* = \argmin_\Delta \|\Delta\|_p + c\cdot\LL(y, t, \theta),
	\label{eq:main}
\end{equation}
where $\Delta$ is the digital attack pattern as input to the projector, $y$ is
the perceived image of the object of interest under attacks, $t$ is the target class, and
$\theta$ represents the targeted neural network.
$\|\cdot\|_p$ is an $\ell_p$-norm that measures the magnitude of a
vector, and $\LL$ is a loss function indicating how (un)successful $\Delta$ is.
Here, we aim to minimize the power of the projector required for a successful attack, meanwhile
maximizing the successful chance of attacks. The relative importance of these
two objectives  is balanced by a constant  $c$. In Sec.~\ref{sec:adv-patterns},
we will detail (\ref{eq:main}) in terms of how we handle $\Delta$ being a
non-negative random tensor that is also able to depict grid-style patterns in
different resolutions.

More importantly, in (\ref{eq:main}) $y$ is the final perceived image used as
input to the
classifier, which is estimated by our channel model in an end-to-end style
(Fig.~\ref{fig:channel_pipe}), in which $\delta$ \footnote{Different than
$\Delta$ which is a $w\times h\times
3$ tensor, $\delta$ is a single pixel with dimension $3\times 1$ for the
convenience of the analysis.} is the input to
the projector, and $y$ is the resulting image captured by the camera.
The model can be formulated as
\begin{equation}
	y = g\left(h_f(\Delta) + h_o(x)\right).
	\label{eq:main-ch}
\end{equation}
where $h_f(\Delta)$ is the ghost model that estimates the perceived adversarial
pixel values in the ghost. For simplicity we let $h_o(x)=x$ because the
attacker possesses  same type of the camera so that $x$ can be obtained a
priori, and $g(\cdot)$ is the auto exposure control that adjusts the
brightness. Sec.~\ref{sec:channel-model} introduces the derivation of
(\ref{eq:main-ch}).

Next, we will first present the channel model, and then
formulate the optimization problem for finding the optimal adversarial ghost
patterns.

\begin{figure}[t]
	\centering
	\input{./figs/channel_pipe.tikz}
	\caption{Projector-camera channel model}
	\label{fig:channel_pipe}
\end{figure}
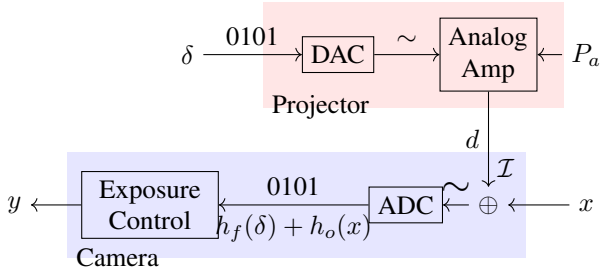

%% file: figs/channel_pipe.tikz
\begin{tikzpicture}

\node (dac) [draw, align=center, anchor=center] at (10, 6) {DAC};
\node (x) [align=center, anchor=center] at (8, 6) {$\delta$};
\node (xo) [align=center, anchor=center] at (13.3, 4) {$x$};
\node (amp) [text width=30, draw, align=center, anchor=center] at (12, 6) {Analog Amp};
\node (pa) [align=center, anchor=center] at (13.3, 6) {$P_a$};
\node (pls) [align=center, anchor=center] at (12, 4) {$\oplus$};
\node (adc) [draw, align=center, anchor=center] at (10.9, 4) {ADC};
\node (exp) [text width=45, draw, align=center, anchor=center] at (7.5, 4) {Exposure Control};
\node (y) [align=center, anchor=center] at (5.7, 4) {$y$};

\draw[->] (x) -- (dac) node [midway, above] {0101};
\draw[->] (pa) -- (amp);
\draw[->] (dac) -- (amp) node [midway, above] {$\sim$};
\draw[->] (amp) -- (pls) node [pos=.8, right] {$\I$} node [midway, left] {$d$};
\draw[->] (pls) -- (adc) node [midway, above] {\Large $\sim$};
\draw[->] (xo) -- (pls);
\draw[->] (adc) -- (exp) node [midway, above] {0101} node [midway, below] {$h_f(\delta)+h_o(x)$};
\draw[->] (exp) -- (y);

\begin{pgfonlayer}{bg}
\draw[draw=none, fill=red!10] (9, 5.3) rectangle (13, 6.7) node[pos=0, anchor=west] {Projector};
\draw[draw=none, fill=blue!10] (6.4, 3.3) rectangle (12.5, 4.7) node[pos=0, anchor=west] {Camera};
\end{pgfonlayer}

\end{tikzpicture}

%% file: secs/detailed_adv.tex
\subsection{Projector-Camera Channel Model}
\label{sec:channel-model}

We consider the projector to camera channel model (Fig.~\ref{fig:channel_pipe})
in which $\delta$ is an RGB value the attacker wishes to project
which is later converted to an analog color by the projector. The attacker can
control the power ($P_a$) of the light source of the projector so that the
luminescence can be adjusted. The targeted camera is situated at a distance of
$d$, which captures the light coming from both the projector and reflected off
the object ($x$). The illuminance received by the camera from the projector is
denoted as $\I$. The camera converts analog signals into digital ones, based on
which it adjusts its exposure, with the final RGB value being $y$. An ideal
channel would yield $y=x+\delta$ but due to channel effects, we need to find a
way to adjust the projected RGB value such that the perceived RGB value is as intended.

\subsubsection{Exposure control}
\label{sec:exp_ctrl}

As we   discussed in Section~\ref{sec:flare}, cameras are usually equipped
with auto-exposure control, where according to the overall brightness of the
image, the camera adjusts its exposure by changing the exposure time, or the
size of its aperture, or both. We observed from our experiments that, as we
increase the luminescence of the projector ($\I$), in the image the brightness
of the object ($x$) decreases but the ghost ($\delta$) does not decrease as
much. Modeling such phenomena helps the attacker  to precisely predict
the perceived image.
For the following, we will first find out how the
illuminance $\I$ depends on $\delta$ and $P_a$ (the normalized power of light
bulb ranging from $0\%$ to $100\%$), and then  how $y$ depends on $\I$.

\begin{figure}[tb]
	\centering
	\input{./figs/illuminance.tikz}

	\caption{Illuminance depends on the RGB amplitude $T_d$, and the light bulb
	intensity $T_a$.}

	\label{fig:illuminance}
\end{figure}
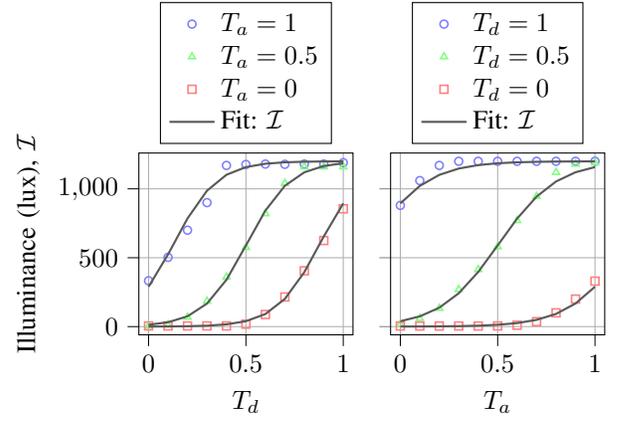

\paragraph{How does $\I$ depend  on $\delta$ and $P_a$?}
We conducted a series of
experiments, where $T_d =
\|\delta\|_\infty=\max_i \delta_i$ and $P_a$ were varied.  We recorded the
illuminance directly in front of the camera using an illuminance
meter~\cite{drmeter}, with the projector one meter away. The
results are plotted in Fig.~\ref{fig:illuminance}, which shows
that
\begin{equation}
	\I(T_d, P_a, d) = \frac{c_d}{d^2} \cdot \frac{\I_\text{max}}{1+e^{-t}},
	\label{eq:ill}
\end{equation}
where $t = a \times T_d + b \times P_a + c_t$, and
$a$, $b$, $c_d$ and $c_t$ are constants derived from the data.
$\I_\text{max}$ is the maximum illuminance of the projector at a
distance of one meter. Such a sigmoid-like
function captures the luminescence saturation property of the projector hardware.

\paragraph{How does the perceived $x$ depend  on $\I$?} In the same experiments we also
recorded the RGB value of the ghost ($\delta$) with a blackboard as background
(in order to reduce ambient impacts), and a piece of white paper ($x$) that was
also on the blackboard but did not overlay with the ghost. Their data are
shown in Fig.~\ref{fig:ill-vs-dim}, from
which we can derive
the \emph{dimming ratio} that measures the change of exposure/brightness:
\begin{equation}
	\gamma(\I) = \frac{\I_\text{env}}{\I+\I_\text{env}},
	\label{eq:gamma}
\end{equation}
where $\I_\text{env}$ is the ambient lighting condition in illuminance which
differs from indoors to outdoors for instances. From this equation, we see that
in an environment with static lighting condition, as the luminescence of the
projector increases, the dimming ratio decreases, hence the objects become darker.
With (\ref{eq:gamma}), the adversary is able to conduct real-time attacks by
simply plugging in the momentary $\I_\text{env}$.

\paragraph{How does the perceived $\delta$ depend  on $\I$?} When
$x=0$, $\|y_f\|=\|y_f\|_\infty$ (the lower subplot of
Fig.~\ref{fig:ill-vs-dim}) depends on $\I$ in two ways:
\[\|y_f\|(\I) = \gamma(\I) \cdot \rho \cdot \I.\] On one hand, the last term $\I$
increases the intensity of ghosts, but on the other hand the dimming ratio
$\gamma(\I)$ dims down ghost, whereby  $\rho$ is a trainable constant. With this, we
can rewrite the perceived flare as \[y_f = \|y_f\| H_c
\frac{\delta}{\|\delta\|},\] where $H_c$ is the color calibration matrix to deal with
color distortion, which will be
discussed in Section~\ref{sec:col-cal}. The term $1/\|\delta\|$ normalizes
$\delta$. In the end, we have the channel model
\begin{equation}
	y = \gamma(\I)\left(\rho\I H_c\frac{\delta}{\|\delta\|} + x\right).
	\label{eq:exp_ctrl}
\end{equation}
Compared to (\ref{eq:main-ch}), \[h_f(\delta)=\rho\I
H_c\frac{\delta}{\|\delta\|}, \quad g(t)=\gamma(\I)t, \quad h_o(x)=x.\]
With (\ref{eq:exp_ctrl}), the attacker is able to predict how bright
and what colors/pixel values the ghost and the object will be, given the projected pixels,
the power of the projector, and the distance.

\begin{figure}[t]
	\centering
	\input{./figs/ill-vs-dim.tikz}
	\caption{Perceived RGB values v.s. illuminance.}
	\label{fig:ill-vs-dim}
\end{figure}
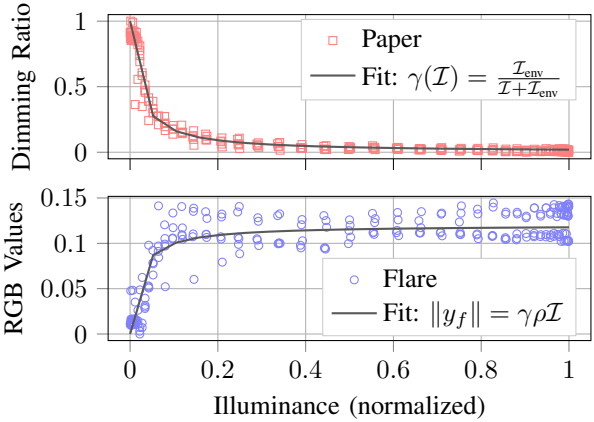

\subsubsection{Color calibration}
\label{sec:col-cal}

Considering a dark background (i.e., $x=0$), (\ref{eq:exp_ctrl}) can
be simplified as $y = \gamma(\I)\rho\I H_c \delta/\|\delta\|$, where $H_c$ is a
$3\times 3$ matrix (as three color channels) that calibrates colors. Both $y$
and $\delta$ are $3\times 1$ column vectors.  $H_c$ should be an identity
matrix for an ideal channel, but due to the color-imperfection of both the
projector and the camera, $H_c$ needs to be learned from data. To simplify
notations, we define corrected $x$ and $y$ as
\[\hat{x}=\frac{\delta}{\|\delta\|}, \quad \hat{y} = \frac{y}{\rho\I\gamma(\I)},\] so that we
can write \[\hat{y} = H_c\hat{x}.\]

We did another set of experiments where we collected $n=100$ pairs of
$(\hat{x}, \hat{y})$ with dark background (to make $x=0$), with $\delta$ being assigned
randomly, and $P_a=30\%$. We grouped them into $X$ and $Y$: \[X =
\left[\hat{x}_1^\top, \hat{x}_2^\top, \cdots, \hat{x}_n^\top\right]^\top, \quad
Y = \left[\hat{y}_1^\top, \hat{y}_2^\top, \cdots, \hat{y}_n^\top\right]^\top,\]
where both $X$ and $Y$ are $n\times 3$ matrices.  We compute
$H_c$ by solving \[\min_{H_c}\|Y-XH_c\|_2^2.\] This is known as a
non-homogeneous least square problem~\cite{szeliski2010computer}, and it has a
closed-form solution: \[H_c = \left((X^\top X)^{-1}X^\top Y\right)^\top.\]
Plugging $H_c$ back to (\ref{eq:exp_ctrl}) completes our channel model.

\subsubsection{Model validation}

Fig.~\ref{fig:ch_mod_eg} demonstrates the accuracy
of our channel model. In it the left image is the original input to the
projector, the middle image is the estimated output from the camera based on
our channel model (Eq.~\ref{eq:exp_ctrl}), and the image on the right is the
actual image in a ghost captured by the camera. As can be seen, the difference
between the actual and predicted is much less than the actual and original.
While blurring effect is apparent in the actual $y$, we do not model it but the
success rates are still high despite it. As we will see in
Section~\ref{sec:ext_eval}, our channel model is general enough that once
trained on one camera in one environment, it can be transferred to different
environments and different cameras without retraining.

\begin{figure}[t]
    \centering
	\begin{subfigure}{0.3\columnwidth}
	  \centering
	  \includegraphics[width=.5\columnwidth]{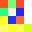}
	  \caption{Input $x_f$}
	\end{subfigure}
	\begin{subfigure}{0.3\columnwidth}
	  \centering
	  \includegraphics[width=.5\columnwidth]{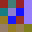}
	  \caption{Estimated $y$}
	\end{subfigure}
	\begin{subfigure}{0.3\columnwidth}
	  \centering
	  \includegraphics[width=.5\columnwidth]{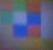}
	  \caption{Actual $y$}
	\end{subfigure}
	\caption{An example of channel model prediction}
	\label{fig:ch_mod_eg}
\end{figure}

\subsection{Optimal Adversarial Projection Patterns}
\label{sec:adv-patterns}

In long-distance, low-resolution GhostImage attacks there are only a few pixels
in the ghost area. A camera-aware attacker's strategy is to simply downsample
attack images into low resolutions, but that does not result in high success
rates.  While (\ref{eq:main}) is abstract, for the rest of this subsection, we
will progressively detail it and show how it can be solved in light of the
channel model to improve attack success rates. We will start with the simplest
case where adversarial perturbations are random noise (Sec.~\ref{sec:adv-cha}).
Then, single-color ghosts will be introduced. Later, we will consider how to
find semi-positive additive noise due to the fact that superposition can only
increase perceived light intensity but not decrease it. Finally, we examine the
optimization problem to find optimal ghost patterns in grids at different
resolutions.

\subsubsection{Single-color ghost}

Let us consider the simplest case first where the random noise $\Delta$ is
drawn from one single Gaussian distribution for all three channels, i.e.,
$\Delta \sim \mathcal{N}(\mu, \sigma^2)$, where the size of $\Delta$ is $w
\times h \times 3$ with $w$ and $h$ representing the width and height of the benign
image $x$. This is because the values of each pixel  that
appear in the ghost area follow Gaussian distributions according to statistics
obtained from our experiments. The adversary needs to find $\mu$ and $\sigma$ such
that when $\Delta$ is added to the benign image $x$, the resulting image $y$
will be classified as the target class $t$. That said, the logits value
(Section~\ref{sec:nn-adv}) of the target class should be as high as possible
compared with the logits values of other classes~\cite{carlini2017towards}.
Such a difference is measured by the loss function $\LL(y, t)$
\begin{equation}
	\LL(y, t) = \max\left\{-\kappa,
	\max_{i: i \ne t} \{\E[Z_i(y)]\} - \E[Z_t(y)]\right\},
	\label{eq:l1}
\end{equation}
where $\E[Z_i(y)]$ is the expectation of logits values at class $i$ of input
$y$. Term $\max_{i: i \ne t} \{\E[Z_i(y)]\}$ is the highest expected logits value
among all the classes except the target class $t$, while $\E[Z_t(y)]$ is the
expected logits value of $t$.  Here, $\kappa$ controls the logits gap between
$\max_{i: i \ne t} \{\E[Z_i(y)]\}$ and $\E[Z_t(y)]$; the larger the $\kappa$
is, the more confident that $\Delta$ is successful. The attacker needs $\LL$ as
low as possible so that the neural network would classify $y$ as Class $t$.
Most importantly, $y$ is computed based on our channel model
(Eq.~\ref{eq:exp_ctrl}),
so that the optimizer finds the optimal
ghost patterns that are resistant to the channel effects.
Unfortunately, due to the complexity of neural networks, the
expectations of logits values $\E[Z_i(y)]$ are   hard to be expressed analytically; we
instead use Monte Carlo methods to approximate it: 
\[
	\hat{\E}[Z_i(y)] = \frac{1}{T} \sum_{j=1}^T Z_i(y_j),
\]
where $T$ is the number of trials, and $y_j$ is of the $j$-th
trial.

Meanwhile, the adversary also needs to minimize the magnitude of $\Delta$ to
reduce the attack power and noticeability, as well as its peak energy
consumption, quantified by $\sigma$. The expectation of the magnitude of
$\Delta$ is
\begin{equation}
	\E[\|\Delta\|_p] = \mu n^{1/p}, \quad \text{with}~~n=3wh.
	\label{eq:e1}
\end{equation}
Putting (\ref{eq:l1}) and (\ref{eq:e1}) together with a tunable
constant $c$, we have our optimization problem for the simplest case
\[
\begin{aligned}
	\mu^*, \sigma^* =~~&\argmin_{\mu,\sigma} & & \E[\|\Delta\|_p] + \sigma  + c
	\cdot \LL(y, t), \\
	& \subjto & & \sigma > \sigma_l,
\end{aligned}
\]
Here,
$\sigma_l$ is the lower bound of the standard deviation $\sigma$, meaning that
the interference generator and the channel environment can provide random noise
with the standard deviation of at least $\sigma_l$. When $\sigma_l=0$, the adversary
is able to manipulate pixels deterministically. Therefore, when we fix $\sigma$
as $\sigma_l$ in the optimization
problem, the attack success rate when deploying $\mu^*$ would be the lower
bound of the attack success rate. In other words, the adversary equipped with
an attack setup that can produce noise with a lower variance (than $\sigma_l^2$) can carry out attacks with
higher success rates. Therefore, we can simplify our formulation by removing the
constraint about $\sigma$, so the optimization problem becomes
\begin{equation}
	\mu^* = \argmin_\mu \E[\|\Delta\|_p]  + c \cdot \LL(y, t).
	\label{eq:m1}
\end{equation}
For the rest of the paper we will simply use $\sigma$ to denote $\sigma_l$.

Since in (\ref{eq:m1}) there is only one variable that the adversary
is able to control, it is infeasible to launch a targeted attack with
such few degrees of freedom. As a result, the adversary needs to manipulate
each channel individually. That is, for each channel, there will be an
independent distribution from which noise will be drawn. This is feasible
because noise can appear in different colors in the ghost areas in which three
channels are perturbed differently when using projectors.
Let us decompose $\Delta$ as $\Delta = [\Delta_R, \Delta_G, \Delta_B]$, where
the dimension of $\Delta_{\{R, G, B\}}$ is $w\times h$, and they follow three
independent Gaussian distributions \[\Delta_R \sim \mathcal{N}\left(\mu_R,
\sigma^2_R\right),~~\Delta_G \sim \mathcal{N}\left(\mu_G,
\sigma^2_G\right),~~\Delta_B \sim \mathcal{N}\left(\mu_B,
\sigma^2_B\right).\] Here, $\mu_{\{R, G, B\}}$ and $\sigma_{\{R, G, B\}}$ are
the means and the standard deviations ($\sigma$) of the three Gaussian distributions,
respectively. The expectation of such $\Delta$ is then
\begin{equation}
	\E[\|\Delta\|_p] = \left[\frac{n}{3}\left(\mu_R^p + \mu_G^p +
	\mu_B^p\right)\right]^{\frac{1}{p}}.
	\label{eq:e2}
\end{equation}
(\ref{eq:e1}) is a special case of (\ref{eq:e2}) when
$\mu=\mu_R=\mu_G=\mu_B$. We denote $\bmmu
= [\mu_R, \mu_G, \mu_B]^\top$. Hence, similar to (\ref{eq:m1}), we have the
optimization problem for single-color perturbation
\begin{equation}
	\bmmu^* = \argmin_{\bmmu} \E[\|\Delta\|_p]  + c \cdot \LL(y, t),
	\label{eq:m2}
\end{equation}
by which the adversary finds the optimal $\bmmu^*$ from which $\Delta$ is drawn.

\subsubsection{Non-negative noise constraint}

\begin{figure}[tb]
	\centering
	\input{./figs/biased_penalty.tikz}
	\caption{Biased penalty}
	\label{fig:biased_penalty}
\end{figure}
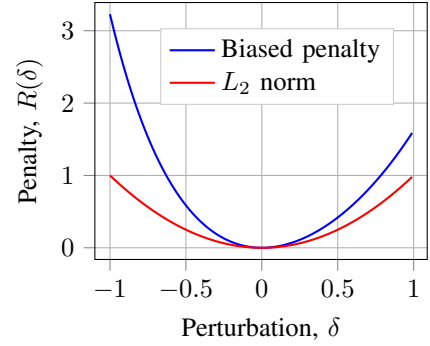

(\ref{eq:m2}) must be solved with the constraint $\Delta \ge 0$
because the adversary can only increase the light. Rather than explicitly place
a constraint in (\ref{eq:m2}), we propose to punish negative values, by intruducing
\emph{biased penalty}
\begin{equation}
	R(\Delta) = e^{-\alpha (\Delta - \omega)} + e^{\beta (\Delta - \omega)} - \eta,
	\label{eq:bia}
\end{equation}
where \[\omega = \frac{\ln \alpha - \ln \beta}{\alpha + \beta},\quad
\eta = \left(\frac{\alpha}{\beta}\right)^\frac{-\alpha}{\alpha+\beta} +
\left(\frac{\alpha}{\beta}\right)^\frac{\beta}{\alpha+\beta}.\] Here
$\omega$ is to center the global minimum at $\Delta$ being zero, and
substracting $\eta$ is to lower the minimum to be zero but it does not
change the optimization results so we will omit it.  An instance of
(\ref{eq:bia}) when $\alpha=2$ and $\beta=1$ is plotted in
Fig.~\ref{fig:biased_penalty} in comparison
with the $L_2$ norm. With the same absolute value, while the $L_p$ norm treats
positive perturbation and negative perturbation equally, the biased
penalty function punishes the negative values more than the
positive one, encouraging the optimization algorithm to find positive
$\Delta$.
We adopt it into our optimization
formulation
\begin{equation}
	\bmmu^* = \argmin_{\bmmu} \E[R(\Delta)]  + c \cdot \LL(y, t),
	\label{eq:m3}
\end{equation}
and in the experiments we set $\alpha=8$ and $\beta=2$.

\subsubsection{Ghost grids}

Since projector's pixels are arranged in grids, the attack patterns are in
grids as well, especially in lower resolutions.
We enable $\Delta$ with patterns in different resolutions.
Such a grid pattern $\bmdelta$ can be composed of several blocks $\Delta_{i, j,
k}$, i.e.,
	$\Delta_{i, j, k} :
	\{1\le i \le N_\text{row},
	1\le j \le N_\text{col},
	1\le k \le N_\text{chn}\}$
where $N_\text{row}$, $N_\text{col}$ and $N_\text{chn}$ is the number of rows, columns, and
channels of a grid pattern, respectively, in terms of blocks. In a word,
$\Delta_{i, j, k}$ is the perturbation block at $i$-th row,
$j$-th column and $k$-th channel.
A block $\Delta_{i, j, k}$ is a random matrix
and its size is $\frac{w}{N_\text{col}} \times \frac{h}{N_\text{row}}$,
so that the size of $\bmdelta$ is still $w\times h\times 3$. Besides,
the elements in the random matrix $\Delta_{i, j, k}$ is i.i.d. drawn from a
Gaussian distribution, i.e., $\Delta_{i, j, k} \sim \mathcal{N}(\mu_{i, j, k}, \sigma^2)$.

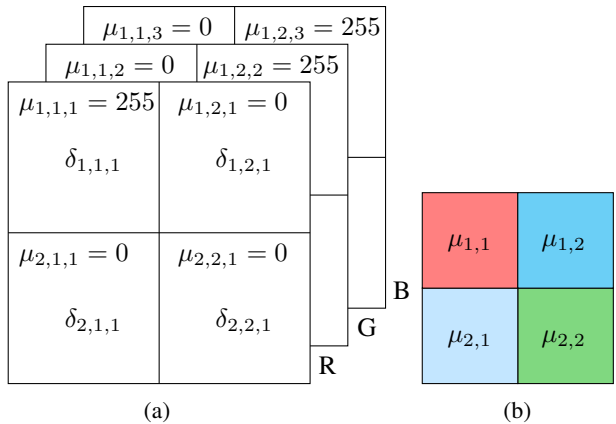
\begin{figure}[tb]
	\centering
	\begin{subfigure}[b]{.45\columnwidth}
		\centering
		\input{./figs/grid.tikz}
		\caption{}
		\label{fig:grid}
	\end{subfigure}
	\hspace{1cm}
	\begin{subfigure}[b]{.35\columnwidth}
		\centering
		\input{./figs/grid_merged.tikz}
		\caption{}
		\label{fig:grid_merged}
	\end{subfigure}
	\caption{A grid pattern when $N_\text{row}=N_\text{col}=2$ and $N_\text{chn}=3$.
	(a) $\bmmu$ is a three-dimensional matrix. (b) The resulting perturbation
	pattern.}
	\label{fig:grid-grid}
\end{figure}

The adversary can find the optimal grid pattern $\bmdelta$ by
solving the optimization problem as in (\ref{eq:m3}) in which
\begin{equation}
	\E[R(\Delta)] = \sum_{i=1}^{N_{row}} \sum_{j=1}^{N_{col}} \sum_{k=1}^{N_{chn}}
	e^{-\alpha (\mu_{i,j,k}-\omega)} + e^{\beta (\mu_{i, j, k}-\omega)},
	\label{eq:pen}
\end{equation}
where $\mu_{i, j, k}$ is the mean of the block $\Delta_{i, j, k}$.
See Fig.~\ref{fig:grid} for an illustration of the dimensionality
of $\mu_{i, j, k}$, and
Fig.~\ref{fig:grid_merged} for the resulting pattern in color.

%% file: figs/illuminance.tikz
\begin{tikzpicture}

\definecolor{color0}{rgb}{1, 0, 0}
\definecolor{color1}{rgb}{0, 1, 0}
\definecolor{color2}{rgb}{0, 0, 1}

\begin{groupplot}[group style={group size=2 by 1, horizontal sep=.5cm}]
\nextgroupplot[
width=.5\columnwidth, height=4cm,
legend cell align={left},
legend style={at={(.5, 1.05)}, anchor=south},
tick align=outside,
tick pos=left,
x grid style={white!69.01960784313725!black},
xlabel={$T_d$},
xmajorgrids,
xmin=-0.05, xmax=1.05,
xtick style={color=black},
y grid style={white!69.01960784313725!black},
ylabel={Illuminance (lux), $\I$},
ymajorgrids,
ymin=-59.508678961411, ymax=1259.97660376007,
ytick style={color=black}
]
\addplot [semithick, blue!50, mark=o, mark size=1.5, mark options={solid}, only marks]
table {%
0 333
0.1 503
0.2 700
0.3 900
0.4 1170
0.5 1177
0.6 1180
0.7 1178
0.8 1180
0.9 1180
1 1190
};
\addlegendentry{$T_a=1$}
\addplot [semithick, green!50, mark=triangle, mark size=1.5, mark options={solid}, only marks]
table {%
0 4
0.1 13
0.2 69
0.3 185
0.4 360
0.5 575
0.6 820
0.7 1040
0.8 1160
0.9 1160
1 1160
};
\addlegendentry{$T_a=0.5$}
\addplot [semithick, red!50, mark=square, mark size=1.5, mark options={solid}, only marks]
table {%
0 5
0.1 5
0.2 5
0.3 5
0.4 5
0.5 18
0.6 88
0.7 215
0.8 405
0.9 625
1 855
};
\addlegendentry{$T_a=0$}
\begin{pgfonlayer}{fg}
\addplot [black!65, thick]
table {%
0 0.467924798656182
0.1 1.14147251842015
0.2 2.78230014740295
0.3 6.7684450304278
0.4 16.3872893948673
0.5 39.2263911840582
0.6 91.4374992173617
0.7 201.102862048224
0.8 395.384233688647
0.9 654.396716137282
1 894.462418934312
};
\addlegendentry{Fit: $\I$}
\addplot [black!65, thick]
table {%
0 13.2190370549237
0.1 31.761007686639
0.2 74.6747803464992
0.3 167.269325345628
0.4 339.989950983076
0.5 589.291271669944
0.6 842.347027095143
0.7 1022.18549702437
0.8 1120.16623411167
0.9 1165.95505614204
1 1185.81419032398
};
\addplot [black!65, thick]
table {%
0 289.563572811766
0.1 524.437107787305
0.2 785.462076561601
0.3 986.659615997931
0.4 1102.34594082624
0.5 1157.97214931296
0.6 1182.41766106028
0.7 1192.73365710878
0.8 1197.01229510865
0.9 1198.77413365032
1 1199.49745869188
};
\end{pgfonlayer}

\nextgroupplot[
width=.5\columnwidth, height=4cm,
legend cell align={left},
legend style={at={(.5,1.05)}, anchor=south},
tick align=outside,
tick pos=left,
x grid style={white!69.01960784313725!black},
xlabel={$T_a$},
xmajorgrids,
xmin=-0.05, xmax=1.05,
xtick style={color=black},
y grid style={white!69.01960784313725!black},
ymajorgrids,
ymin=-59.508678961411, ymax=1259.97660376007,
yticklabel=\empty
]
\addplot [semithick, blue!50, mark=o, mark size=1.5, mark options={solid}, only marks]
table {%
0 880
0.1 1060
0.2 1170
0.3 1200
0.4 1200
0.5 1200
0.6 1200
0.7 1200
0.8 1200
0.9 1200
1 1200
};
\addlegendentry{$T_d=1$}
\addplot [semithick, green!50, mark=triangle, mark size=1.5, mark options={solid}, only marks]
table {%
0 17
0.1 60
0.2 133
0.3 270
0.4 415
0.5 580
0.6 770
0.7 945
0.8 1120
0.9 1180
1 1190
};
\addlegendentry{$T_d=0.5$}
\addplot [semithick, red!50, mark=square, mark size=1.5, mark options={solid}, only marks]
table {%
0 4.4
0.1 4.3
0.2 4.2
0.3 4.5
0.4 4.3
0.5 5
0.6 10
0.7 36
0.8 100
0.9 200
1 330
};
\addlegendentry{$T_d=0$}
\begin{pgfonlayer}{fg}
\addplot [black!65, thick]
table {%
0 0.467924798656182
0.1 0.914423515941336
0.2 1.78634148185803
0.3 3.4872297905456
0.4 6.79845216953555
0.5 13.2190370549237
0.6 25.5733807939282
0.7 48.9972418992649
0.8 92.1919099365591
0.9 167.910764447927
1 289.563572811766
};
\addlegendentry{Fit: $\I$}
\addplot [black!65, thick]
table {%
0 39.2263911840582
0.1 74.363850937781
0.2 137.254369283627
0.3 241.902149101848
0.4 396.564533191206
0.5 589.291271669944
0.6 784.254143226216
0.7 944.013852290599
0.8 1053.82485754827
0.9 1120.49712680614
1 1157.97214931296
};
\addplot [black!65, thick]
table {%
0 894.462418934312
0.1 1021.51060866282
0.2 1101.54479524536
0.3 1147.53502659646
0.4 1172.57719104161
0.5 1185.81419032398
0.6 1192.70145609723
0.7 1196.25548644095
0.8 1198.08166312184
0.9 1199.01795558314
1 1199.49745869188
};
\end{pgfonlayer}
\end{groupplot}

\end{tikzpicture}

%% file: figs/ill-vs-dim.tikz
\begin{tikzpicture}

\begin{groupplot}[group style={group size=1 by 2, vertical sep=.5cm}]

\nextgroupplot[
width=8cm,height=3.5cm,
legend cell align={left},
legend style={at={(0.97,0.97)}, anchor=north east, draw=white!80.0!black},
tick align=outside,
tick pos=left,
x grid style={white!69.01960784313725!black},
xticklabel=\empty,
xmajorgrids,
xmin=-0.05, xmax=1.05,
xtick style={color=black},
y grid style={white!69.01960784313725!black},
ylabel={Dimming Ratio},
ymajorgrids,
ymin=-0.05, ymax=1.05,
yticklabel style={/pgf/number format/fixed,
/pgf/number format/precision=5},
]
\addplot [red!50, mark=square, mark size=1.5, only marks]
table {%
0.00040956716498605 0.912418673584784
0.000638978862542847 0.913069754680312
0.000996763268162254 0.914946183602076
0.00155457080607247 0.915691473389663
0.00242378060897396 0.917492484571422
0.00377715459274024 0.915147043643637
0.00588176068092146 0.913930964327536
0.00914827105343449 0.923223978879986
0.0142029613726911 0.913907221705683
0.0219885395529197 0.774416905965632
0.0338951641591782 0.513219752185432
0.0519069479759904 0.388661260205042
0.0787103371705352 0.272679695144735
0.117636992125077 0.206616802082139
0.172216391733878 0.138167331142664
0.245085013132372 0.0976252172545124
0.336261302595648 0.0647910757197919
0.441518887561839 0.0450020447856949
0.552307909574325 0.0493427547104029
0.658136315048018 0.0172211457630656
0.750260105595118 0.00823069093188013
0.000800078748926495 0.883051716635998
0.00124795541509161 0.865133909022082
0.00194606098508556 0.880535506583118
0.00303350030301862 0.884992263333355
0.00472571260239548 0.88377593008547
0.00735494702158044 0.889648356534446
0.0114302031264707 0.886645740148304
0.0177231760900623 0.846736805164871
0.0273848109501257 0.605873746894256
0.0420877279156188 0.415601137157317
0.064163876405174 0.275761665213542
0.096651210941591 0.177384809876829
0.143072723480081 0.117339973141635
0.206689012837442 0.0853095257049432
0.289050497374996 0.0612516475411595
0.388172630998263 0.0359260149431238
0.497500020833125 0.0205264488363639
0.607066988071704 0.0309037521087448
0.706822221093568 0.0248718564990857
0.790012373426398 0.017068151862888
0.854457671063035 0.0102384025862444
0.00156235095032896 1
0.00243590026547166 0.852903792483898
0.00379601578870736 0.862291904489654
0.0059110688562438 0.872318020104033
0.00919370536728809 0.864082377502355
0.0142731378704689 0.869141206513706
0.0220963221533702 0.737388065282307
0.0340592776978742 0.473770433731994
0.0521535630784178 0.332826485491417
0.0790736768698401 0.222372380550806
0.11815697780927 0.152846365810828
0.172930349896139 0.0965476577268713
0.246011283551052 0.0637659531056556
0.337378162829918 0.0448618744406359
0.442752145401444 0.0235915324414995
0.553543903151118 0.0335983492402552
0.659260388451386 0.0271150898835989
0.751195782230862 0.0210875111142768
0.82491373183596 0.0134378161053949
0.880271109710736 0.00404322883913436
0.919827087827188 0.00919969462163593
0.0030486594302237 0.88370190897028
0.00474928783552836 0.865004403811973
0.00739154134428197 0.87914065929071
0.0114868391295966 0.880725320592052
0.0178104316405254 0.821031925446644
0.0275183005334425 0.585026201316357
0.0422897718420338 0.388433483394321
0.0644647658375897 0.273895647494944
0.097088640986499 0.175728412846358
0.143686832758829 0.126314303304503
0.207510058559636 0.0802697415774319
0.290079081665547 0.0507031942460073
0.389360766050778 0.0308417927533201
0.49875000260416 0.0425915973210705
0.608259030746514 0.0338228249377772
0.707857271424482 0.0259744283075028
0.790840634786937 0.0170706911807333
0.855078368512177 0.00916173181984912
0.902031195702446 0.0123570824304471
0.934929510471975 0.0042892887383412
0.957303355781686 0.00585998379153419
0.00594052219834037 0.988663977275137
0.00923936322372884 0.848858532190488
0.0143436560634881 0.862658074122942
0.022204621083536 0.690444838969794
0.0342241576923959 0.46920727956404
0.0524012851141689 0.298744142269853
0.0794385491839784 0.219008165303497
0.118678952816415 0.141453589332123
0.173646647019005 0.0982821387810842
0.246939909578662 0.0554475287422213
0.338496840797048 0.0434879765204515
0.44398610945538 0.0229777793182973
0.554779235107215 0.0333926644947882
0.660382673079216 0.0256386035224655
0.75212911143957 0.0195217677308821
0.825634714382587 0.0114767009334913
0.880797077977883 0.00584131980537147
0.920195040759213 0.00957817994647373
0.947349881564697 0.0069810926202218
0.965610192550819 0.00608204713710323
0.977686561300968 0.00719680767117765
0.0115437524839223 0.363864059650608
0.0178981089440968 0.555861881931624
0.0276524223228231 0.348236081903666
0.042492742663481 0.241223768370536
0.0647669686067856 0.209627671251278
0.0975278370407926 0.137509139957169
0.144303134090519 0.0980800090806006
0.208333509270146 0.0618965073079663
0.29110982743388 0.0491958551730539
0.390550216371675 0.0283062838848161
0.5 0.0261930635739807
0.609449783628325 0.0212709768285977
0.70889017256612 0.0127887664641434
0.791666490729855 0.0155534487681833
0.855696865909481 0.00383462387814526
0.902472162959208 0
0.935233031393214 0.00555958249043852
0.957507257336519 0.00508625364407984
0.972347577677177 0.00432496615406727
0.982101891055903 0.0058244333417004
0.988456247516078 0.00450728917535781
0.0223134386990322 0.811224825996419
0.0343898074491805 0.457397166197464
0.0526501184353029 0.303792814009823
0.0798049592407874 0.194809354999301
0.119202922022118 0.149018217193188
0.174365285617413 0.0924373908966217
0.24787088856043 0.0632592322296315
0.339617326920784 0.042634257860871
0.445220764892785 0.0204905174888533
0.55601389054462 0.0311829501058324
0.661503159202952 0.0271748908188551
0.753060090421338 0.0173210679202771
0.826353352980995 0.01079400533079
0.881321047183586 0.0133456388676115
0.920561450816022 0.0113266272488358
0.947598714885831 0.0110125136313756
0.965775842307604 0.00962909326927142
0.977795378916464 0.0101232445219614
0.985656343936512 0.00932589871854598
0.990760636776271 0.0111252593437057
0.99405947780166 0.00998802584670063
0.0426966442183139 0.388402122818932
0.0650704895280251 0.275062591010739
0.097968804297554 0.203720075250145
0.144921631487823 0.164091988820399
0.209159365213063 0.109299984979935
0.292142728575519 0.0713541966226819
0.391740969253485 0.0584345283236148
0.50124999739584 0.0385693178541596
0.610639233949222 0.0368234098696428
0.709920918334453 0.0349571382192602
0.792489941440364 0.0297920387560847
0.856313167241171 0.0239169460233711
0.902911359013501 0.0164942660298566
0.93553523416241 0.0163858371578634
0.957710228157966 0.0186592884247354
0.972481699466557 0.0178279157621934
0.982189568359475 0.0184443351691332
0.988513160870403 0.0174788865243608
0.992608458655718 0.0181256507795515
0.995250712164472 0.0168603086972525
0.996951340569776 0.0181204451779687
0.0801729121728124 0.165497374281865
0.119728890289265 0.151888281187807
0.17508626816404 0.121314132535378
0.248804217769139 0.107132042369534
0.340739611548615 0.0741171014042301
0.446456096848882 0.0496545829418023
0.557247854598556 0.0379776567962113
0.662621837170082 0.0184510643614232
0.753988716448948 0.0297609321124801
0.827069650103861 0.0335018551621348
0.881843022190731 0.0264636278903944
0.92092632313016 0.0232656109960589
0.947846436921582 0.0273337251500769
0.965940722302126 0.0235125596565116
0.97790367784663 0.0268888366635852
0.985726862129531 0.026698006927513
0.990806294632712 0.0261986500732403
0.994088931143756 0.0187692408874357
0.996203984211293 0.0160124304687162
0.997564099734528 0.0146276134817971
0.998437649049671 0.0154333390341019
0.145542328936965 0.0535879862841286
0.209987626573603 0.0405706812139108
0.293177778906432 0.0374600168534526
0.392933011928297 0.0308315085160468
0.502499979166875 0.0316930990609476
0.611827369001737 0.0255063050627269
0.710949502625004 0.0197219929429818
0.793310987162558 0.0121155933033618
0.856927276519919 0.00797688611324521
0.903348789058409 0.00738509808940456
0.935836123594826 0.0112411792033424
0.957912272084381 0.00848500361408413
0.972615189049874 0.0117472652499051
0.982276823909938 0.0020925248703964
0.988569796873529 0.00266310959022901
0.99264505297842 0.000646510323406254
0.995274287397605 0.00222152221693621
0.996966499696981 0.00087593769072658
0.998053939014914 0.00144944262608129
0.998752044584909 0.00229579726391043
0.999199921251073 0.000802170507321437
0.249739894404882 0.0278919941783599
0.341863684951982 0.0155342769184515
0.447692090425675 0.0202835630844637
0.558481112438161 0.0197889039682047
0.663738697404353 0.0237772835418812
0.754914986867628 0.0223553925144211
0.827783608266122 0.0215800118103673
0.882363007874923 0.00597488792403276
0.921289662829465 0.0103670190351075
0.94809305202401 0.0112416870669114
0.966104835840822 0.013076598141905
0.97801146044708 0.0159434879892171
0.985797038627309 0.0106763079486616
0.990851728946565 0.013023653364831
0.994118239319079 0.0108532984024771
0.99622284540726 0.0137667847322499
0.997576219391026 0.0150559964022945
0.998445429193927 0.0134249925502763
0.999003236731838 0.0141385408647977
0.999361021137457 0.0138461384149146
0.999590432835014 0.0138797843763645
};
\addlegendentry{Paper}
\begin{pgfonlayer}{fg}
\addplot [thick, black!65]
table {%
0 1
0.0526315789473684 0.27536231884058
0.105263157894737 0.159663865546219
0.157894736842105 0.112426035502959
0.210526315789474 0.0867579908675799
0.263157894736842 0.0706319702602231
0.315789473684211 0.0595611285266458
0.368421052631579 0.0514905149051491
0.421052631578947 0.045346062052506
0.473684210526316 0.0405117270788913
0.526315789473684 0.0366088631984586
0.578947368421053 0.0333919156414763
0.631578947368421 0.0306946688206785
0.684210526315789 0.0284005979073244
0.736842105263158 0.0264255910987483
0.789473684210526 0.0247074122236671
0.842105263157895 0.0231990231990232
0.894736842105263 0.0218642117376295
0.947368421052632 0.0206746463547334
1 0.0196078431372549
};
\end{pgfonlayer}
\addlegendentry{Fit: $\gamma(\I)=\frac{\I_\text{env}}{\I+\I_\text{env}}$}

\nextgroupplot[
width=8cm,height=3.5cm,
legend cell align={left},
legend style={at={(0.97,0.03)}, anchor=south east, draw=white!80.0!black},
tick align=outside,
tick pos=left,
x grid style={white!69.01960784313725!black},
xmajorgrids,
xmin=-0.05, xmax=1.05,
y grid style={white!69.01960784313725!black},
ylabel={RGB Values},
ymajorgrids,
xlabel={Illuminance (normalized)},
ymin=-0.00722735993236432, ymax=0.151774558579651,
yticklabel style={/pgf/number format/fixed,
/pgf/number format/precision=5},
]
\addplot [blue!50, mark=o, mark size=1.5, only marks]
table {%
0.00040956716498605 0.0108209135874729
0.000638978862542847 0.0110887826661897
0.000996763268162254 0.0114442312764428
0.00155457080607247 0.0127279362085722
0.00242378060897396 0.0132534856581291
0.00377715459274024 0.01328072780614
0.00588176068092146 0.0127279362085722
0.00914827105343449 0.0123203156331631
0.0142029613726911 0.0140075945616868
0.0219885395529197 0.00308660040393527
0.0338951641591782 0.0281508217807115
0.0519069479759904 0.0921940045596274
0.0787103371705352 0.120404079818771
0.117636992125077 0.140546649468706
0.172216391733878 0.134836203875322
0.245085013132372 0.134543260458781
0.336261302595648 0.130517636942376
0.441518887561839 0.130599507907088
0.552307909574325 0.133007439201957
0.658136315048018 0.124929817145211
0.750260105595118 0.12351741654834
0.000800078748926495 0.00974914823124751
0.00124795541509161 0.0131257293778023
0.00194606098508556 0.0163999176232129
0.00303350030301862 0.015951686736976
0.00472571260239548 0.0164228964111902
0.00735494702158044 0.0158697435119247
0.0114302031264707 0.0161764886533202
0.0177231760900623 0.0146692824909584
0.0273848109501257 0.0120335143454839
0.0420877279156188 0.0530328748414789
0.064163876405174 0.0911054748046261
0.096651210941591 0.10509413715735
0.143072723480081 0.107279506606402
0.206689012837442 0.107205439758361
0.289050497374996 0.104958359979334
0.388172630998263 0.0978646347059185
0.497500020833125 0.0996837164937838
0.607066988071704 0.111768969242386
0.706822221093568 0.111456732315185
0.790012373426398 0.109259223129089
0.854457671063035 0.104706243654639
0.00156235095032896 0.0477177655657804
0.00243590026547166 0.00995935355900235
0.00379601578870736 0.0103418275362476
0.0059110688562438 0.0100232316991657
0.00919370536728809 0.0104218919924705
0.0142731378704689 0.0121302709401431
0.0220963221533702 0.00718542364430569
0.0340592776978742 0.0297254468398747
0.0521535630784178 0.0785493014231674
0.0790736768698401 0.0991174844730595
0.11815697780927 0.109524707616601
0.172930349896139 0.109716052995733
0.246011283551052 0.106885760016186
0.337378162829918 0.099488324535637
0.442752145401444 0.0947955213041546
0.553543903151118 0.111749314430028
0.659260388451386 0.111547129999964
0.751195782230862 0.110189647261152
0.82491373183596 0.107684742590605
0.880271109710736 0.103816068531706
0.919827087827188 0.108563500579889
0.0030486594302237 0.0107003110807617
0.00474928783552836 0.0139017331642441
0.00739154134428197 0.0164161761996119
0.0114868391295966 0.0166417729796912
0.0178104316405254 0.0105573078687897
0.0275183005334425 0.00709582082326202
0.0422897718420338 0.0518093627721957
0.0644647658375897 0.0970322678546266
0.097088640986499 0.10581413918064
0.143686832758829 0.110290450434827
0.207510058559636 0.104617002135293
0.290079081665547 0.10363505638113
0.389360766050778 0.0950345585073904
0.49875000260416 0.113987144885594
0.608259030746514 0.114183403967815
0.707857271424482 0.112647510450652
0.790840634786937 0.110362783034717
0.855078368512177 0.106217929958053
0.902031195702446 0.110816288925742
0.934929510471975 0.107444621482276
0.957303355781686 0.106402844166965
0.00594052219834037 0.0482294410301434
0.00923936322372884 0.0109437561647102
0.0143436560634881 0.0148976251639406
0.022204621083536 0
0.0342241576923959 0.0394397655874585
0.0524012851141689 0.0650764333741604
0.0794385491839784 0.105304342485105
0.118678952816415 0.108387257611723
0.173646647019005 0.110030168691763
0.246939909578662 0.0988538787543762
0.338496840797048 0.10017067892202
0.44398610945538 0.0949912023036596
0.554779235107215 0.112259038865224
0.660382673079216 0.111743461342525
0.75212911143957 0.111927219386004
0.825634714382587 0.10871799318585
0.880797077977883 0.105331801414135
0.920195040759213 0.107267077827999
0.947349881564697 0.107121400983463
0.965610192550819 0.107668484014206
0.977686561300968 0.107519194152693
0.0115437524839223 0.0352213514851306
0.0178981089440968 0.0261989254887509
0.0276524223228231 0.0623968032025782
0.042492742663481 0.0979028604255411
0.0647669686067856 0.141292448433215
0.0975278370407926 0.133439700553153
0.144303134090519 0.135239777871716
0.208333509270146 0.129063397608905
0.29110982743388 0.128201187237379
0.390550216371675 0.120474822691192
0.5 0.125902369055232
0.609449783628325 0.133985193856426
0.70889017256612 0.132168207618408
0.791666490729855 0.137713538335917
0.855696865909481 0.132168496659766
0.902472162959208 0.128228212604371
0.935233031393214 0.131984810876626
0.957507257336519 0.132175722693721
0.972347577677177 0.132984026851942
0.982101891055903 0.13215809117087
0.988456247516078 0.130994266142057
0.0223134386990322 0.0486823688384511
0.0343898074491805 0.0315786355080082
0.0526501184353029 0.0774244969777113
0.0798049592407874 0.0969654993008812
0.119202922022118 0.117011456876836
0.174365285617413 0.106489990136464
0.24787088856043 0.105225867756353
0.339617326920784 0.100404874682506
0.445220764892785 0.0936759196031462
0.55601389054462 0.109896125761895
0.661503159202952 0.111471256643435
0.753060090421338 0.109597546038869
0.826353352980995 0.108095831662313
0.881321047183586 0.110429045766086
0.920561450816022 0.11113069366313
0.947598714885831 0.109693796811151
0.965775842307604 0.111019123698862
0.977795378916464 0.109595016926985
0.985656343936512 0.111047666532985
0.990760636776271 0.111057638459843
0.99405947780166 0.110443497833996
0.0426966442183139 0.0506876655213403
0.0650704895280251 0.0842976836948157
0.097968804297554 0.117229683102281
0.144921631487823 0.138754087225456
0.209159365213063 0.132273924495171
0.292142728575519 0.124322035429244
0.391740969253485 0.126211715568851
0.50124999739584 0.119831272107148
0.610639233949222 0.125892541649053
0.709920918334453 0.133305440842267
0.792489941440364 0.132437016081539
0.856313167241171 0.134398956560697
0.902911359013501 0.129926258323488
0.93553523416241 0.12997257720114
0.957710228157966 0.134205804673076
0.972481699466557 0.129416244846935
0.982189568359475 0.132605093631335
0.988513160870403 0.131866303919762
0.992608458655718 0.134911860450832
0.995250712164472 0.131085241909552
0.996951340569776 0.133776361475123
0.0801729121728124 0.0521822983846201
0.119728890289265 0.10438071082496
0.17508626816404 0.127442164630732
0.248804217769139 0.140586826217496
0.340739611548615 0.132466281519057
0.446456096848882 0.125031776484318
0.557247854598556 0.122905010170643
0.662621837170082 0.117265813272057
0.753988716448948 0.133293951448278
0.827069650103861 0.139919646502419
0.881843022190731 0.134236370796706
0.92092632313016 0.138857130469656
0.947846436921582 0.139127745441276
0.965940722302126 0.139091326230142
0.97790367784663 0.139517300931797
0.985726862129531 0.139882143386192
0.990806294632712 0.139445546414623
0.994088931143756 0.130826333112939
0.996203984211293 0.13142529906748
0.997564099734528 0.129973010763178
0.998437649049671 0.130226355513645
0.145542328936965 0.0602486478283961
0.209987626573603 0.079253189390737
0.293177778906432 0.0975271066598742
0.392933011928297 0.102836435108409
0.502499979166875 0.106790593148997
0.611827369001737 0.108417317912977
0.710949502625004 0.108248228718427
0.793310987162558 0.105111190597485
0.856927276519919 0.104430859500609
0.903348789058409 0.106783511635721
0.935836123594826 0.109407429085509
0.957912272084381 0.109299688919238
0.972615189049874 0.109850457227298
0.982276823909938 0.10212271973466
0.988569796873529 0.102416891576973
0.99264505297842 0.10310191959592
0.995274287397605 0.102870686509356
0.996966499696981 0.10260101092215
0.998053939014914 0.102148733456899
0.998752044584909 0.103094476780946
0.999199921251073 0.102091647788653
0.249739894404882 0.0885049697048526
0.341863684951982 0.0999278119207882
0.447692090425675 0.117672638983731
0.558481112438161 0.127272858655163
0.663738697404353 0.13493635670594
0.754914986867628 0.141290569664387
0.827783608266122 0.144547198647286
0.882363007874923 0.132353772170375
0.921289662829465 0.139179122542697
0.94809305202401 0.139045296393848
0.966104835840822 0.142660047619564
0.97801146044708 0.141630193260278
0.985797038627309 0.140863438797299
0.990851728946565 0.13918345816307
0.994118239319079 0.141545865444022
0.99622284540726 0.140086712407462
0.997576219391026 0.143157849098733
0.998445429193927 0.139392073763355
0.999003236731838 0.141136727401482
0.999361021137457 0.143247957742153
0.999590432835014 0.140585670052064
};
\addlegendentry{Flare}
\begin{pgfonlayer}{fg}
\addplot [thick, black!65]
table {%
0 0
0.0526315789473684 0.0869565217391304
0.105263157894737 0.100840336134454
0.157894736842105 0.106508875739645
0.210526315789474 0.10958904109589
0.263157894736842 0.111524163568773
0.315789473684211 0.112852664576802
0.368421052631579 0.113821138211382
0.421052631578947 0.114558472553699
0.473684210526316 0.115138592750533
0.526315789473684 0.115606936416185
0.578947368421053 0.115992970123023
0.631578947368421 0.116316639741519
0.684210526315789 0.116591928251121
0.736842105263158 0.11682892906815
0.789473684210526 0.11703511053316
0.842105263157895 0.117216117216117
0.894736842105263 0.117376294591484
0.947368421052632 0.117519042437432
1 0.117647058823529
};
\end{pgfonlayer}
\addlegendentry{Fit: $\|y_f\|=\gamma \rho \mathcal{I}$}

\end{groupplot}

\end{tikzpicture}

%% file: figs/biased_penalty.tikz
\begin{tikzpicture}

\definecolor{color0}{rgb}{0.12156862745098,0.466666666666667,0.705882352941177}
\definecolor{color1}{rgb}{1,0.498039215686275,0.0549019607843137}

\begin{axis}[
height=5cm, width=6cm,
xlabel={Perturbation, $\delta$},
ylabel={Penalty, $R(\delta)$},
xmin=-1.0995, xmax=1.0895,
ymin=-0.161421556793078, ymax=3.38985269265464,
tick align=outside,
tick pos=left,
xmajorgrids,
ymajorgrids,
x grid style={white!69.01960784313725!black},
y grid style={white!69.01960784313725!black},
legend cell align={left},
legend entries={{Biased penalty},{$L_2$ norm}},
legend style={draw=white!80.0!black, at={(.9, 0.9)}, anchor=north east}
]
\addlegendimage{no markers, blue, thick}
\addlegendimage{no markers, red, thick}
\addplot [thick, blue]
table {%
-1 3.22843113586156
-0.99 3.14091789281804
-0.98 3.05527658359936
-0.97 2.9714715389634
-0.96 2.8894678100118
-0.95 2.8092311540673
-0.94 2.73072802083227
-0.93 2.65392553882264
-0.92 2.57879150207196
-0.91 2.50529435710018
-0.9 2.43340319014198
-0.89 2.3630877146294
-0.88 2.29431825892397
-0.87 2.22706575429313
-0.86 2.16130172312641
-0.85 2.09699826738639
-0.84 2.03412805728997
-0.83 1.97266432021534
-0.82 1.9125808298302
-0.81 1.85385189543688
-0.8 1.79645235153006
-0.79 1.7403575475629
-0.78 1.68554333791754
-0.77 1.6319860720758
-0.76 1.57966258498625
-0.75 1.52855018762373
-0.74 1.47862665773756
-0.73 1.42987023078468
-0.72 1.38225959104407
-0.71 1.33577386290897
-0.7 1.29039260235323
-0.69 1.24609578856855
-0.68 1.20286381576913
-0.67 1.1606774851604
-0.66 1.11951799706875
-0.65 1.07936694322891
-0.64 1.04020629922605
-0.63 1.00201841708937
-0.62 0.964786018034404
-0.61 0.928492185350938
-0.6 0.893120357433781
-0.59 0.858654320953536
-0.58 0.825078204164643
-0.57 0.792376470347977
-0.56 0.760533911385373
-0.55 0.729535641463468
-0.54 0.699367090904348
-0.53 0.670014000120477
-0.52 0.641462413691499
-0.51 0.613698674560501
-0.5 0.586709418347398
-0.49 0.560481567777159
-0.48 0.535002327220589
-0.47 0.510259177345498
-0.46 0.486239869876066
-0.45 0.462932422458302
-0.44 0.440325113629504
-0.429999999999999 0.418406477889705
-0.419999999999999 0.397165300873087
-0.409999999999999 0.37659061461742
-0.399999999999999 0.356671692929612
-0.389999999999999 0.337398046845479
-0.379999999999999 0.318759420181905
-0.369999999999999 0.300745785179583
-0.359999999999999 0.283347338234572
-0.349999999999999 0.266554495716926
-0.339999999999999 0.250357889874718
-0.329999999999999 0.234748364821764
-0.319999999999999 0.219716972607449
-0.309999999999999 0.205254969367027
-0.299999999999999 0.19135381155085
-0.289999999999999 0.178005152230971
-0.279999999999999 0.165200837483636
-0.269999999999999 0.152932902846172
-0.259999999999999 0.141193569846839
-0.249999999999999 0.129975242606221
-0.239999999999999 0.11927050450877
-0.229999999999999 0.109072114943148
-0.219999999999999 0.0993730061100226
-0.209999999999999 0.0901662798960237
-0.199999999999999 0.0814452048125665
-0.189999999999999 0.0732032129982967
-0.179999999999999 0.0654338972839261
-0.169999999999999 0.0581310083182489
-0.159999999999999 0.0512884517541634
-0.149999999999999 0.0449002854935401
-0.139999999999999 0.0389607169898025
-0.129999999999999 0.0334641006071088
-0.119999999999999 0.0284049350350453
-0.109999999999999 0.0237778607577697
-0.0999999999999992 0.0195776575765501
-0.0899999999999992 0.0157992421846831
-0.0799999999999992 0.0124376657937839
-0.0699999999999992 0.00948811181046461
-0.0599999999999992 0.0069458935624378
-0.0499999999999992 0.00480645207310126
-0.0399999999999992 0.00306535388367779
-0.0299999999999992 0.00171828892200487
-0.0199999999999992 0.00076106841708401
-0.00999999999999915 0.000189622858520089
8.60422844084496e-16 0
0.0100000000000009 0.000188362905972017
0.0200000000000009 0.000750988040707945
0.0300000000000009 0.00168426339894867
0.0400000000000009 0.00298468667734486
0.0500000000000009 0.00464886348592142
0.0600000000000009 0.00667350559881652
0.0700000000000009 0.00905542924355274
0.0800000000000009 0.011791553428115
0.0900000000000009 0.0148788983051313
0.100000000000001 0.0183145835724572
0.110000000000001 0.0220958269094838
0.120000000000001 0.026219942448507
0.130000000000001 0.0306843392804974
0.140000000000001 0.0354865199946395
0.150000000000001 0.0406240792510051
0.160000000000001 0.0460947023857574
0.170000000000001 0.0518961640482725
0.180000000000001 0.0580263268695993
0.190000000000001 0.0644831401616774
0.200000000000001 0.0712646386467473
0.210000000000001 0.0783689412163984
0.220000000000001 0.0857942497197188
0.230000000000001 0.0935388477800114
0.240000000000001 0.101601099639553
0.250000000000001 0.109979449031898
0.260000000000001 0.11867241808122
0.270000000000001 0.127678606228205
0.280000000000001 0.136996689182014
0.290000000000001 0.146625417897865
0.300000000000001 0.156563617579745
0.310000000000001 0.166810186707837
0.320000000000001 0.177364096090195
0.330000000000001 0.188224387938258
0.340000000000001 0.199390174965766
0.350000000000001 0.210860639510671
0.360000000000001 0.222635032679652
0.370000000000001 0.234712673514815
0.380000000000001 0.247092948182208
0.390000000000001 0.259775309181773
0.400000000000001 0.272759274578356
0.410000000000001 0.286044427253413
0.420000000000001 0.299630414177071
0.430000000000001 0.313516945700175
0.440000000000001 0.327703794865998
0.450000000000001 0.342190796741277
0.460000000000001 0.356977847766245
0.470000000000001 0.372064905123344
0.480000000000001 0.387451986124313
0.490000000000001 0.403139167615331
0.500000000000001 0.419126585399939
0.510000000000001 0.435414433679423
0.520000000000001 0.452002964510402
0.530000000000001 0.46889248727931
0.540000000000001 0.486083368193533
0.550000000000001 0.503576029788903
0.560000000000001 0.521370950453311
0.570000000000001 0.53946866396617
0.580000000000001 0.55786975905349
0.590000000000001 0.576574878958304
0.600000000000001 0.59558472102623
0.610000000000001 0.614900036305915
0.620000000000001 0.634521629164145
0.630000000000001 0.654450356915395
0.640000000000001 0.674687129465602
0.650000000000001 0.695232908969949
0.660000000000001 0.716088709504453
0.670000000000001 0.737255596751145
0.680000000000001 0.758734687696669
0.690000000000002 0.780527150344075
0.700000000000002 0.802634203437639
0.710000000000002 0.825057116200513
0.720000000000002 0.847797208085045
0.730000000000002 0.870855848535562
0.740000000000002 0.894234456763472
0.750000000000002 0.917934501534509
0.760000000000002 0.941957500967948
0.770000000000002 0.96630502234765
0.780000000000002 0.990978681944761
0.790000000000002 1.01598014485193
0.800000000000002 1.04131112482891
0.810000000000002 1.06697338415931
0.820000000000002 1.09296873351856
0.830000000000002 1.11929903185266
0.840000000000002 1.14596618626787
0.850000000000002 1.17297215193103
0.860000000000002 1.20031893198044
0.870000000000002 1.22800857744714
0.880000000000002 1.25604318718659
0.890000000000002 1.28442490782046
0.900000000000002 1.31315593368855
0.910000000000002 1.34223850681069
0.920000000000002 1.37167491685847
0.930000000000002 1.40146750113675
0.940000000000002 1.43161864457487
0.950000000000002 1.46213077972739
0.960000000000002 1.49300638678427
0.970000000000002 1.52424799359052
0.980000000000002 1.55585817567504
0.990000000000002 1.58783955628872
};
\addplot [thick, red]
table {%
-1 1
-0.99 0.9801
-0.98 0.9604
-0.97 0.9409
-0.96 0.9216
-0.95 0.9025
-0.94 0.8836
-0.93 0.8649
-0.92 0.8464
-0.91 0.8281
-0.9 0.81
-0.89 0.7921
-0.88 0.7744
-0.87 0.7569
-0.86 0.7396
-0.85 0.7225
-0.84 0.7056
-0.83 0.6889
-0.82 0.6724
-0.81 0.6561
-0.8 0.64
-0.79 0.6241
-0.78 0.6084
-0.77 0.5929
-0.76 0.5776
-0.75 0.5625
-0.74 0.5476
-0.73 0.5329
-0.72 0.5184
-0.71 0.5041
-0.7 0.49
-0.69 0.4761
-0.68 0.4624
-0.67 0.4489
-0.66 0.4356
-0.65 0.4225
-0.64 0.4096
-0.63 0.3969
-0.62 0.3844
-0.61 0.3721
-0.6 0.36
-0.59 0.3481
-0.58 0.3364
-0.57 0.3249
-0.56 0.3136
-0.55 0.3025
-0.54 0.2916
-0.53 0.2809
-0.52 0.2704
-0.51 0.2601
-0.5 0.25
-0.49 0.2401
-0.48 0.2304
-0.47 0.2209
-0.46 0.2116
-0.45 0.2025
-0.44 0.1936
-0.429999999999999 0.1849
-0.419999999999999 0.1764
-0.409999999999999 0.1681
-0.399999999999999 0.16
-0.389999999999999 0.1521
-0.379999999999999 0.1444
-0.369999999999999 0.1369
-0.359999999999999 0.1296
-0.349999999999999 0.1225
-0.339999999999999 0.1156
-0.329999999999999 0.1089
-0.319999999999999 0.1024
-0.309999999999999 0.0960999999999996
-0.299999999999999 0.0899999999999996
-0.289999999999999 0.0840999999999996
-0.279999999999999 0.0783999999999996
-0.269999999999999 0.0728999999999996
-0.259999999999999 0.0675999999999997
-0.249999999999999 0.0624999999999997
-0.239999999999999 0.0575999999999997
-0.229999999999999 0.0528999999999997
-0.219999999999999 0.0483999999999997
-0.209999999999999 0.0440999999999997
-0.199999999999999 0.0399999999999997
-0.189999999999999 0.0360999999999997
-0.179999999999999 0.0323999999999997
-0.169999999999999 0.0288999999999997
-0.159999999999999 0.0255999999999998
-0.149999999999999 0.0224999999999998
-0.139999999999999 0.0195999999999998
-0.129999999999999 0.0168999999999998
-0.119999999999999 0.0143999999999998
-0.109999999999999 0.0120999999999998
-0.0999999999999992 0.00999999999999984
-0.0899999999999992 0.00809999999999985
-0.0799999999999992 0.00639999999999987
-0.0699999999999992 0.00489999999999988
-0.0599999999999992 0.0035999999999999
-0.0499999999999992 0.00249999999999992
-0.0399999999999991 0.00159999999999993
-0.0299999999999991 0.000899999999999948
-0.0199999999999991 0.000399999999999965
-0.00999999999999912 9.99999999999824e-05
8.88178419700125e-16 7.88860905221012e-31
0.0100000000000009 0.000100000000000018
0.0200000000000009 0.000400000000000036
0.0300000000000009 0.000900000000000055
0.0400000000000009 0.00160000000000007
0.0500000000000009 0.00250000000000009
0.0600000000000009 0.00360000000000011
0.070000000000001 0.00490000000000013
0.080000000000001 0.00640000000000015
0.090000000000001 0.00810000000000017
0.100000000000001 0.0100000000000002
0.110000000000001 0.0121000000000002
0.120000000000001 0.0144000000000002
0.130000000000001 0.0169000000000003
0.140000000000001 0.0196000000000003
0.150000000000001 0.0225000000000003
0.160000000000001 0.0256000000000003
0.170000000000001 0.0289000000000004
0.180000000000001 0.0324000000000004
0.190000000000001 0.0361000000000004
0.200000000000001 0.0400000000000004
0.210000000000001 0.0441000000000005
0.220000000000001 0.0484000000000005
0.230000000000001 0.0529000000000005
0.240000000000001 0.0576000000000005
0.250000000000001 0.0625000000000006
0.260000000000001 0.0676000000000006
0.270000000000001 0.0729000000000006
0.280000000000001 0.0784000000000006
0.290000000000001 0.0841000000000007
0.300000000000001 0.0900000000000007
0.310000000000001 0.0961000000000007
0.320000000000001 0.102400000000001
0.330000000000001 0.108900000000001
0.340000000000001 0.115600000000001
0.350000000000001 0.122500000000001
0.360000000000001 0.129600000000001
0.370000000000001 0.136900000000001
0.380000000000001 0.144400000000001
0.390000000000001 0.152100000000001
0.400000000000001 0.160000000000001
0.410000000000001 0.168100000000001
0.420000000000001 0.176400000000001
0.430000000000001 0.184900000000001
0.440000000000001 0.193600000000001
0.450000000000001 0.202500000000001
0.460000000000001 0.211600000000001
0.470000000000001 0.220900000000001
0.480000000000001 0.230400000000001
0.490000000000001 0.240100000000001
0.500000000000001 0.250000000000001
0.510000000000001 0.260100000000001
0.520000000000001 0.270400000000001
0.530000000000001 0.280900000000001
0.540000000000001 0.291600000000001
0.550000000000001 0.302500000000001
0.560000000000001 0.313600000000002
0.570000000000001 0.324900000000002
0.580000000000001 0.336400000000002
0.590000000000001 0.348100000000002
0.600000000000001 0.360000000000002
0.610000000000001 0.372100000000002
0.620000000000001 0.384400000000002
0.630000000000001 0.396900000000002
0.640000000000001 0.409600000000002
0.650000000000001 0.422500000000002
0.660000000000001 0.435600000000002
0.670000000000001 0.448900000000002
0.680000000000001 0.462400000000002
0.690000000000002 0.476100000000002
0.700000000000002 0.490000000000002
0.710000000000002 0.504100000000002
0.720000000000002 0.518400000000002
0.730000000000002 0.532900000000002
0.740000000000002 0.547600000000002
0.750000000000002 0.562500000000002
0.760000000000002 0.577600000000002
0.770000000000002 0.592900000000002
0.780000000000002 0.608400000000002
0.790000000000002 0.624100000000003
0.800000000000002 0.640000000000003
0.810000000000002 0.656100000000003
0.820000000000002 0.672400000000003
0.830000000000002 0.688900000000003
0.840000000000002 0.705600000000003
0.850000000000002 0.722500000000003
0.860000000000002 0.739600000000003
0.870000000000002 0.756900000000003
0.880000000000002 0.774400000000003
0.890000000000002 0.792100000000003
0.900000000000002 0.810000000000003
0.910000000000002 0.828100000000003
0.920000000000002 0.846400000000003
0.930000000000002 0.864900000000003
0.940000000000002 0.883600000000003
0.950000000000002 0.902500000000003
0.960000000000002 0.921600000000003
0.970000000000002 0.940900000000003
0.980000000000002 0.960400000000003
0.990000000000002 0.980100000000004
};
\end{axis}

\end{tikzpicture}

%% file: figs/grid.tikz
\begin{tikzpicture}[y=-1cm, scale=.79]
\path[draw=black,fill=white] (6.35,3.81) rectangle (11.43,8.89);
\draw[black] (8.89,3.81) -- (8.89,8.89);
\draw[black] (6.35,6.35) -- (11.43,6.35);
\path[draw=black,fill=white] (5.715,4.445) rectangle (10.795,9.525);
\draw[black] (5.715,6.985) -- (10.795,6.985);
\draw[black] (8.255,4.445) -- (8.255,9.525);
\path (10.15,10.00125) node[text=black,anchor=base west] {R};
\path (10.8,9.36625) node[text=black,anchor=base west] {G};
\path (11.4,8.73125) node[text=black,anchor=base west] {B};
\path[draw=black,fill=white] (5.08,5.08) rectangle (10.16,10.16);
\draw[black] (5.08,7.62) -- (10.16,7.62);
\draw[black] (7.62,5.08) -- (7.62,10.16);
\path (5.1,5.55625) node[text=black,anchor=base west] {$\mu_{1,1,1}=255$};
\path (7.77875,5.55625) node[text=black,anchor=base west] {$\mu_{1,2,1}=0$};
\path (5.1,8.09625) node[text=black,anchor=base west] {$\mu_{2,1,1}=0$};
\path (5.87375,4.92125) node[text=black,anchor=base west] {$\mu_{1,1,2}=0$};
\path (6.50875,4.28625) node[text=black,anchor=base west] {$\mu_{1,1,3}=0$};
\path (8.9,4.28625) node[text=black,anchor=base west] {$\mu_{1,2,3}=255$};
\path (8.2,4.92125) node[text=black,anchor=base west] {$\mu_{1,2,2}=255$};
\path (7.77875,8.09625) node[text=black,anchor=base west] {$\mu_{2,2,1}=0$};
\path (5.87375,6.50875) node[text=black,anchor=base west] {$\delta_{1,1,1}$};
\path (5.87375,9.2075) node[text=black,anchor=base west] {$\delta_{2,1,1}$};
\path (8.41375,9.2075) node[text=black,anchor=base west] {$\delta_{2, 2, 1}$};
\path (8.41375,6.50875) node[text=black,anchor=base west] {$\delta_{1,2,1}$};

\end{tikzpicture}%


%% file: figs/grid_merged.tikz
\begin{tikzpicture}[y=-1cm, scale=.5]
\definecolor{fillColor}{rgb}{0,0.6902,0}
\path[draw=black,fill=fillColor!50] (7.62,7.62) rectangle (10.16,10.16);
\definecolor{fillColor}{rgb}{0.52941,0.80784,1}
\path[draw=black,fill=fillColor!50] (5.08,7.62) rectangle (7.62,10.16);
\path[draw=black,fill=red!50] (5.08,5.08) rectangle (7.62,7.62);
\path[draw=black,fill=cyan!50] (7.62,5.08) rectangle (10.16,7.62);
\path (8,6.50875) node[text=black,anchor=base west] {$\mu_{1,2}$};
\path (5.5,9.04875) node[text=black,anchor=base west] {$\mu_{2,1}$};
\path (8,9.04875) node[text=black,anchor=base west] {$\mu_{2,2}$};
\path (5.5,6.50875) node[text=black,anchor=base west] {$\mu_{1,1}$};

\end{tikzpicture}%


%% file: secs/experiments.tex
\section{System-aware Attack Evaluation}
\label{sec:ext_eval}

In this section, we consider camera-based image classification systems, as used in self-driving vehicles and surveillance systems, to illustrate the potential
impact of our attacks. We present
proof-of-concept system-aware attacks in terms of \emph{attack
effectiveness}, namely how well system-aware attacks perform in the same setup
as camera-aware attacks (Section~\ref{sec:nai-eval}), and \emph{attack
robustness}, namely how well system-aware attacks are when being evaluated in
different setups.

We will again use attack success rates (Algorithm~\ref{alg:naive}) as our
metric. We used the Adam Optimizer~\cite{kingma2015adam} to solve our
optimization problems. There are two sets of results: \emph{Emulation results}
refer to the classification results on emulated, combined images of benign
images and attack patterns using our channel model
(Equation~\ref{eq:exp_ctrl}). Emulation helps us conduct scalable and fast
evaluations of GhostImage attacks before conducting real-world
experiments\footnote{Source code is available at
\url{https://github.com/Harry1993/GhostImage}}. \emph{Experimental results} refer to
the classification results on the images that are actually captured by the
victim cameras when the projector is on.

\subsection{Attack Effectiveness}
\label{sec:attack-effect}

To compare with camera-aware attacks,
system-aware attacks are evaluated in a similar procedure (Algorithm~\ref{alg:naive}), targeting a camera-based
object classification system with the LISA dataset and its classifier, and the
Aptina MT9M034 camera~\cite{aptina}, in an in-lab environment. %

\begin{figure}[t]
	\centering
	\input{./figs/adv_obfus.tikz}
	\caption{System-aware creation and alteration}
	\label{fig:adv-spoofing}
\end{figure}
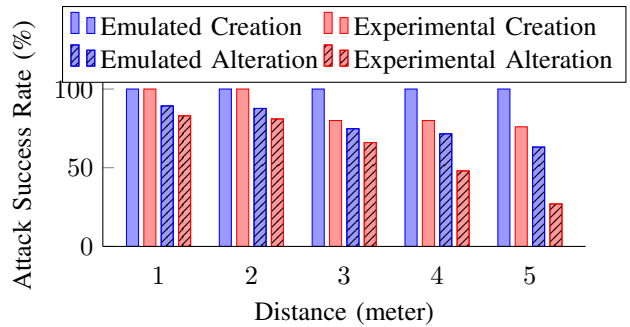

\subsubsection{Creation attacks}
\label{sec:exp-creation}

For emulated creation attacks, all distances (or all resolutions) yield attack
success rates of 100\% (Fig.~\ref{fig:adv-spoofing}), which means that our
optimization problem is easy to solve. In terms of computational overhead, we
need roughly \SI{30}{\second} per image at $2\times 2$-resolution, and
\SI{10}{\second} at $4\times 4$ or above (because of more degrees of freedom)
using an NVIDIA Tesla P100~\cite{nvidiap100}. Fig.~\ref{fig:dig_eg} shows
examples of emulated attack patterns for creation attacks, along with the
images of real signs on the top. Interestingly, high-resolution shapes do look
like real signs. For example, we can see two vertical bars for ADDEDLANE, and
also we can see a circle at the middle south for STOPAHEAD, etc. These results
are consistent with the ones from the MNIST
dataset~\cite{papernot2016limitations} where we could also roughly observe the
shapes of digits. Secondly, they are blue tinted because our channel model
suggests that ghosts tend to be blue, thus the optimizer is trying to find
``blue'' attack patterns that are able to deceive the classifier.

Interestingly, the all $k$ resulting patterns of solving the optimization
problem targeting one class from $k$ different (random) starting points look
similar to the ones shown in Fig.~\ref{fig:dig_eg}. However,
CIFAR-10~\cite{krizhevsky2009learning} and ImageNet~\cite{imagenet_cvpr09}
yield much different results: those patterns look rather random compared to the
results from LISA or MNIST. The reason might be that in CIFAR-10, images in the
same category are still very different, such as two different cats, but in
LISA, two images of STOP signs do not look as different as two cats.

For the experimental results of creation attacks, we see that as distances
increase, success rates decrease a little (Fig.~\ref{fig:adv-spoofing}),
but much better than the camera-aware attacks
(Table~\ref{tab:naive-rates}), because the optimization formulation helped find
those optimal attack patterns with high confidence.

\subsubsection{Alteration attacks}

The emulated and experimental results of alteration attacks are shown in
Fig.~\ref{fig:adv-spoofing}. Compared with creation attacks, alteration attacks
perform a bit worse, especially for large distances (three meters or further).
This is because the classifier also ``sees'' the benign image in the background
and tends to classify the entire image as the benign class.  Moreover, the
alignment of attack patterns and the benign signs is imperfect.  However, when
we compare Fig.~\ref{fig:adv-spoofing} with Table~\ref{tab:naive-rates} for
camera-aware alteration attacks, we can see large improvements.
Fig.~\ref{fig:per_eg} provides an example of system-aware alteration attacks in
the perception domain, which were trying to alter the (printed) STOP sign into
other signs: they look ``blue'' as the channel model predicted. The fifth
column is not showing as it is STOP.

A misclassification matrix of emulated alteration attacks at $8\times 8$ is
given in Table~\ref{tab:adv-obf}. The overall attack success rate was 75\%.
Each cell denotes the success rate of altering a benign class (actual) into a
target class (predicted). Most of them are 100\%, but the SCHOOL sign, for
example, was the most difficult to perturb \emph{into} (the 3rd column) and yet
not that hard to perturb \emph{from} (the 3rd row), probably because it is in
green (RGB: 0-255-0) and in an opened-envelope shape, while all the others are
either red (255-0-0) or yellow (255-255-0) colors, and either polygon or
rectangle shapes.

\begin{figure}[t]
	\begin{subfigure}{\columnwidth}
		\centering
		\input{./figs/dig_eg.tikz}
		\caption{Emulated creation attacks}
		\label{fig:dig_eg}
	\end{subfigure}
	\begin{subfigure}{\columnwidth}
		\centering
		\input{./figs/per_eg.tikz}
		\caption{Experimental alteration attacks}
		\label{fig:per_eg}
	\end{subfigure}

	\caption{System-aware attack pattern examples.}

\end{figure}
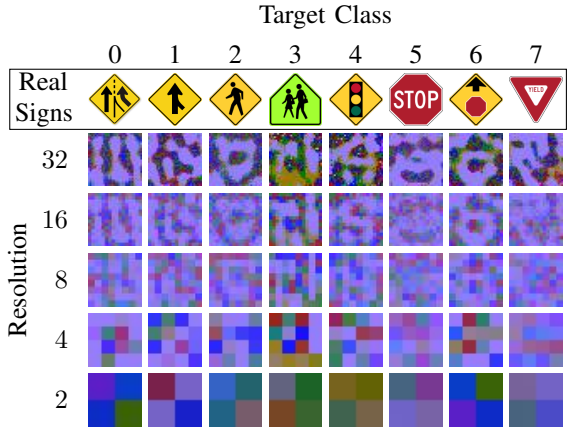

\input{tabs/obf8.tex}

\subsection{Attack Robustness}

We evaluate the robustness of our attacks in terms of different
datasets, environments, and cameras.

\subsubsection{Different image datasets}

Here we evaluate our system-aware attacks on two other datasets,
CIFAR-10~\cite{krizhevsky2009learning} and ImageNet~\cite{imagenet_cvpr09}, by
emulation only because previous results show that our attack emulation 
yields similar success rates as experimental results.

\paragraph{CIFAR-10} The network architecture and model hyper parameters are
shown in Table~\ref{tab:cifar10-arch} and Table~\ref{tab:cifar10-hyper} in
Appendix~\ref{sec:nntables}, which are identical to \cite{carlini2017towards}.
The network was trained with the distillation
defense~\cite{papernot2016distillation} so that we can evaluate the robustness
of our attacks in terms of adversarial defenses.  A classification accuracy of
$80\%$ was achieved. The evaluation procedure is similar to
Algorithm~\ref{alg:naive}. Results are shown in Fig.~\ref{fig:imagenet_dig}.
The overall trend is similar to the LISA dataset, but the success rates are
higher. The reason might still be the large variation within one class
(Section~\ref{sec:exp-creation}), so that the CIFAR-10 classifier is not as
sure about one class as the LISA classifier is, hence is more vulnerable to our
attacks.

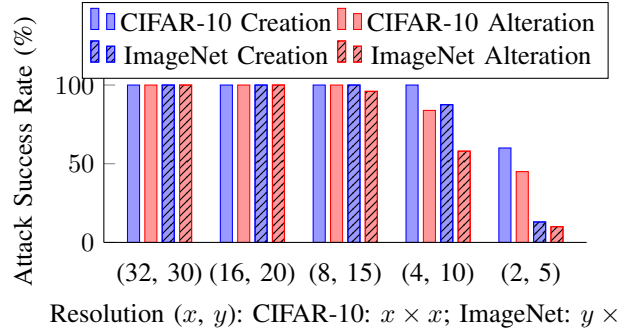
\begin{figure}[t]
	\centering
	\input{./figs/imagenet_dig.tikz}
	\caption{System-aware attacks on CIFAR-10 and ImageNet}
	\label{fig:imagenet_dig}
\end{figure}

\paragraph{ImageNet} We used a pre-trained Inception V3 neural
network~\cite{szegedy2016rethinking} for the ImageNet dataset to evaluate the
attack robustness against large networks.
Since the pre-trained network can recognize 1000 classes, we did not iterate
all of them (similar to \cite{carlini2017towards}). Instead, for alteration attacks,
we randomly picked ten benign images from the validation set, and twenty random
target classes, while for creation attacks, the ``benign'' images were  
purely black. Results are given in Fig.~\ref{fig:imagenet_dig}. 

For high resolutions ($\ge 15\times 15$), the attack success rates were nearly
100\%.  But as soon as the resolutions went down to $10 \times 10$ or below,
the rates decreased sharply.  The reason might be that in order to mount
successful \emph{targeted} attacks on a 1000-class image classifier, a large
number of degrees of freedom are required. $10 \times 10$ or lower resolutions
plus three color channels might not be enough to accomplish targeted attacks.
To verify this, we also evaluated untargeted alteration attacks on ImageNet.
Results show that when the resolutions are $1\times 1$ or $2\times 2$, the
success rates are 50\% or 80\%, respectively. But as soon as the
resolutions go to $3\times 3$ or above, the success rates reach 100\%.
Lastly, similar to CIFAR-10, system-aware attacks on ImageNet were more
successful than on LISA, because of the high variation within one class.

\subsubsection{Outdoor experiments}

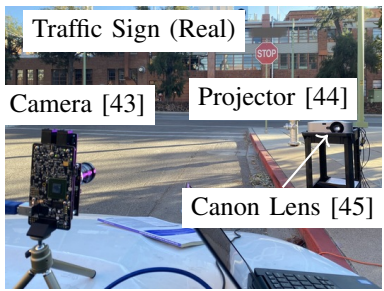
\begin{figure}[t]
	\centering
	\input{./figs/inveh_setup.tikz}
	\caption{Outdoor experiment setup}
	\label{fig:inveh_setup}
\end{figure}

In order to evaluate system-aware attacks in a real-world environment, we also
conducted experiments outdoor (Fig.~\ref{fig:inveh_setup}), where the
camera was put on the hood of a vehicle that was about to pass an intersection
with a STOP sign. The attacker's projector was placed on the right curb, and it
was about four meters away from the camera. The experiments were done at noon,
at dusk and at night (with the vehicle's front lights on) to examine the
effects of ambient light on attack efficacy. The illuminances were
\SI{4e4}{\lux}, \SI{4e3}{\lux}, and \SI{30}{\lux}, respectively. The
experiments at noon were unsuccessful due to the strong sunlight. Although
more powerful projectors \cite{barco} could be acquired, we argue that a
typical projector is effective in dimmer environments (e.g., cloudy days, at
dawn, dusk, and night, or urban areas where buildings cause shades), which
accounts for more than half of a day. See Sec.~\ref{sec:disc-prac} for more
discussion on ambient lighting conditions.

Results (Tab.~\ref{tab:inveh_obf}) of the other cases show that the success
rates are 30\% lower than our in-lab experiments (the four-meter case from
Fig.~\ref{fig:adv-spoofing}), because we used our in-lab channel model directly
in the road experiments without retraining it, and also the environmental
conditions are more unpredictable. Moreover, the attack rates on
altering some classes (e.g., the STOP sign) into three other signs (e.g.,
YIELD) were 100\%, which is critical as an attacker can easily prevent an
autonomous vehicle from stopping at a STOP sign.

\begin{table}[b]
	\centering
	\arrayrulecolor{black}
	\caption{Outdoor alteration attack success rates}
	\rowcolors{2}{gray!10}{}
	\label{tab:inveh_obf}
	\begin{tabular}{lllll}
		\toprule
		Success rates of & Noon & Dusk & Night &  \\
		\midrule
		Overall & 0\%        & 51\%    & 42.9\%  &  \\
		STOP $\rightarrow$ YIELD     & 0\%        & 100\%   & 100\%   &  \\
		STOP $\rightarrow$ ADDEDLANE & 0\%        & 100\%   & 100\%   &  \\
		STOP $\rightarrow$ PEDESTRIAN & 0\%        & 100\%   & 100\%   &  \\
		\bottomrule
	\end{tabular}
\end{table}

\subsubsection{Different cameras}
\label{sec:diff_cam}

Previously, we conducted GhostImage attacks on Aptina MT9M034 camera
\cite{aptina} designed for autonomous driving. Here, we  evaluate two other
cameras, an Aptina MT9V034~\cite{aptinamt9v} with a simpler lens design, and a
Ring indoor security camera~\cite{ring} for surveillance applications.

\paragraph{Aptina MT9V034} We mounted system-aware creation attacks against the
same camera-based object classification system as in
Section~\ref{sec:attack-effect} but we replaced the camera with the Aptina
MT9V034 camera. Since this camera has a smaller aperture size and also a
simpler lens design than Aptina MT9M034, for a distance of one meter, only
$16\times 16$-resolution attack patterns could be achieved (previously we had
$32\times 32$ at one meter). We did not train a new channel model for this
camera, so the attack success rate at one meter was only 75\%, which is 25\%
lower than the Aptina MT9M034 camera.  As the distances increased up to four
meters, creation attacks yielded success rates as 46.25\%, 33.75\%, and 12.5\%,
respectively.  Another reason why the overall success rate was lower is that
even though the data sheet of Aptina MT9V034~\cite{aptinamt9v} states that the
camera also has the auto exposure control feature, we could not enable the
feature in our experiments.  In other words, system-aware creation attacks did
not benefit from the exposure control. This, on the other hand, indicates
the robustness of GhostImage attacks: Even without taking advantage of exposure
control, the attacks were still effective.

\begin{table}[b]
	\centering
	\arrayrulecolor{black}

	\caption{GhostImage untargeted alteration attacks against Ring camera on
	ImageNet dataset in perception domain}

	\rowcolors{2}{gray!10}{}
	\label{tab:ring}
	\begin{tabular}{llll}
		\toprule
		Index & Benign Class & Rate & Common Prediction \\
		\midrule
		19992 & fur boat   & 100\% & geyser, parachute \\
		21539 & sunglasses & 100\% & screen, microwave \\
		22285 & sunglasses & 100\% & plastic bag, geyser \\
		31664 & sarong     & 100\% & jellyfish, plastic bag \\
		2849  & sweatshirt & 100\% & laptop, candle \\
		26236 & puncho     & 100\% & table lamp \\
		\bottomrule
	\end{tabular}
\end{table}

\paragraph{Ring indoor security camera} We tested GhostImage untargeted attacks
against a Ring indoor security camera~\cite{ring} on the ImageNet dataset. To
demonstrate that our attacks can be applied to surveillance scenarios, we
assume the camera would issue an intrusion warning if a specific object type
\cite{motion-detection} is detected by the Inception V3 neural
network~\cite{szegedy2016rethinking}. The attacker's goal is to change an
object for an intruder class to a non-intruder class. However, we could not
find ``human'', ``person'' or ``people'', etc. in the output classes, we
instead used five human related items (such as sunglasses) as the benign
classes. We found six images from the validation set of ImageNet, of which
top-1 classification results are one of those five benign classes. The six
images were displayed on a monitor. For each benign image, we calculated ten
alternative $3\times 3$ attack patterns (the highest resolution at one meter by
the Ring camera). Results show that our attacks achieve an overall untargeted
attack success rate of 100\% (Tab.~\ref{tab:ring}).

%% file: figs/adv_obfus.tikz
\begin{tikzpicture}
  \centering
  \begin{axis}[
        height=4cm,
		ybar,
		bar width=.13,
        axis on top,
		major x tick style = transparent,
    	enlargelimits=0.15,
    	legend style={
			at={(0.5,.9)},
    		anchor=south,
			align=left,
			legend columns=2},
		legend cell align={left},
    	ylabel={\#participants},
		ytick align=inside,
        major grid style={draw=white},
        enlarge y limits={value=.1,upper},
        ymin=0, ymax=105,
        axis x line*=bottom,
        axis y line*=left,
        ylabel={Attack Success Rate (\%)},
		xtick={1, 2, 3, 4, 5},
        xlabel={Distance (meter)},
    ]
    \addplot [draw=blue, fill=blue!40] coordinates {
      (1,100)
      (2,100) 
      (3,100)
      (4,100)
      (5,100) };
   \addplot [draw=red, fill=red!40] coordinates {
      (1,100)
      (2,100) 
      (3,80)
      (4,80) 
      (5,76)};
   \addplot [draw=blue, postaction={pattern=north east lines}, fill=blue!40] coordinates {
      (1,89.28)
      (2,87.67) 
      (3,74.78)
      (4,71.54)
      (5,63.12) };
   \addplot [draw=red, postaction={pattern=north east lines}, fill=red!40] coordinates {
      (1,83)
      (2,81) 
      (3,66)
      (4,48) 
      (5,27)};
    \legend{Emulated Creation, Experimental Creation, Emulated Alteration, Experimental Alteration}
  \end{axis}
\end{tikzpicture}

%% file: figs/dig_eg.tikz
\begin{tikzpicture}
  \foreach \x in {0,...,4} {
    \pgfmathsetmacro\result{2^(5-\x)};
	\node[anchor=east] at (-.5, -.8*\x) {\pgfmathprintnumber{\result}};
    \foreach \y in {0,...,7}
	   \node at (.8*\y, -.8*\x) {\includegraphics[width=.7cm]{figs/dig_eg/\x-\y.jpg}};
	}

  \foreach \y in {0,...,7}
	   \node at (.8*\y, 1.4) {\y};

  \foreach \y in {0,...,7}
     \node at (.8*\y, .8) {\includegraphics[width=.7cm]{figs/dig_eg/s\y.png}};

  \draw[black] (-1.4, .4) rectangle (6, 1.2);
  \node[text width=20, anchor=east, align=right] at (-.5, .8) {Real Signs};

  \node at (2.8, 1.9) {Target Class};
  \node[rotate=90] at (-1.3, -1.6) {Resolution};
\end{tikzpicture}

%% file: figs/per_eg.tikz
\begin{tikzpicture}
  \foreach \x in {0,...,4} {
    \pgfmathsetmacro\result{2^(5-\x)};
	\node[anchor=east] at (-.5, -.8*\x) {\pgfmathprintnumber{\result}};
    \foreach \y in {0,1,2,3,4,6,7}
	   \node at (.8*\y, -.8*\x) {\includegraphics[width=.7cm]{figs/per_eg/\x-\y.jpg}};
	}


  \node[rotate=90] at (-1.3, -1.6) {Resolution};
\end{tikzpicture}

%% file: tabs/obf8.tex
\begin{table}[b]
	\centering
	\noindent
	\newcommand\items{8}   %
	\arrayrulecolor{white} %

	\caption{Emulated system-aware obfuscation attacks: misclassification
	matrix at $8\times 8$-resolution.}

	\label{tab:adv-obf}
	\begin{tabular}{cc*{\items}{|E}|}
		\multicolumn{1}{c}{} &\multicolumn{1}{c}{} &\multicolumn{\items}{c}{Predicted} \\ \hhline{~*\items{|-}|}
		\multicolumn{1}{c}{} & 
		\multicolumn{1}{c}{} & 
		\multicolumn{1}{c}{0} & 
		\multicolumn{1}{c}{1} & 
		\multicolumn{1}{c}{2} & 
		\multicolumn{1}{c}{3} & 
		\multicolumn{1}{c}{4} & 
		\multicolumn{1}{c}{5} & 
		\multicolumn{1}{c}{6} & 
		\multicolumn{1}{c}{7} \\
		\hhline{~*\items{|-}|}
		\multirow{\items}{*}{\rotatebox{90}{Actual}} 
		&\rot{0}  & 100   & 0   & 100 & 100 & 100 & 100 & 0 & 100  \\ \hhline{~*\items{|-}|}
		&\rot{1}  & 0   & 100  & 100 & 0 & 100 & 100 & 100 & 100  \\ \hhline{~*\items{|-}|}
		&\rot{2}  & 0   & 100  & 100 & 0 & 100 & 100 & 100 & 100  \\ \hhline{~*\items{|-}|}
		&\rot{3}  & 100  & 0 & 100 & 100 & 100 & 100 & 0 & 97.2  \\ \hhline{~*\items{|-}|}
		&\rot{4}  & 0   & 100  & 100 & 0 & 100 & 100 & 100 & 100  \\ \hhline{~*\items{|-}|}
		&\rot{5}  & 100   & 100  & 90 & 0 & 100 & 100 & 100 & 100  \\ \hhline{~*\items{|-}|}
		&\rot{6}  & 100   & 100  & 100 & 100 & 100 & 100 & 100 & 100  \\ \hhline{~*\items{|-}|}
		&\rot{7}  & 100   & 100  & 0 & 0 & 0 & 100 & 100 & 100  \\ \hhline{~*\items{|-}|}
	\end{tabular}
\end{table}

%% file: figs/imagenet_dig.tikz
\begin{tikzpicture}
  \centering
  \begin{axis}[
        height=4cm,
		ybar,
		bar width=.13,
        axis on top,
		major x tick style = transparent,
    	enlargelimits=0.15,
    	legend style={
			at={(0.5,.9)},
    		anchor=south,
			align=left,
			legend columns=2},
    	ylabel={\#participants},
		ytick align=inside,
        major grid style={draw=white},
        enlarge y limits={value=.1,upper},
        ymin=0, ymax=105,
        axis x line*=bottom,
        axis y line*=left,
        ylabel={Attack Success Rate (\%)},
		xticklabels={(32, 30), (16, 20), (8, 15), (4, 10), (2, 5)},
		xtick={1, 2, 3, 4, 5},
        xlabel={Resolution ($x$, $y$): CIFAR-10: $x\times x$; ImageNet: $y\times y$},
    ]
   \addplot [draw=blue, fill=blue!40] coordinates {
      (1,100)
      (2,100) 
      (3,100)
      (4,100) 
      (5,60)};
   \addplot [draw=red, fill=red!40] coordinates {
      (1,100)
      (2,100) 
      (3,99.98)
      (4,83.89)
      (5,45) };
   \addplot [draw=blue, postaction={pattern=north east lines}, fill=blue!40] coordinates {
      (1,100)
      (2,100) 
      (3,100)
      (4,87.5)
      (5,13) };
   \addplot [draw=red, postaction={pattern=north east lines}, fill=red!40] coordinates {
      (1,100)
      (2,100) 
      (3,96)
      (4,58) 
      (5,10)};
    \legend{CIFAR-10 Creation, CIFAR-10 Alteration, ImageNet Creation, ImageNet Alteration}
  \end{axis}
\end{tikzpicture}

%% file: figs/inveh_setup.tikz
\begin{tikzpicture}[y=-1cm]
\node[anchor=north west,inner sep=0pt] at (3.81,5.08) {\includegraphics[width=5cm]{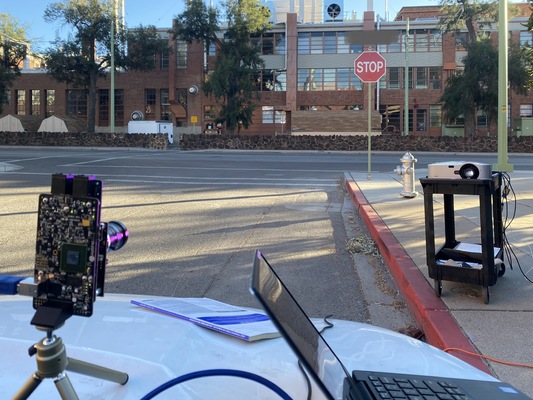}};
\path (5.8,6.5) node[text=black,anchor=base east,fill=white] {Camera~\cite{aptina}};
\path (8.5,6.4) node[text=black,anchor=base east,fill=white] {Projector~\cite{nec-np1150-manual}};
\path (7,5.5) node[text=black,anchor=base east,fill=white] {Traffic Sign (Real)};
\draw[->, draw=white, thick] (7.5, 7.5) -- (8.1, 6.8) node[pos=0, anchor=north, text=black, fill=white] {Canon
Lens~\cite{canon}};

\end{tikzpicture}%


%% file: secs/discussion.tex
\section{Discussion}
\label{sec:disc}

In this section, we discuss practical challenges to GhostImage attacks,
speculate as to effective countermeasures, and outline variations on the
original attacks.

\subsection{Practicality of GhostImage Attacks}
\label{sec:disc-prac}

\textbf{Moving targets and alignment:} The overlap of ghosts and objects of
interest in images must be nearly complete for the attacks to succeed.  In the
cases of a moving camera (e.g., one mounted to a vehicle), the attacker needs
to be able to accurately track the movement of the targeted camera, otherwise
the attacker can only sporadically inject ghosts. Note that,
although aiming (or tracking) moving targets is generally challenging in remote
sensor attacks (e.g., the AdvLiDAR attack \cite{cao2019adversarial}  
assumes the attacker can achieve this via camera-based object detection and
tracking), existing works \cite{truong2005preventing, chinalasers} have
demonstrated the feasibility of tracking cameras and then neutralizing them.
This paper's main goal is to propose a new category of camera attacks, which
enables an attacker to inject arbitrary patterns.

\textbf{Conspicuousness:} The light bursts around the light source in Figures
\ref{fig:yolov3_stop} and \ref{fig:naive_eg} may raise  stealthiness concerns
about our attacks. However, according to our analysis in Sec.~\ref{sec:ghost-pos}, such
bursts can actually be eliminated because the light source can be outside of
view (See Fig.~\ref{fig:oov} and \cite{wiki-flare}). Even the light source has
to be in the frame (due to the lens configuration), we argue that a
camera-based object classification system used in autonomous systems generally
make decisions without human input (for example, in a Waymo self-driving
taxi~\cite{waymo}, no human driver is required, or in a Tesla
car~\cite{autopilot}, real-time images would not typically be displayed).
Additionally, the attack beam is so concentrated that only the victim camera
can observe it while other human-beings (e.g., pedestrians) cannot
(Fig.~\ref{fig:inveh_setup}).  Finally, the light source only needs to be on
for a short amount of time, as a few tampered frames can cause
incorrect actions \cite{nassi2020phantom}.

\textbf{Projectors, lenses, and attack distances:} Based on our model
(Eq.~\ref{eq:gamma}) and experiments (Tab.~\ref{tab:cm_para}), the
illuminance on the camera from the projector would better be $4/3$ of the part
from ambient illuminance (to achieve a success rate of 100\%). Since
	$\text{Illuminance} \propto \text{Luminance}\cdot r_\text{throw}^2/d^2$,
in order to carry out an attack during sunny days (typically with Illuminance
\SI{40e3}{\lux}), a typical projector (e.g., \cite{epsonproj} with Luminance
\SI{9e3}{\lumen}) should work with a telephoto lens \cite{opteka} (with a
throwing radio $100$) at a distance of one meter. For longer distances or
brighter backgrounds, one can either acquire a more powerful projector (e.g.,
\cite{barco} with \SI{75e3}{\lumen}), or combine multiple lenses to achieve
much larger throwing ratios (e.g., two Optela lenses \cite{opteka} yield $200$,
etc.), or both.

\textbf{Ghost effect dependence:} There are several challenges an attacker
needs to overcome to launch GhostImage attacks. First, the attacks rely largely
on ghost effects; if ghosts cannot be induced, or if they are not significant
enough, the attacks might be infeasible against a given camera (lens). However,
this is unlikely because these effects occur in most cameras (e.g., Apple
iPhones~\cite{iphone8flare, iphone11flare}). Moreover, no ``flare-free'' lens
exists to the best of our knowledge (even with anti-glare coatings).
In addition, if ghost effects are unavailable to the attacker there are other
optics effects available, such as blooming effects~\cite{blooming}, that can
also be
leveraged to produce GhostImage-like attacks.

\textbf{Knowledge of the targeted system} We assume that both types of attackers know
about the camera matrix $M$ and color calibration matrix $H_c$.  We note that
the attacks can still be \emph{effective} without such knowledge but with it
the attacks can be more \emph{efficient}. For example, the attacker may choose
to lower their attack success expectation but the probability of successful
attack may still be too high for potential victims to bear (e.g., a success
rate of only 10\% might be unacceptable for reasons of safety in automated
vehicles).  This challenge can be largely eliminated if the attacker is able to
purchase a camera of the same, or similar, model as used in the targeted system
and use it to derive the matrices.  Although the duplicate camera may not
be exactly the same to the target one, the channel model would still be in the
same form with approximate, probably fine-tuned parameters (via retraining),
thanks to the generality of our channel model. Lastly, assuming white-box
knowledge on sensors is widely adopted and accepted in the literature, e.g.,
the AdvLiDAR attack~\cite{cao2019adversarial}.  Also, we assume white-box
attacks on the neural network, though this assumption can be eliminated by
leveraging the transferability of adversarial examples
\cite{papernot2016transferability, papernot2017practical,
bhagoji2018practical, chen2020devil}.

\textbf{Object detection:} We have
assumed that   the object detector can crop out the region of the image which contains the projected ghost pattern(s). Though it cannot be guaranteed that an object detector will
automatically include the ghost patterns, we note that a GhostImage attacker
could design ghost patterns that cause an object detector to include them
\cite{song2018physical, zhao2019seeing} and, at the same time, the cropped
image would fool the subsequent object classifier.

\subsection{Attack Variations}
\label{sec:alter}

Should ghost effect  be not available, we investigated alternative strategies
that still allow an attacker to cause misclassificaiton of objects of interest.
inject adversarial noise, absent ghosts and/or flares, without placing the
light source directly in front of the object, which would allow for
straightforward detection of the attack.
For example, using a beamsplitter that merges two light beams coming from two
directions: one is the light reflected from the object the attacker wishes to
obscure and the other is the light from the projector. The merged light beams
enter the targeted camera as a superposition of the original object and the
adversarial pattern, with the resulting image able to fool the classifier.
Appendix~\ref{sec:beamsplit} provides details on this attack vector and its
efficacy.

Finally, while projectors provide an attacker with the greatest control over
adversarial patterns, and hence the ability to spoof complex objects, we have
found that RGB lasers~\cite{rgblasers} can be used at greater distances to
spoof simple objects. It may also be possible for an attacker to use multi-laser systems, or even flashlights~\cite{rgbflashlights}, to create complex patterns akin to the decorative lights displayed on
Christmas trees \cite{christmastree}.

\subsection{Countermeasures}

The most straightforward countermeasure to GhostImage attacks is flare
elimination, either by using a lens hood~\cite{fotographee, lenshood} or
through flare detection.  Lens hoods are generally not favored as they reduce
the angle of view of the camera, which is unacceptable for many autonomous
vehicle and surveillance applications. Note that there are so-called
liquid lenses that can change its focal length by reshaping itself, which
results in only one reflection path, hence fewer flares. However, such lenses
have not been widely adopted yet \cite{liquidlenses}.

Currently, we are working on a defense where we   first identify all
flares/ghosts in an image, and then detect if a ghost contains  malicious patterns. However, the challenges
include: First, ghosts are typically transparent thus hard to detect
\cite{xu2015transcut}; Second, natural ghosts are so common and varied that false positives can occur inevitably; Third, just as adversarial noise can be crafted
to deceive neural networks an analogous procedure could be used to craft
flares/ghosts that deceive flare/ghost detectors. 

A complementary line of defense would be to make neural networks themselves
robust to GhostImage attacks. Existing approaches against adversarial examples
(e.g., \cite{papernot2016distillation, madry2018towards, lecuyer2019certified,
wu2019defending}, etc.) are ill-suited for this task, however, as GhostImage attacks
do not necessarily follow the constraints placed on traditional adversarial
examples in that perturbations do not have to be bounded within a small norm,
meanwhile these defenses were not designed for arbitrarily
large perturbations.

Another complementary approach of defense is to exploit prior knowledge, such
as GPS locations of signs, to make decisions, instead of only depending on
real-time sensor perception (though this approach would not work for
spontaneous appearance of objects, e.g., in the context of collision
avoidance). Sensor redundancy/fusion could also be helpful: autonomous vehicles
could be equipped with multiple cameras and/or other types of sensors, such as
LiDARs and radars, which would at least increase the cost of the attack by
requiring the attacker to target multiple sensors.

%% file: secs/related_work.tex
\section{Related Work}
\label{sec:relatedwork}

Since our attack spans two domains, in this section we review both sensor
attacks and adversarial examples.

\subsection{Sensor attacks} Perception in autonomous and surveillance systems
occurs through sensors, which convert analog signals into digital ones that are
further analyzed by computing systems. Recent work has demonstrated that the
sensing mechanism itself is vulnerable to attack and that such attacks may be
used to bypass digital protections~\cite{giechaskiel2020taxonomy,
yan2020minimalist}. For example, anti-lock braking system (ABS) sensors have
been manipulated via magnetic fields by Shoukry~et~al.~\cite{shoukry2013non},
microphones have been subject to inaudible voice and light-based attacks
\cite{zhang2017dolphinattack,sugawaralight}, and light sensors can be
influenced via electromagnetic interference to report lighter or darker
conditions~\cite{selvaraj2018electromagnetic}.  The reader is referred to
\cite{giechaskiel2020taxonomy, yan2020minimalist} for a review of
analog sensor attacks. %

Existing remote attacks against cameras \cite{petit2015remote, yan2016can,
truong2005preventing} are denial-of-service attacks and do not seek to
compromise the object classifier as our GhostImage attacks do. Those attacks
that do target object classification \cite{zhao2019seeing, eykholt2018robust,
chernikova2019are} are either digital or physical domain attacks (i.e., they
need to modify the object of interest in this case a traffic sign or road
pavement, physically or after the object has been captured by a camera) rather
than perception domain attacks~\cite{yan2020minimalist,
giechaskiel2020taxonomy}. Li et al.~\cite{li2019adversarial}'s attacks on
cameras require attackers to place stickers on lenses, to which is generally
hard to get access. Similarly, several light-based attacks
\cite{zhou2018invisible, nassi2020phantom, nguyen2020adversarial} fall within
the domain of physical attacks, as opposed to our perception domain attack,
because these approaches illuminate the object of interest with visible
or infrared light. We did not consider infrared noise in our attacks as it can
be easily eliminated from visible light systems using infrared filters.
Attacks on LiDAR systems \cite{shin2017illusion, petit2015remote,
cao2019adversarial, tu2020physically} are also related to this work;
however, these attacks are considerably easier to carry out than our visible
light-based attacks against cameras because attackers can directly inject
adversarial laser pulses into LiDARs without worrying about blocking the object
of interest.

\subsection{Adversarial examples} State-of-the-art adversarial examples can be
categorized as digital \cite{szegedy2014intriguing, carlini2017towards,
papernot2016limitations, moosavi2016deepfool, goodfellow2014explaining,
kurakin2018adversarial, moosavi2017universal, xiao2018spatially,
bhagoji2018practical, papernot2016transferability, papernot2017practical,
liu2016delving, brendel2018decision}, or physical domain attacks
\cite{kurakin2018adversarial, eykholt2018robust, sharif2016accessorize,
zhou2018invisible, sitawarin2018rogue, zhao2019seeing, wu2019making,
komkov2019advhat, duan2020adversarial, sato2020security, chen2020devil} in
which objects of interest are physically modified to cause misclassification.
The latter differs from GhostImage attacks in that we target the sensor
(camera) without needing to physically modify any real-world object. Another
line of work focuses on unrestricted adversarial examples
\cite{song2018constructing, hosseini2018semantic, Bhattad2020Unrestricted},
though they are limited in the digital domain.

In terms of defending neural networks from adversarial examples, be they
physical or digital, schemes include modifying the network to be more
robust~\cite{papernot2016distillation, madry2018towards,
goodfellow2014explaining, dhillon2018stochastic, xie2018mitigating,
liu2018towards, wong2018provable, lecuyer2019certified,
raghunathan2018certified, cao2017mitigating, wu2019defending}, while other
defenses have focused on either detecting adversarial
inputs~\cite{xu2017feature, hendrycks2017early, metzen2017detecting,
meng2017magnet, ma2019nic, cai2020detecting} or transforming them into benign
images~\cite{meng2017magnet, buckman2018thermometer, guo2018countering},
most of which are under the general assumption of bounded perturbations,
hence are inapplicable  to our attacks;  while others could also be bypassed by being
taken as constraints in the optimization formulation. As this work mainly focuses on sensor attacks, similar to \cite{cao2019adversarial, zhao2019seeing,
eykholt2018robust, sharif2016accessorize} we leave the validation of defenses as future work.

%% file: secs/conclusion.tex
\section{Conclusion}
\label{sec:con}

In this work we presented GhostImage attacks against camera-based object
classifiers.  Using common optical effects, viz.\ lens flare/ghost effects, an
attacker is able to inject arbitrary adversarial patterns into camera images
using a projector. To increase the efficacy of the attack, we proposed a
projector-camera channel model that predicts the location of ghosts, the
resolution of the patterns in ghosts, given the projector-camera arrangement,
and accounts for exposure control and color calibration. GhostImage attacks
also leverage adversarial examples generation techniques to find optimal attack
patterns. We evaluated GhostImage attacks using three image datasets and in
both indoor and outdoor environments on three cameras.  Experimental results
show that GhostImage attacks were able to achieve attack success rates as high
as 100\%, and also have potential impact on autonomous systems, such as
self-driving cars and surveillance systems.

%% file: secs/appendix.tex
\section{Neural networks and datasets}
\label{sec:nntables}

\begin{figure}[h]
	\centering
	\includegraphics[width=.5\columnwidth]{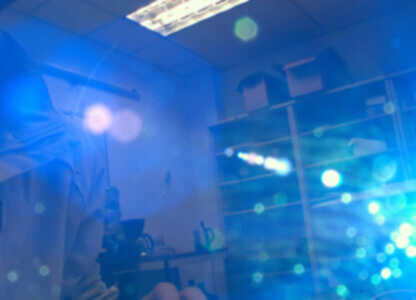}
	\caption{There are ghosts but the light source is out of view.}
	\label{fig:oov}
\end{figure}

Here we present the architectures of two neural networks
(Tables~\ref{tab:lisa_arch} and \ref{tab:cifar10-arch}) and their
hyper-parameters (Table~\ref{tab:cifar10-hyper}). The balanced LISA dataset is
also detailed in Table~\ref{tab:lisa}.

\begin{table}[h]
	\centering
	\caption{Neural network architecture for LISA dataset}
	\rowcolors{2}{gray!10}{}
	\label{tab:lisa_arch}
	\begin{tabular}{ll}
		\toprule
		Layer Type           & Model                     \\ \midrule
		ReLU convolution     & $64$ $8 \times 8$-filters  \\
		ReLU convolution     & $128$ $6\times 6$-filters  \\
		ReLU convolution     & $128$ $5\times 5$-filters \\
		ReLU Fully-connected & $256$                     \\
		Softmax              & $8$                      \\ \bottomrule
	\end{tabular}
\end{table}

\begin{table}[h]
	\centering
	\caption{Balanced LISA dataset}
	\label{tab:lisa}
	\rowcolors{2}{gray!10}{}
	\begin{tabular}{lll}
		\toprule
		Index & Sign Name           & Quantity \\
		\midrule
		0 & Added Lane          & 80       \\
		1 & Merge               & 80       \\
		2 & Pedestrian Crossing & 80       \\
		3 & School              & 77       \\
		4 & Signal Ahead        & 80       \\
		5 & Stop                & 80       \\
		6 & Stop Ahead          & 80       \\
		7 & Yield               & 80       \\
		\bottomrule
	\end{tabular}
\end{table}

\begin{table}[h]
	\centering
	\caption{Neural network architecture for CIFAR-10 dataset}
	\rowcolors{2}{gray!10}{}
	\label{tab:cifar10-arch}
	\begin{tabular}{ll}
		\toprule
		Layer Type           & Model                     \\
		\midrule
		ReLU convolution     & $64$ $8 \times 8$-filters  \\
		ReLU convolution     & $128$ $6\times 6$-filters  \\
		ReLU convolution     & $128$ $5\times 5$-filters \\
		ReLU convolution     & $64$ $3\times 3$-filters  \\
		Max Pooling          & $2\times 2$               \\
		ReLU convolution     & $128$ $3\times 3$-filters \\
		ReLU convolution     & $128$ $3\times 3$-filters \\
		Max Pooling          & $2\times 2$               \\
		ReLU Fully-connected & $256$                     \\
		ReLU Fully-connected & $256$                     \\
		Softmax              & $10$                      \\
		\bottomrule
	\end{tabular}
\end{table}

\begin{table}[h]
	\centering
	\arrayrulecolor{black}
	\caption{Training hyper-parameters}
	\rowcolors{2}{gray!10}{}
	\label{tab:cifar10-hyper}
	\begin{tabular}{ll}
		\toprule
		Parameter     & Value \\
		\midrule
		Learning Rate & $0.1$ \\
		Momentum      & $0.9$ \\
		Dropout       & $0.5$ \\
		Batch Size    & $128$ \\
		Epochs        & $50$  \\
		\bottomrule
	\end{tabular}
\end{table}

\section{Detailed Projector-Camera Model Parameters}

Table~\ref{tab:cm_para} lists all parameters of the projector-camera channel
model. The color calibration matrix is
\[
	H_c = \begin{bmatrix}0.5&0&0.1\\0&0.5&0\\0&0&0.8\end{bmatrix},
\]
and the camera matrix is
\[
	M =
	\begin{bmatrix}
		-0.1406  &  0.0537 &  -0.0200 &   0.8452 \\
		0.0321  &  0.0547 &  -0.1385 &  0.4893 \\
		-0.0000 &  -0.0000 &  -0.0000  &  0.0009
	\end{bmatrix}.
\]

\begin{table}[t]
	\centering
	\caption{Channel model parameter examples}
	\label{tab:cm_para}
	\rowcolors{2}{gray!10}{}
	\begin{tabular}{lll}
		\toprule
		Description & Symbol           & Value \\
		\midrule
		Throwing ratio & $r_\text{throw}$ & 20 \\
		Physical size of ghosts & $S_f$ & $0.0156~\text{cm}^2$ \\
		Projection resolution & $P_O$ & $1024\times 768$ \\
		Flare booster & $\rho$ & 30 \\
		Bulb intensity & $T_a$ & $[0, 1]$ \\
		Ambient illuminance & $\I_{env}$ & \SI{300}{\lux} (indoor) \\
		Projector ill. & $\I$ & \SI{400}{\lux} (at 1 m) \\
		Projector max ill. & $\I_\text{max}$ & \SI{1200}{\lux} (at 1 m) \\
		Camera matrix & $M$          & See below       \\
		Color calibration matrix & $H_c$ & See below \\
		In Equation~\ref{eq:ill}& $a$ & 8.9 \\
		In Equation~\ref{eq:ill}& $b$ & 6.7 \\
		In Equation~\ref{eq:ill}& $c_t$ & -7.8 \\
		In Equation~\ref{eq:ill}& $c_d$ & 0.25 \\
		\bottomrule
	\end{tabular}
\end{table}

\section{Direct projection using an additional lens}
\label{sec:dir}

\begin{lemma}

	However the additional lenses are placed, the image of noise cannot overlay
	with the image of the object without obscuring the object.

\end{lemma}

\begin{proof}
We are going to prove that even with an additional concave lens. See
Figure~\ref{fig:concave} for a diagram, where a concave lens $L_1$ is placed
between an object $AB$ and a camera's convex lens $L_2$. $L_1$'s focal length
is $f_1$ and $L_2$'s is $f_2$.  The distance between $L_1$ and $L_2$ is $d_2$.
The noise source $N$ is right upon $A$. From the perspective of $L_2$, $AB$ is
completely obfuscated by $L_1$; in other words, all light rays of $AB$ that go
through $L_2$ must at first go through $L_1$. Both the object and the noise
source share the same object distance to $L_1$, which is $d_1$. The image of
$AB$ formed by $L_1$ is $A_1B_1$, and $N_1$ is the image of $N$. The distance
between $A_1B_1$ and $L_1$ is $d_3$. The image of formed by $L_2$ is $A_2B_2$
and $N_2$. The distance between $A_2B_2$ and $L_2$ is $d_4$.

In order to solve a problem consisting of multiple lenses, we usually analyze
each lens individually and sequentially. That is, we calculate the image formed
by the first lens, then use that image as the input to the second lens.

For $L_1$, the input is $AB$, thus we have
\begin{equation}
	\frac{1}{d_1} + \frac{1}{d_3} = \frac{1}{f_1},
\end{equation}
and the magnification is calculated as
\begin{equation}
	M = \frac{|N_1B_1|}{|NB|}
	= \frac{|A_1B_1|}{|AB|}
	= \frac{|A_1N_1|}{|AN|}
	= \frac{d_3}{d_1},
\end{equation}
from which we know that $N_1$ does not overlap with $A_1B_1$.

For $L_2$, the input is $A_1B_1$ ($L_2$ cannot ``see'' $AB$ directly because
$AB$ is completely obfuscated by $L_1$), thus we get
\begin{equation}
	\frac{1}{d_2+d_3} + \frac{1}{d_4} = \frac{1}{f_2},
\end{equation}
and the magnification is calculated as
\begin{equation}
	M = \frac{|N_2B_2|}{|N_1B_1|}
	= \frac{|A_2B_2|}{|A_1B_1|}
	= \frac{|A_2N_2|}{|A_1N_1|}
	= \frac{d_4}{d_2+d_3},
\end{equation}
from which we know that $N_2$ does not overlap with $A_2B_2$.
As a result, no matter how we place the additional concave lens,
we cannot apply the noise to the image of the object without obscuring the object.
The proof for a convex lens follows the same logic.

\end{proof}

\begin{figure}[t]
	\centering
	\resizebox{.9\columnwidth}{!}{\input{./figs/concave.tikz}}
	\caption{Noise $N$ cannot overlap with the image of the object $AB$ even
	with an additional concave lens.}
	\label{fig:concave}
\end{figure}
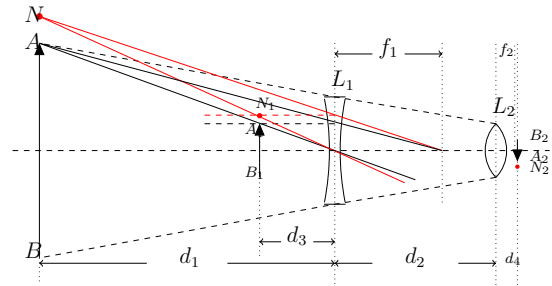

\section{Beamsplitting Method}
\label{sec:beamsplit}

See Figure~\ref{fig:beamsplitter} for a diagram of this method, where a
beamsplitter is used to merge two light beams coming from two directions. The
light beams coming from the object (marked in red) get reflected and
transmitted, i.e., the beamsplitter does not obscure the object. The
transmitted portions go into the camera, forming an image of the object. The
light beams from the light noise source (marked in blue) also get reflected and
transmitted. The reflected portions (instead of the transmitted portions) are
captured by the camera, forming an image of the noise. Two images overlap as a
potential adversarial example, depicted as a small magenta (red plus blue) stop
sign in the camera.

An in-lab experimental setup is shown in Figure~\ref{fig:glass_real_setup}. In
this setup, we used the NEC projector~\cite{nec-np1150} as the light noise
source. Instead of using an expensive beamsplitter, we found out that a single
piece of glass could also function like a beamsplitter. We placed a piece of
white paper in front of the projector's lens to reduce the projected image size
(otherwise the projected image becomes too large even within a small throwing
distance). This does not change the attack plausibility because the imager can
clearly capture the noise pattern on the paper. These elements were placed in a
way that the noise image (from the paper) would overlap with the image of the
object from the view of the camera. A misclassification matrix is shown in
Figure~\ref{fig:beamsplitter_matrix} where an overall attack success rate of 55\% was
achieved.

\begin{figure}
	\centering
	\begin{subfigure}[t]{.9\columnwidth}
		\centering
		\resizebox{.8\columnwidth}{!}{\input{./figs/beamsplitter.tikz}}
		\caption{Noise injection with a beam-splitter/glass}
		\vspace{5pt}
		\label{fig:beamsplitter}
	\end{subfigure}
	\begin{subfigure}[t]{.9\columnwidth}
		\centering
		\includegraphics[width=.8\columnwidth]{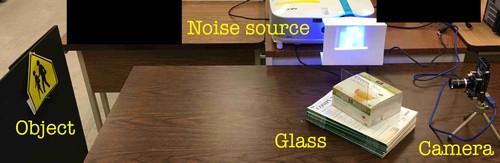}
		\caption{Experiment setup with a piece of glass.}
		\label{fig:glass_real_setup}
	\end{subfigure}
	\caption{The beamsplitter method setup.}
	\label{fig:glass_setup}
\end{figure}

\begin{figure}
	\centering
	\includegraphics[width=.8\columnwidth]{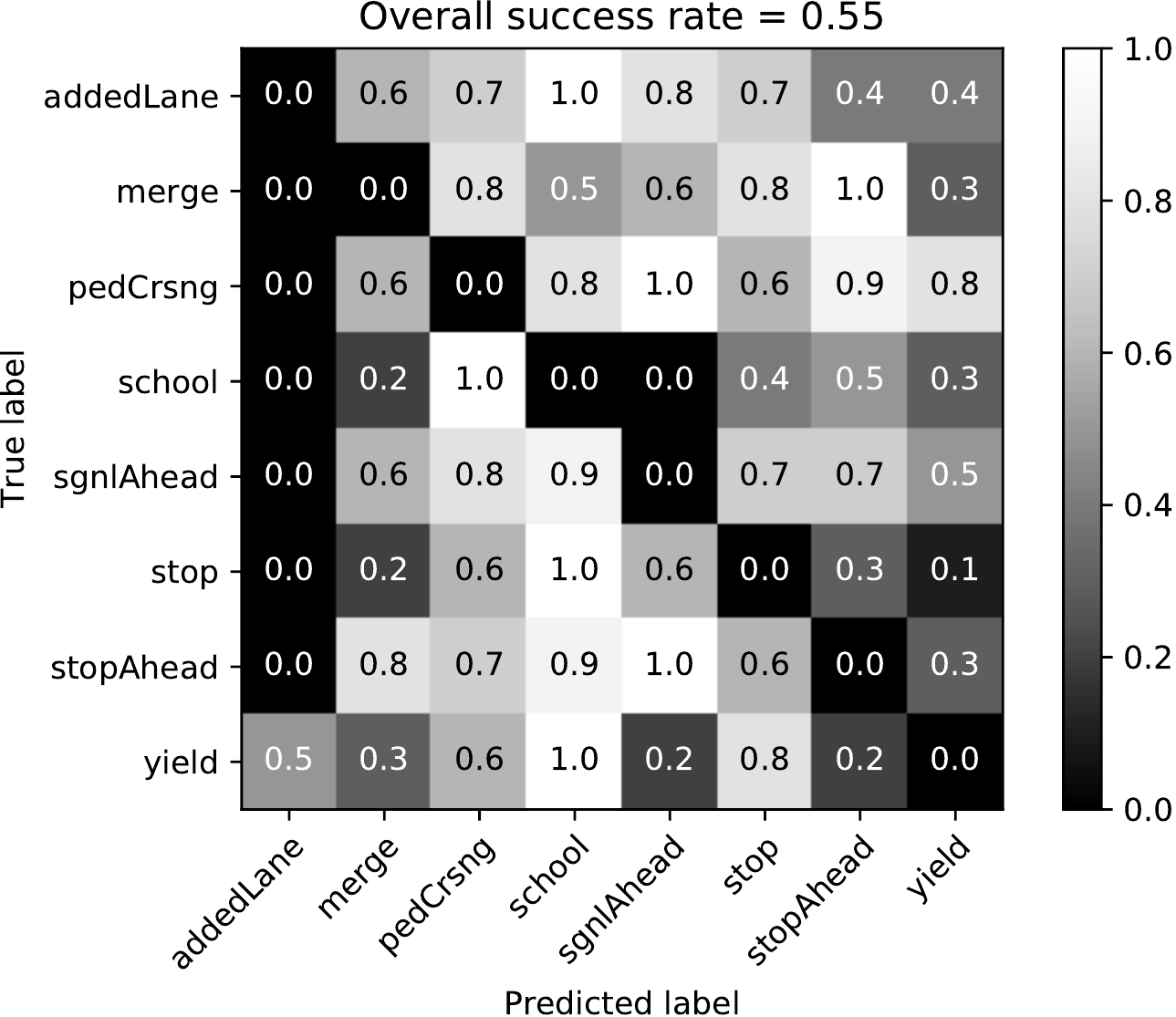}
	\caption{Misclassification matrix of the beamsplitting method at $4\times
	4$ resolutions}
	\label{fig:beamsplitter_matrix}
\end{figure}

%% file: figs/concave.tikz
\begin{tikzpicture}[y=-1cm]
\draw[black] (11.475,10) +(-44:0.725) arc (-44:44:0.725);
\draw[black] (12.525,10) +(-136:0.725) arc (-136:-224:0.725);
\draw[black] (14.15,10) +(-169:5.05) arc (-169:-191:5.05);
\draw[black] (3.85,10) +(-11:5.05) arc (-11:11:5.05);
\filldraw[red] (3.5,7.5) circle (0.05111cm);
\filldraw[red] (7.6,9.35111) circle (0.04cm);
\filldraw[red] (12.4,10.3) circle (0.03111cm);
\draw[black] (8.8,9) -- (9.2,9);
\draw[black] (8.8,11) -- (9.2,11);
\draw[dashed,black] (12,9.5) -- (9,9) -- (3.5,8);
\draw[dashed,black] (12,10.5) -- (9,11) -- (3.5,12);
\draw[arrows=triangle 45-,black] (3.5,8) -- (3.5,12);
\draw[black] (3.5,8) -- (9,10) -- (10.5,10.54444);
\draw[dashed,black] (9,9.5) -- (6.5,9.5);
\draw[arrows=triangle 45-,black] (7.6,9.5) -- (7.6,10.5);
\draw[dotted,black] (11,8) -- (11,11);
\draw[arrows=-to,black] (9.7,8.2) -- (9,8.2);
\draw[red] (3.5,7.5) -- (10.3,10.6);
\draw[red] (3.5,7.5) -- (11,10);
\draw[black] (3.5,8) -- (11,10);
\draw[dashed,red] (9,9.34) -- (6.5,9.34);
\draw[arrows=-triangle 45,black] (12.4,9.8) -- (12.4,10.2);
\draw[dotted,black] (9,8) -- (9,12.5);
\draw[arrows=-to,black] (10.3,8.2) -- (11,8.2);
\draw[dotted,black] (3.5,12) -- (3.5,12.5);
\draw[arrows=-to,black] (7,12.1) -- (9,12.1);
\draw[arrows=-to,black] (5.5,12.1) -- (3.5,12.1);
\draw[dotted,black] (12.35111,10) -- (12.35111,8);
\draw[dotted,black] (12,12.5) -- (12,8);
\draw[dotted,black] (12.4,12.5) -- (12.4,8);
\draw[dotted,black] (7.6,9.5) -- (7.6,12);
\draw[arrows=-to,black] (9.7,12.1) -- (9,12.1);
\draw[arrows=-to,black] (11.3,12.1) -- (12,12.1);
\draw[arrows=-to,black] (8,11.7) -- (7.6,11.7);
\draw[arrows=-to,black] (8.6,11.7) -- (9,11.7);
\draw[dashed,black] (3,10) -- (13,10);
\path (12.05111,12.1) node[text=black,anchor=base west] {\footnotesize{}$d_4$};
\path (3.1,7.6) node[text=black,anchor=base west] {\large{}$N$};
\path (3.1,8.1) node[text=black,anchor=base west] {\large{}$A$};
\path (7.2,9.7) node[text=black,anchor=base west] {\footnotesize{}$A_1$};
\path (7.4,9.2) node[text=black,anchor=base west] {\footnotesize{}$N_1$};
\path (12.5,10.2) node[text=black,anchor=base west] {\footnotesize{}$A_2$};
\path (12.5,10.4) node[text=black,anchor=base west] {\footnotesize{}$N_2$};
\path (3.1,12) node[text=black,anchor=base west] {\large{}$B$};
\path (7.2,10.5) node[text=black,anchor=base west] {\footnotesize{}$B_1$};
\path (12.5,9.8) node[text=black,anchor=base west] {\footnotesize{}$B_2$};
\path (8.8,8.8) node[text=black,anchor=base west] {\large{}$L_1$};
\path (11.8,9.3) node[text=black,anchor=base west] {\large{}$L_2$};
\path (10,8.2) node[text=black,anchor=base] {\large{}$f_1$};
\path (10.5,12.1) node[text=black,anchor=base] {\large{}$d_2$};
\path (8.3,11.7) node[text=black,anchor=base] {\large{}$d_3$};
\path (6.3,12.1) node[text=black,anchor=base] {\large{}$d_1$};
\path (12.2,8.2) node[text=black,anchor=base] {\footnotesize{}$f_2$};

\end{tikzpicture}%

%% file: figs/beamsplitter.tikz
\begin{tikzpicture}[y=-1cm]
\draw[black] (13.9,11.7) -- (13.9,12.8) -- (12.4,12.8) -- (12.4,13.4) -- (13.9,13.4) -- (13.9,14.4) -- (14.6,14.4) -- (14.6,11.7) -- cycle;
\draw[black] (14.2,12.2) rectangle (14.3,12.4);
\draw[arrows=-to,red] (7.5,13.2) -- (11.9,13.2);
\draw[very thick,black] (8.73778,14.03556) -- (10.63778,12.03556);
\draw[black] (9.5,12) -- cycle;
\draw[dashed,arrows=-to,blue] (9.7,13) -- (9.7,14.4);
\draw[dashed,arrows=-to,red] (9.5,13.2) -- (9.5,12);
\draw[arrows=-to,blue] (9.7,11.5) -- (9.7,13) -- (11.9,13);
\draw[black] (9.2,11) -- (9.2,11.3) -- (9.5,11.3) -- (9.5,11.5) -- (9.9,11.5) -- (9.9,11.3) -- (10.2,11.3) -- (10.2,11) -- cycle;
\path[draw=blue,thick,fill=magenta] (14.36444,13.4) -- (14.50889,13.15111) -- (14.36667,12.90222) -- (14.07778,12.9) -- (13.93333,13.14889) -- (14.07556,13.39778) -- cycle;
\path[draw=black,thick,fill=red] (7,13.8) -- (7.40667,13.10444) -- (7.00667,12.40444) -- (6.2,12.4) -- (5.79333,13.09556) -- (6.19333,13.79556) -- cycle;
\path (10.3,11.3) node[text=black,anchor=base west] {\large{}Noise source};
\path (10.8,12.3) node[text=black,anchor=base west] {\large{}Beamsplitter};
\path (12.3,14.1) node[text=black,anchor=base west] {\large{}Camera};
\path (6,12) node[text=black,anchor=base west] {\large{}Object};
\path (14.22222,13.22222) node[text=white,anchor=base] {\scriptsize{}STOP};
\path (6.6,13.3) node[text=white,anchor=base] {\Large STOP};

\end{tikzpicture}%


%% file: paper.bbl
\begin{thebibliography}{100}

\bibitem{uber}
Uber.
\newblock Self-driving car technology by uber.
\newblock \url{https://www.uber.com/us/en/atg/technology/}, 2020.

\bibitem{waymo}
Waymo.
\newblock Waymo.
\newblock \url{https://waymo.com}, 2020.

\bibitem{autopilot}
Tesla.
\newblock Autopilot.
\newblock \url{https://www.tesla.com/autopilot}, 2020.

\bibitem{skydio}
Skydio.
\newblock Skydio 2.
\newblock \url{skydio.com}, 2020.

\bibitem{primeair}
Amazon.
\newblock Prime air delivery.
\newblock
  \url{https://www.amazon.com/Amazon-Prime-Air/b?ie=UTF8\&node=8037720011},
  2020.

\bibitem{nest-systems}
Google.
\newblock Nest and google home. now under one roof.
\newblock \url{nest.com}, 2020.

\bibitem{ring-systems}
Amazon.
\newblock Ring.
\newblock \url{ring.com}, 2020.

\bibitem{checkoway2011comprehensive}
Stephen Checkoway, Damon McCoy, Brian Kantor, Danny Anderson, Hovav Shacham,
  Stefan Savage, Karl Koscher, Alexei Czeskis, Franziska Roesner, Tadayoshi
  Kohno, et~al.
\newblock Comprehensive experimental analyses of automotive attack surfaces.
\newblock In {\em USENIX Security Symposium}, volume~4, pages 447--462. San
  Francisco, 2011.

\bibitem{miller2014survey}
Charlie Miller and Chris Valasek.
\newblock A survey of remote automotive attack surfaces.
\newblock {\em black hat USA}, 2014:94, 2014.

\bibitem{petit2014potential}
Jonathan Petit and Steven~E Shladover.
\newblock Potential cyberattacks on automated vehicles.
\newblock {\em IEEE Transactions on Intelligent Transportation Systems},
  16(2):546--556, 2014.

\bibitem{costin2014large}
Andrei Costin, Jonas Zaddach, Aur{\'e}lien Francillon, and Davide Balzarotti.
\newblock A large-scale analysis of the security of embedded firmwares.
\newblock In {\em 23rd $\{$USENIX$\}$ Security Symposium ($\{$USENIX$\}$
  Security 14)}, pages 95--110, 2014.

\bibitem{selvaraj2018electromagnetic}
Jayaprakash Selvaraj, G{\"o}k{\c{c}}en~Y Dayan{\i}kl{\i}, Neelam~Prabhu
  Gaunkar, David Ware, Ryan~M Gerdes, Mani Mina, et~al.
\newblock Electromagnetic induction attacks against embedded systems.
\newblock In {\em Proceedings of the 2018 on Asia Conference on Computer and
  Communications Security}, pages 499--510. ACM, 2018.

\bibitem{sugawaralight}
Takeshi Sugawara, Benjamin Cyr, Sara Rampazzi, Daniel Genkin, and Kevin Fu.
\newblock Light commands: Laser-based audio injection attacks on
  voice-controllable systems.

\bibitem{shin2017illusion}
Hocheol Shin, Dohyun Kim, Yujin Kwon, and Yongdae Kim.
\newblock Illusion and dazzle: Adversarial optical channel exploits against
  lidars for automotive applications.
\newblock In {\em International Conference on Cryptographic Hardware and
  Embedded Systems}, pages 445--467. Springer, 2017.

\bibitem{petit2015remote}
Jonathan Petit, Bas Stottelaar, Michael Feiri, and Frank Kargl.
\newblock Remote attacks on automated vehicles sensors: Experiments on camera
  and lidar.
\newblock {\em Black Hat Europe}, 11:2015, 2015.

\bibitem{cao2019adversarial}
Yulong Cao, Chaowei Xiao, Benjamin Cyr, Yimeng Zhou, Won Park, Sara Rampazzi,
  Qi~Alfred Chen, Kevin Fu, and Z~Morley Mao.
\newblock Adversarial sensor attack on lidar-based perception in autonomous
  driving.
\newblock In {\em Proceedings of the 2019 ACM SIGSAC Conference on Computer and
  Communications Security}, pages 2267--2281. ACM, 2019.

\bibitem{son2015rocking}
Yunmok Son, Hocheol Shin, Dongkwan Kim, Youngseok Park, Juhwan Noh, Kibum Choi,
  Jungwoo Choi, and Yongdae Kim.
\newblock Rocking drones with intentional sound noise on gyroscopic sensors.
\newblock In {\em 24th $\{$USENIX$\}$ Security Symposium ($\{$USENIX$\}$
  Security 15)}, pages 881--896, 2015.

\bibitem{yan2020surfingattack}
Qiben Yan, Kehai Liu, Qin Zhou, Hanqing Guo, and Ning Zhang.
\newblock Surfingattack: Interactive hidden attack on voice assistants using
  ultrasonic guided wave.
\newblock In {\em Network and Distributed Systems Security (NDSS) Symposium},
  2020.

\bibitem{yan2020minimalist}
C.~Yan, H.~Shin, C.~Bolton, W.~Xu, Y.~Kim, and K.~Fu.
\newblock Sok: A minimalist approach to formalizing analog sensor security.
\newblock In {\em 2020 IEEE Symposium on Security and Privacy (SP)}, pages
  480--495, Los Alamitos, CA, USA, may 2020. IEEE Computer Society.

\bibitem{giechaskiel2020taxonomy}
Ilias Giechaskiel and Kasper~Bonne Rasmussen.
\newblock Taxonomy and challenges of out-of-band signal injection attacks and
  defenses.
\newblock {\em IEEE Communication Surveys \& Tutorials}, 2020.

\bibitem{redmon2016you}
Joseph Redmon, Santosh Divvala, Ross Girshick, and Ali Farhadi.
\newblock You only look once: Unified, real-time object detection.
\newblock In {\em Proceedings of the IEEE conference on computer vision and
  pattern recognition}, pages 779--788, 2016.

\bibitem{truong2005preventing}
Khai~N Truong, Shwetak~N Patel, Jay~W Summet, and Gregory~D Abowd.
\newblock Preventing camera recording by designing a capture-resistant
  environment.
\newblock In {\em International conference on ubiquitous computing}, pages
  73--86. Springer, 2005.

\bibitem{yan2016can}
Chen Yan, Wenyuan Xu, and Jianhao Liu.
\newblock Can you trust autonomous vehicles: Contactless attacks against
  sensors of self-driving vehicle.
\newblock {\em DEF CON}, 24, 2016.

\bibitem{ribnick2006real}
Evan Ribnick, Stefan Atev, Osama Masoud, Nikolaos Papanikolopoulos, and Richard
  Voyles.
\newblock Real-time detection of camera tampering.
\newblock In {\em 2006 IEEE International Conference on Video and Signal Based
  Surveillance}, pages 10--10. IEEE, 2006.

\bibitem{li2013sensor}
Qingquan Li, Long Chen, Ming Li, Shih-Lung Shaw, and Andreas N{\"u}chter.
\newblock A sensor-fusion drivable-region and lane-detection system for
  autonomous vehicle navigation in challenging road scenarios.
\newblock {\em IEEE Transactions on Vehicular Technology}, 63(2):540--555,
  2013.

\bibitem{fotographee}
Yaopey.
\newblock How to deal with lens flare.
\newblock \url{www.fotographee.com/how-to-deal-with-lens-flare}, 2019.

\bibitem{hullin2011flare}
Matthias~B. Hullin, Elmar Eisemann, Hans-Peter Seidel, and Sungkil Lee.
\newblock Physically-based real-time lens flare rendering.
\newblock {\em ACM Trans. Graph. (Proc. SIGGRAPH 2011)}, 30(4):108:1--108:9,
  2011.

\bibitem{szegedy2014intriguing}
Christian Szegedy, Wojciech Zaremba, Ilya Sutskever, Joan Bruna, Dumitru Erhan,
  Ian Goodfellow, and Rob Fergus.
\newblock Intriguing properties of neural networks.
\newblock In {\em ICLR}, 2014.

\bibitem{carlini2017towards}
Nicholas Carlini and David Wagner.
\newblock Towards evaluating the robustness of neural networks.
\newblock In {\em Proceedings of the IEEE Symposium on Security and Privacy},
  pages 39--57. IEEE, 2017.

\bibitem{eykholt2018robust}
Kevin Eykholt, Ivan Evtimov, Earlence Fernandes, Bo~Li, Amir Rahmati, Chaowei
  Xiao, Atul Prakash, Tadayoshi Kohno, and Dawn Song.
\newblock Robust physical-world attacks on deep learning visual classification.
\newblock In {\em Proceedings of the IEEE Conference on Computer Vision and
  Pattern Recognition}, pages 1625--1634, 2018.

\bibitem{sharif2016accessorize}
Mahmood Sharif, Sruti Bhagavatula, Lujo Bauer, and Michael~K Reiter.
\newblock Accessorize to a crime: Real and stealthy attacks on state-of-the-art
  face recognition.
\newblock In {\em Proceedings of the 2016 ACM SIGSAC Conference on Computer and
  Communications Security}, pages 1528--1540. ACM, 2016.

\bibitem{zhao2019seeing}
Yue Zhao, Hong Zhu, Ruigang Liang, Qintao Shen, Shengzhi Zhang, and Kai Chen.
\newblock Seeing isn't believing: Towards more robust adversarial attack
  against real world object detectors.
\newblock In {\em Proceedings of the 2019 ACM SIGSAC Conference on Computer and
  Communications Security}, pages 1989--2004. ACM, 2019.

\bibitem{goodfellow2014explaining}
Ian~J Goodfellow, Jonathon Shlens, and Christian Szegedy.
\newblock Explaining and harnessing adversarial examples.
\newblock {\em arXiv preprint arXiv:1412.6572}, 2014.

\bibitem{li2019adversarial}
Juncheng~B Li, Frank~R Schmidt, and J~Zico Kolter.
\newblock Adversarial camera stickers: A physical camera attack on deep
  learning classifier.
\newblock In {\em Proceedings of the 35th International Conference on Machine
  Learning, {ICML} 2018}, 2019.

\bibitem{chinalasers}
KYLE MIZOKAMI.
\newblock China could blind u.s. satellites with lasers.
\newblock
  \url{https://www.popularmechanics.com/military/weapons/a29307535/china-satellite-laser-blinding/},
  2019.

\bibitem{wiki-flare}
Gunawan Kartapranata.
\newblock Lens flare at borobudur stairs kala arches, 2010.

\bibitem{lee2013practical}
Sungkil Lee and Elmar Eisemann.
\newblock Practical real-time lens-flare rendering.
\newblock In {\em Computer Graphics Forum}, volume~32, pages 1--6. Wiley Online
  Library, 2013.

\bibitem{steinert2011general}
Benjamin Steinert, Holger Dammertz, Johannes Hanika, and Hendrik~PA Lensch.
\newblock General spectral camera lens simulation.
\newblock In {\em Computer Graphics Forum}, volume~30, pages 1643--1654. Wiley
  Online Library, 2011.

\bibitem{vitoria2019automatic}
Patricia Vitoria and Coloma Ballester.
\newblock Automatic flare spot artifact detection and removal in photographs.
\newblock {\em Journal of Mathematical Imaging and Vision}, 61(4):515--533,
  2019.

\bibitem{lee2005introduction}
Hsien-Che Lee.
\newblock {\em Introduction to color imaging science}.
\newblock Cambridge University Press, 2005.

\bibitem{exposurecontrol}
Spencer Cox.
\newblock What is exposure? (a beginner’s guide).
\newblock \url{https://photographylife.com/what-is-exposure}, 2019.

\bibitem{szeliski2010computer}
Richard Szeliski.
\newblock {\em Computer vision: algorithms and applications}.
\newblock Springer Science \& Business Media, 2010.

\bibitem{aptina}
ON semiconductor.
\newblock {\em MT9M034 1/3-Inch CMOS Digital Image Sensor}, 2017.

\bibitem{nec-np1150-manual}
NEC.
\newblock {\em NP Installation Series User's Manual}, 10 2007.

\bibitem{canon}
Canon.
\newblock Telephoto zoom ef-s 55-250mm.
\newblock
  \url{https://www.usa.canon.com/internet/portal/us/home/products/details/lenses/ef/telephoto-zoom/ef-s-55-250mm-f-4-5-6-is-ii},
  2019.

\bibitem{np05zl}
NEC.
\newblock Np05zl, 4.62–7.02:1 zoom lens.
\newblock \url{https://www.necdisplay.com/p/optional-lenses/np05zl}, 2019.

\bibitem{ring}
Ring.
\newblock Indoor security cameras.
\newblock \url{https://shop.ring.com/collections/security-cams\#indoor}, 2019.

\bibitem{mogelmose2012vision}
Andreas Mogelmose, Mohan~Manubhai Trivedi, and Thomas~B Moeslund.
\newblock Vision-based traffic sign detection and analysis for intelligent
  driver assistance systems: Perspectives and survey.
\newblock {\em IEEE Transactions on Intelligent Transportation Systems},
  13(4):1484--1497, 2012.

\bibitem{drmeter}
Dr. Meter.
\newblock Digital illuminance meter.
\newblock \url{https://www.amazon.com/gp/product/B07LF4BT8V}, 2019.

\bibitem{kingma2015adam}
Diederik~P Kingma and Jimmy Ba.
\newblock Adam: A method for stochastic optimization.
\newblock In {\em ICLR}, 2015.

\bibitem{nvidiap100}
Nvidia.
\newblock {\em NVIDIA TESLA P100 GPU ACCELERATOR}, 2016.

\bibitem{papernot2016limitations}
Nicolas Papernot, Patrick McDaniel, Somesh Jha, Matt Fredrikson, Z~Berkay
  Celik, and Ananthram Swami.
\newblock The limitations of deep learning in adversarial settings.
\newblock In {\em 2016 IEEE European Symposium on Security and Privacy
  (EuroS\&P)}, pages 372--387. IEEE, 2016.

\bibitem{krizhevsky2009learning}
Alex Krizhevsky and Geoffrey Hinton.
\newblock Learning multiple layers of features from tiny images.
\newblock Technical report, Citeseer, 2009.

\bibitem{imagenet_cvpr09}
J.~Deng, W.~Dong, R.~Socher, L.-J. Li, K.~Li, and L.~Fei-Fei.
\newblock {ImageNet: A Large-Scale Hierarchical Image Database}.
\newblock In {\em CVPR09}, 2009.

\bibitem{papernot2016distillation}
Nicolas Papernot, Patrick McDaniel, Xi~Wu, Somesh Jha, and Ananthram Swami.
\newblock Distillation as a defense to adversarial perturbations against deep
  neural networks.
\newblock In {\em Proceedings of the IEEE Symposium on Security and Privacy},
  pages 582--597. IEEE, 2016.

\bibitem{szegedy2016rethinking}
Christian Szegedy, Vincent Vanhoucke, Sergey Ioffe, Jon Shlens, and Zbigniew
  Wojna.
\newblock Rethinking the inception architecture for computer vision.
\newblock In {\em Proceedings of the IEEE conference on computer vision and
  pattern recognition}, pages 2818--2826, 2016.

\bibitem{barco}
Barco.
\newblock Xdl-4k75.
\newblock \url{https://www.barco.com/en/product/xdl-4k75}, 2019.

\bibitem{aptinamt9v}
ON semiconductor.
\newblock {\em MT9V034 1/3-Inch Wide-VGA CMOS Digital Image Sensor}, 2017.

\bibitem{motion-detection}
Ring.
\newblock Standard and advanced motion detection systems used in ring devices.
\newblock
  \url{https://support.ring.com/hc/en-us/articles/115005914666-Standard-and-Advanced-Motion-Detection-Systems-Used-in-Ring-Devices},
  2020.

\bibitem{nassi2020phantom}
Ben Nassi, Dudi Nassi, Raz Ben-Netanel, Yisroel Mirsky, Oleg Drokin, and Yuval
  Elovici.
\newblock Phantom of the adas: Phantom attacks on driver-assistance systems.

\bibitem{epsonproj}
Epson.
\newblock Pro l1490u wuxga 3lcd laser projector.
\newblock
  \url{https://epson.com/For-Work/Projectors/Large-Venue/Pro-L1490U-WUXGA-3LCD-Laser-Projector-with-4K-Enhancement-and-Lens/p/V11HA16020},
  2019.

\bibitem{opteka}
Opteka.
\newblock Opteka 650-1300mm telephoto zoom lens.
\newblock
  \url{https://www.amazon.com/Opteka-650-1300mm-1300-2600mm-Telephoto-Digital/dp/B001VDLZIG},
  2019.

\bibitem{iphone8flare}
Apple.
\newblock iphone 8 plus horrible lens flare and reflections.
\newblock \url{https://discussions.apple.com/thread/8118139}, 2017.

\bibitem{iphone11flare}
Apple.
\newblock iphone 11 pro’s irritating lens flare problem.
\newblock
  \url{https://www.reddit.com/r/apple/comments/da6qft/iphone_11_pros_irritating_lens_flare_problem/},
  2019.

\bibitem{blooming}
Wikipedia.
\newblock Bloom (shader effect).
\newblock \url{https://en.wikipedia.org/wiki/Bloom_(shader_effect)}, 2020.

\bibitem{papernot2016transferability}
Nicolas Papernot, Patrick McDaniel, and Ian Goodfellow.
\newblock Transferability in machine learning: from phenomena to black-box
  attacks using adversarial samples.
\newblock {\em arXiv preprint arXiv:1605.07277}, 2016.

\bibitem{papernot2017practical}
Nicolas Papernot, Patrick McDaniel, Ian Goodfellow, Somesh Jha, Z~Berkay Celik,
  and Ananthram Swami.
\newblock Practical black-box attacks against machine learning.
\newblock In {\em Proceedings of the 2017 ACM on Asia Conference on Computer
  and Communications Security}, pages 506--519. ACM, 2017.

\bibitem{bhagoji2018practical}
Arjun~Nitin Bhagoji, Warren He, Bo~Li, and Dawn Song.
\newblock Practical black-box attacks on deep neural networks using efficient
  query mechanisms.
\newblock In {\em European Conference on Computer Vision}, pages 158--174.
  Springer, 2018.

\bibitem{chen2020devil}
Yuxuan Chen, Xuejing Yuan, Jiangshan Zhang, Yue Zhao, Shengzhi Zhang, Kai Chen,
  and XiaoFeng Wang.
\newblock Devil’s whisper: A general approach for physical adversarial
  attacks against commercial black-box speech recognition devices.
\newblock In {\em 29th USENIX Security Symposium (USENIX Security 20)}, 2020.

\bibitem{song2018physical}
Dawn Song, Kevin Eykholt, Ivan Evtimov, Earlence Fernandes, Bo~Li, Amir
  Rahmati, Florian Tramer, Atul Prakash, and Tadayoshi Kohno.
\newblock Physical adversarial examples for object detectors.
\newblock In {\em 12th $\{$USENIX$\}$ Workshop on Offensive Technologies
  ($\{$WOOT$\}$ 18)}, 2018.

\bibitem{rgblasers}
Laser World.
\newblock Rgb laser.
\newblock
  \url{https://www.laserworld.com/en/show-laser-light-faq/glossary-definitions/88-r/2549-rgb-laser.html}.

\bibitem{rgbflashlights}
Christian Carlberg.
\newblock Hexbright, an open source light.
\newblock
  \url{https://www.kickstarter.com/projects/christian-carlberg/hexbright-an-open-source-light}.

\bibitem{christmastree}
\url{https://linx.li/wl9efae2.jpg}.

\bibitem{lenshood}
Anna Altez.
\newblock Lens flare: How to reduce or avoid it?
\newblock
  \url{https://www.photopoly.net/lens-flare-how-to-reduce-or-avoid-it/}, 2011.

\bibitem{liquidlenses}
Edmond Optics.
\newblock Introduction to liquid lenses.
\newblock
  \url{https://www.edmundoptics.com/knowledge-center/application-notes/imaging/introduction-to-liquid-lenses/}.

\bibitem{xu2015transcut}
Yichao Xu, Hajime Nagahara, Atsushi Shimada, and Rin-ichiro Taniguchi.
\newblock Transcut: Transparent object segmentation from a light-field image.
\newblock In {\em Proceedings of the IEEE International Conference on Computer
  Vision}, pages 3442--3450, 2015.

\bibitem{madry2018towards}
Aleksander Madry, Aleksandar Makelov, Ludwig Schmidt, Dimitris Tsipras, and
  Adrian Vladu.
\newblock Towards deep learning models resistant to adversarial attacks.
\newblock In {\em International Conference on Learning Representations}, 2018.

\bibitem{lecuyer2019certified}
Mathias Lecuyer, Vaggelis Atlidakis, Roxana Geambasu, Daniel Hsu, and Suman
  Jana.
\newblock Certified robustness to adversarial examples with differential
  privacy.
\newblock In {\em IEEE Symposium on Security and Privacy (S\&P)}, 2019.

\bibitem{wu2019defending}
Tong Wu, Liang Tong, and Yevgeniy Vorobeychik.
\newblock Defending against physically realizable attacks on image
  classification.
\newblock In {\em 8th International Conference on Learning Representations
  (ICLR)}, 2020.

\bibitem{shoukry2013non}
Yasser Shoukry, Paul Martin, Paulo Tabuada, and Mani Srivastava.
\newblock Non-invasive spoofing attacks for anti-lock braking systems.
\newblock In {\em International Workshop on Cryptographic Hardware and Embedded
  Systems}, pages 55--72. Springer, 2013.

\bibitem{zhang2017dolphinattack}
Guoming Zhang, Chen Yan, Xiaoyu Ji, Tianchen Zhang, Taimin Zhang, and Wenyuan
  Xu.
\newblock Dolphinattack: Inaudible voice commands.
\newblock In {\em Proceedings of the 2017 ACM SIGSAC Conference on Computer and
  Communications Security}, CCS ’17, page 103–117, New York, NY, USA, 2017.
  Association for Computing Machinery.

\bibitem{chernikova2019are}
Alesia Chernikova, Alina Oprea, Cristina Nita-Rotaru, and BaekGyu Kim.
\newblock Are self-driving cars secure? evasion attacks against deep neural
  networks for steering angle prediction.
\newblock In {\em IEEE Security and Privacy Workshop on IoT}. IEEE, 2019.

\bibitem{zhou2018invisible}
Zhe Zhou, Di~Tang, Xiaofeng Wang, Weili Han, Xiangyu Liu, and Kehuan Zhang.
\newblock Invisible mask: Practical attacks on face recognition with infrared.
\newblock {\em arXiv preprint arXiv:1803.04683}, 2018.

\bibitem{nguyen2020adversarial}
Luan Nguyen, Sunpreet~S. Arora, Yuhang Wu, and Hao Yang.
\newblock Adversarial light projection attacks on face recognition systems: A
  feasibility study, 2020.

\bibitem{tu2020physically}
James Tu, Mengye Ren, Siva Manivasagam, Ming Liang, Bin Yang, Richard Du, Frank
  Cheng, and Raquel Urtasun.
\newblock Physically realizable adversarial examples for lidar object
  detection, 2020.

\bibitem{moosavi2016deepfool}
Seyed-Mohsen Moosavi-Dezfooli, Alhussein Fawzi, and Pascal Frossard.
\newblock Deepfool: a simple and accurate method to fool deep neural networks.
\newblock In {\em Proceedings of the IEEE Conference on Computer Vision and
  Pattern Recognition}, pages 2574--2582, 2016.

\bibitem{kurakin2018adversarial}
Alexey Kurakin, Ian~J Goodfellow, and Samy Bengio.
\newblock Adversarial examples in the physical world.
\newblock In {\em Artificial Intelligence Safety and Security}, pages 99--112.
  Chapman and Hall/CRC, 2018.

\bibitem{moosavi2017universal}
Seyed-Mohsen Moosavi-Dezfooli, Alhussein Fawzi, Omar Fawzi, and Pascal
  Frossard.
\newblock Universal adversarial perturbations.
\newblock In {\em Proceedings of the IEEE conference on computer vision and
  pattern recognition}, pages 1765--1773, 2017.

\bibitem{xiao2018spatially}
Chaowei Xiao, Jun-Yan Zhu, Bo~Li, Warren He, Mingyan Liu, and Dawn Song.
\newblock Spatially transformed adversarial examples.
\newblock In {\em Proceedings of the IEEE conference on learning
  representations}, 2018.

\bibitem{liu2016delving}
Yanpei Liu, Xinyun Chen, Chang Liu, and Dawn Song.
\newblock Delving into transferable adversarial examples and black-box attacks.
\newblock In {\em Proceedings of the International Conference on Learning
  Representations}, 2016.

\bibitem{brendel2018decision}
Wieland Brendel, Jonas Rauber, and Matthias Bethge.
\newblock Decision-based adversarial attacks: Reliable attacks against
  black-box machine learning models.
\newblock In {\em Proceedings of the IEEE conference on learning
  representations}, 2018.

\bibitem{sitawarin2018rogue}
Chawin Sitawarin, Arjun~Nitin Bhagoji, Arsalan Mosenia, Prateek Mittal, and
  Mung Chiang.
\newblock Rogue signs: Deceiving traffic sign recognition with malicious ads
  and logos.
\newblock In {\em IEEE Security and Privacy Workshop on Deep Learning and
  Security}. IEEE, 2018.

\bibitem{wu2019making}
Zuxuan Wu, Ser-Nam Lim, Larry Davis, and Tom Goldstein.
\newblock Making an invisibility cloak: Real world adversarial attacks on
  object detectors.
\newblock {\em arXiv preprint arXiv:1910.14667}, 2019.

\bibitem{komkov2019advhat}
Stepan Komkov and Aleksandr Petiushko.
\newblock Advhat: Real-world adversarial attack on arcface face id system.
\newblock {\em arXiv preprint arXiv:1908.08705}, 2019.

\bibitem{duan2020adversarial}
Ranjie Duan, Xingjun Ma, Yisen Wang, James Bailey, A.~K. Qin, and Yun Yang.
\newblock Adversarial camouflage: Hiding physical-world attacks with natural
  styles, 2020.

\bibitem{sato2020security}
Takami Sato, Junjie Shen, Ningfei Wang, Yunhan~Jack Jia, Xue Lin, and Qi~Alfred
  Chen.
\newblock Security of deep learning based lane keeping system under
  physical-world adversarial attack, 2020.

\bibitem{song2018constructing}
Yang Song, Rui Shu, Nate Kushman, and Stefano Ermon.
\newblock Constructing unrestricted adversarial examples with generative
  models.
\newblock In {\em Advances in Neural Information Processing Systems}, pages
  8312--8323, 2018.

\bibitem{hosseini2018semantic}
Hossein Hosseini and Radha Poovendran.
\newblock Semantic adversarial examples.
\newblock In {\em Proceedings of the IEEE Conference on Computer Vision and
  Pattern Recognition Workshops}, pages 1614--1619, 2018.

\bibitem{Bhattad2020Unrestricted}
Anand Bhattad, Min~Jin Chong, Kaizhao Liang, Bo~Li, and D.~A. Forsyth.
\newblock Unrestricted adversarial examples via semantic manipulation.
\newblock In {\em International Conference on Learning Representations}, 2020.

\bibitem{dhillon2018stochastic}
Guneet~S. Dhillon, Kamyar Azizzadenesheli, Jeremy~D. Bernstein, Jean Kossaifi,
  Aran Khanna, Zachary~C. Lipton, and Animashree Anandkumar.
\newblock Stochastic activation pruning for robust adversarial defense.
\newblock In {\em International Conference on Learning Representations}, 2018.

\bibitem{xie2018mitigating}
Cihang Xie, Jianyu Wang, Zhishuai Zhang, Zhou Ren, and Alan Yuille.
\newblock Mitigating adversarial effects through randomization.
\newblock In {\em International Conference on Learning Representations}, 2018.

\bibitem{liu2018towards}
Xuanqing Liu, Minhao Cheng, Huan Zhang, and Cho-Jui Hsieh.
\newblock Towards robust neural networks via random self-ensemble.
\newblock In {\em Proceedings of the European Conference on Computer Vision
  (ECCV)}, pages 369--385, 2018.

\bibitem{wong2018provable}
Eric Wong and J~Zico Kolter.
\newblock Provable defenses against adversarial examples via the convex outer
  adversarial polytope.
\newblock In {\em International Conference on Machine Learning}, 2018.

\bibitem{raghunathan2018certified}
Aditi Raghunathan, Jacob Steinhardt, and Percy Liang.
\newblock Certified defenses against adversarial examples.
\newblock In {\em International Conference on Learning Representations}, 2018.

\bibitem{cao2017mitigating}
Xiaoyu Cao and Neil~Zhenqiang Gong.
\newblock Mitigating evasion attacks to deep neural networks via region-based
  classification.
\newblock In {\em Proceedings of the 33rd Annual Computer Security Applications
  Conference}, pages 278--287. ACM, 2017.

\bibitem{xu2017feature}
Weilin Xu, David Evans, and Yanjun Qi.
\newblock Feature squeezing: Detecting adversarial examples in deep neural
  networks.
\newblock In {\em Network and Distributed Systems Security Symposium (NDSS)},
  2018.

\bibitem{hendrycks2017early}
Dan Hendrycks and Kevin Gimpel.
\newblock Early methods for detecting adversarial images.
\newblock 2017.

\bibitem{metzen2017detecting}
Jan~Hendrik Metzen, Tim Genewein, Volker Fischer, and Bastian Bischoff.
\newblock On detecting adversarial perturbations.
\newblock In {\em Proceedings of the IEEE conference on learning
  representations}, 2017.

\bibitem{meng2017magnet}
Dongyu Meng and Hao Chen.
\newblock Magnet: a two-pronged defense against adversarial examples.
\newblock In {\em Proceedings of the 2017 ACM SIGSAC Conference on Computer and
  Communications Security}, pages 135--147. ACM, 2017.

\bibitem{ma2019nic}
Shiqing Ma, Yingqi Liu, Guanhong Tao, Wen-Chuan Lee, and Xiangyu Zhang.
\newblock Nic: Detecting adversarial samples with neural network invariant
  checking.
\newblock In {\em Network and Distributed System Security Symposium}, 2019.

\bibitem{cai2020detecting}
Feiyang Cai, Jiani Li, and Xenofon Koutsoukos.
\newblock Detecting adversarial examples in learning-enabled cyber-physical
  systems using variational autoencoder for regression.
\newblock In {\em Workshop on Assured Autonomous Systems}, 2020.

\bibitem{buckman2018thermometer}
Jacob Buckman, Aurko Roy, Colin Raffel, and Ian Goodfellow.
\newblock Thermometer encoding: One hot way to resist adversarial examples.
\newblock In {\em International Conference on Learning Representations}, 2018.

\bibitem{guo2018countering}
Chuan Guo, Mayank Rana, Moustapha Cisse, and Laurens van~der Maaten.
\newblock Countering adversarial images using input transformations.
\newblock In {\em International Conference on Learning Representations}, 2018.

\bibitem{nec-np1150}
NEC.
\newblock {\em NP Installation Series Specification Sheet}, 11 2009.

\end{thebibliography}
